\documentclass[12pt]{article}
\usepackage{amssymb}
\usepackage{amsfonts}
\usepackage{amsmath}
\usepackage{eurosym}
\usepackage{makeidx}
\usepackage{amssymb}
\usepackage{times}
\usepackage{anysize}
\usepackage{cite}
\usepackage{lineno}
\usepackage{units}

\marginsize{2cm}{2cm}{1cm}{1cm}
\newtheorem{theorem}{Theorem}[section]

\newtheorem{corollary}[theorem]{Corollary}

\newtheorem{definition}[theorem]{Definition}

\newtheorem{lemma}[theorem]{Lemma}
\newtheorem{notation}[theorem]{Notation}

\newtheorem{proposition}[theorem]{Proposition}
\newtheorem{remark}[theorem]{Remark}

\newenvironment{proof}[1][Proof]{\noindent\textbf{#1.} }{\ \rule{0.5em}{0.5em}}

\input{tcilatex}

\begin{document}

\title{Accuracy of Classical Conductivity Theory at Atomic Scales for Free
Fermions in Disordered Media}
\author{N. J. B. Aza \and J.-B. Bru \and W. de Siqueira Pedra \and A.
Ratsimanetrimanana}
\date{\today }
\maketitle

\begin{abstract}
The growing need for smaller electronic components has recently sparked the
interest in the breakdown of the classical conductivity theory near the
atomic scale, at which quantum effects should dominate. In 2012,
experimental measurements of electric resistance of nanowires in Si doped
with phosphorus atoms demonstrate that quantum effects on charge transport
almost disappear for nanowires of lengths larger than a few nanometers, even
at very low temperature ($4.2~\mathrm{K}$). We mathematically prove, for
non-interacting lattice fermions with disorder, that quantum uncertainty of
microscopic electric current density around their (classical) macroscopic
values is suppressed, exponentially fast with respect to the volume of the
region of the lattice where an external electric field is applied. This is
in accordance with the above experimental observation. Disorder is modeled
by a random external potential along with random, complex-valued, hopping
amplitudes. The celebrated tight-binding Anderson model is one particular
example of the general case considered here. Our mathematical analysis is
based on Combes-Thomas estimates, the Akcoglu-Krengel ergodic theorem, and
the large deviation formalism, in particular the G\"{a}rtner-Ellis
theorem.
\end{abstract}

\bigskip


\tableofcontents%

\section{Introduction}

The classical conductivity theory of materials, based on the existence of a
well-defined bulk resistivity, was expected to break down as atomic scales
and low temperatures are reached, because quantum effects would dominate. In
particular, the linear dependence of the resistance as a function of the
length of conducting nanowires should be violated at atomic lengths, as
explained in \cite{Ohm-exp2}.

The growing need for smaller electronic components has recently sparked the
interest in such a question. For instance, in 2006, the validity of the
classical theory was experimentally verified, at room temperature, for
nanowires in InAs with lengths down to $\sim 200~\mathrm{nm}$ \cite{ZDAW}.
Indeed, the measured resistivity for the nanowires is $23~\mathrm{\Omega /nm}
$, which is very near to the resistivity deduced from bulk properties of the
material ($24~\mathrm{\Omega /nm}$). See \cite[discussions after Eq. (2)]%
{ZDAW}. A few years later, in 2012, the same property was observed \cite%
{Ohm-exp}, even at very low temperature ($4.2~\mathrm{K}$) and lengths down
to $20~\mathrm{nm}$ (atomic scale), in experiments on nanowires in Si doped
with phosphorus atoms. The breakdown of the classical description of these
nanowires is expected \cite{Ohm-exp2} to be around $\sim 10~\mathrm{nm}$ (at
similar temperature) since other experimental studies \cite{16exp,15exp} on
similar doped Si wires show strong deviations from bulk values of the
resistivity around this length scale.

These experimental results demonstrate that quantum effects on charge
transport can very rapidly disappear with respect to (w.r.t.) growing
space-scales. We mathematically prove this fact by studying the suppression
rate of the probability of finding microscopic current densities that differ
from the macroscopic one. Observe that \cite{OhmIII,OhmVI} already proved
the convergence of the expectation values of microscopic current densities,
but no information about the suppression of quantum uncertainty was obtained
in the macroscopic limit.

There is a large mathematical literature on charged transport properties of
fermions in disordered media, for instance by Bellissard and Schulz-Baldes
in the nineties \cite{[SBB98],BvESB94} or, more recently, by Klein, M\"{u}%
ller and coauthors \cite{jfa,klm,JMP-autre,JMP-autre2,germinet}. See also
\cite{Pro13,Cornean} and references therein, etc. However, it is not the
purpose of this introduction to go into the details of the history of this
specific research field. For a (non-exhaustive) historical perspective on
linear conductivity (Ohm's law), see, e.g., \cite{brupedrahistoire} or our
previous papers \cite{OhmI,OhmII,OhmIII,OhmIV,OhmV,OhmVI,brupedraLR}.

In spite of that large mathematical literature on quantum charged transport,
the study performed in the current paper covers a completely new theoretical
aspect of this problem, not exploited in the available literature, yet.
Observe that although we were able in \cite{OhmVI} to deal with interacting
fermions, in the present paper we restrict ourselves to the non-interacting
case, similar to \cite{OhmIII}. Within the class of non-interacting
particles the considered Hamiltonians are however completely general, since
disorder is defined via random potentials and random, complex valued,
hopping amplitudes, which are only assumed to have ergodic distributions.
The celebrated tight-binding Anderson model is one particular example of the
general case analyzed here and models with random vector potentials are also
included within the present study.

We prove that quantum uncertainty of microscopic electric current densities
(around their classical, macroscopic values) is suppressed, \emph{%
exponentially fast} w.r.t. the volume $\left\vert \Lambda _{L}\right\vert =%
\mathcal{O}(L^{d})$ (in lattice units (l.u.), $d\in \mathbb{N}$ being the
space dimension) of the region of the lattice where an external electric
field is applied. In order to achieve this, we use the large deviation
formalism \cite{DS89,dembo1998large}, which has been adopted in quantum
statistical mechanics since the eighties \cite[Section 7]{ABPM1}. Other
mathematical results which are pivotal in our analysis are the Combes-Thomas
estimates \cite{CT73,AizenmanWarzel}, the Akcoglu-Krengel ergodic theorem
\cite{birkoff} and the (Arzel\`{a}-) Ascoli theorem \cite[Theorem A5]{Rudin}%
. Indeed, combined with the celebrated G\"{a}rtner-Ellis theorem (Theorem %
\ref{prop Gartner--Ellis}), they allow us to prove a large deviation
principle (LDP) for the current density distributions, which quantify the
probability of deviations, due to quantum uncertainty, from the expected
value.

The interacting case, as studied in \cite{OhmV,OhmVI}, is technically much
more involved. The mathematical techniques allowing to tackle such questions
for interacting fermions are partially developed in \cite{ABPM1,ABPM2}, and
use Grassmann integrals and Brydges-Kennedy tree expansions to construct G%
\"{a}rtner-Ellis generating functions. For the non-interacting case, in
order to study properties of G\"{a}rtner-Ellis generating functions, one can
use the Bogoliubov-type inequality
\begin{equation*}
\left\vert \ln \mathrm{tr}\left( C\mathrm{e}^{H_{1}}\right) -\ln \mathrm{tr}%
\left( C\mathrm{e}^{H_{0}}\right) \right\vert \leq \sup_{\alpha \in \left[
0,1\right] }\sup_{u\in \left[ -1/2,1/2\right] }\left\Vert \mathrm{e}%
^{u\left( \alpha H_{1}+\left( 1-\alpha \right) H_{0}\right) }\left(
H_{1}-H_{0}\right) \mathrm{e}^{-u\left( \alpha H_{1}+\left( 1-\alpha \right)
H_{0}\right) }\right\Vert _{\mathcal{B}\left( \mathbb{C}^{n}\right) },
\end{equation*}%
where $H_{0},H_{1}$ are arbitrary self-adjoint matrices, $C$ is any positive
matrix and $\mathrm{tr}$ denotes the normalized trace. See \cite[Lemma 3.6]%
{lenci2005large} or Lemma \ref{lemma:suppor_non_int1} below. The above bound
turns out to be useful for fermionic systems that are quasi-free (i.e. $%
H_{0},H_{1}$ are polynomials of degree two in the fermionic creation and
annihilation operators). In this special case, the right-hand side of the
inequality can be efficiently bounded by $\left\Vert H_{1}-H_{0}\right\Vert
_{\mathcal{B}\left( \mathbb{C}^{n}\right) }$, using Combes-Thomas estimates.
In contrast, for interacting fermions, explicit examples for which the
right-hand side is arbitrarily bigger than $\left\Vert
H_{1}-H_{0}\right\Vert _{\mathcal{B}\left( \mathbb{C}^{n}\right) }$ at large
volumes are known \cite{Bouch}.

Our main results are Theorems \ref{Theorem main result}, \ref{Theorem main
result copy(1)} and Corollaries \ref{LDP copy(1)}, \ref{LDP}. From the
technical point of view, Theorem \ref{Theorem main result} is the pivotal
statement of the paper, the other assertions, basically the LDP for currents
with a good rate function (Theorem \ref{Theorem main result copy(1)} and
Corollaries \ref{LDP copy(1)}, \ref{LDP}), being all deduced from Theorem %
\ref{Theorem main result} by relatively standard methods of large
deviations. Theorem \ref{Theorem main result} refers to the existence,
continuity and differentiability of the (infinite volume) deterministic
generating function for currents, which appears in the G\"{a}rtner-Ellis
theorem (Theorem \ref{prop Gartner--Ellis}). Besides the Bogoliubov-type
inequality, as discussed above, its proof requires the Akcoglu-Krengel
ergodic theorem \cite{birkoff} as an important argument, for one has to
control the thermodynamical limit of (finite volume) generating functions
that are random. To make possible the use of this important result from
ergodic theory, various technical preliminaries are needed and the proof of
Theorem \ref{prop Gartner--Ellis} is highly non-trivial, as a whole: We
perform a rather complicated box decomposition of these random functions,
which can be justified with the help of the Bogoliubov-type inequality and
the \textquotedblleft locality\textquotedblright\ (or space decay) of both
the quasi-free dynamics and space correlations of KMS states, as a
consequence of Combes-Thomas estimates (Appendix \ref{Section Combes-Thomas
Estimate}).

To conclude, this paper is organized as follows:

\begin{itemize}
\item In Section \ref{sec:SetPro}, the mathematical setting is described in
detail. It refers to quasi-free fermions on the lattice in disordered media.
We also discuss the physical motivations of the model, which are
supplemented by Appendix \ref{linear_response} to reduce the length of this
section.

\item In Section \ref{sec:main}, the main results are stated and the large
deviation (LD) formalism is shortly defined, being supplemented by Appendix %
\ref{sec:LDP}. More precisely, we present the mathematical statements
related to the existence of generating functions of the LD formalism, an LD
principle (LDP) for currents, as well as the behavior of the corresponding
rate function. We finally combine them to state and discuss the
exponentially fast suppression of quantum uncertainty of currents around the
classical value of the current.

\item Section \ref{sec:proofs} gathers all technical proofs. In particular,
Bogoliubov-type inequalities discussed above are stated and proven in
Section \ref{Sectino tech1}. Section \ref{Sect Bilinear Elements} collects
some useful, albeit elementary, properties of bilinear elements, which are
basically quadratic elements in the CAR algebra resulting from the
second-quantization of one-particle operators. Then, in Section \ref{Sectino
tech3}, we show that current observable are bilinear elements associated
with explicit one-particle operators that satisfy several explicit
estimates. These upper bounds are pivotal for the proof of our main theorem,
i.e., Theorem \ref{Theorem main result}, which, effectively, only starts in
Section \ref{Sectino tech4} and is finished in Section \ref{Sectino tech5}
with the use of the Akcoglu-Krengel ergodic theorem \cite{birkoff} and the
(Arzel\`{a}-) Ascoli theorem \cite[Theorem A5]{Rudin}.

\item We finally include Appendices \ref{Section Combes-Thomas Estimate}, %
\ref{sec:LDP} and \ref{linear_response}, stating general results used
throughout the current paper, in a way well-adapted to our proofs. Appendix %
\ref{Section Combes-Thomas Estimate} is about the Combes-Thomas estimates
while Appendix \ref{sec:LDP} explains the large deviation formalism, in
particular\ the G\"{a}rtner-Ellis theorem. Appendix \ref{linear_response}
contains supplementary information on the mathematical framework and
relevant physical concepts, in order to make unnecessary the use of further
references for a clear understanding of the subject of the current paper.
More precisely, Appendix \ref{linear_response} summaries some important
results on linear response current of our papers \cite%
{OhmI,OhmII,OhmIII,OhmIV,OhmV,OhmVI,brupedraLR}. Appendix \ref{One-particle
formulation2} explains the origin of current observables in relation with
the discrete continuity equation within the CAR\ algebra. Finally, Appendix %
\ref{One-particle formulation} makes explicit the link between the algebraic
formulation we use here and the (more popular) one-particle Hilbert space
formulation of non-interacting fermion systems.
\end{itemize}

\begin{notation}
\label{remark constant}\mbox{
}\newline
A norm on a generic vector space $\mathcal{X}$ is denoted by $\Vert \cdot
\Vert _{\mathcal{X}}$. The space of all bounded linear operators on $(%
\mathcal{X},\Vert \cdot \Vert _{\mathcal{X}}\mathcal{)}$ is denoted by $%
\mathcal{B}(\mathcal{X})$. The scalar product of any Hilbert space $\mathcal{%
X}$ is denoted by $\langle \cdot ,\cdot \rangle _{\mathcal{X}}$. Note that $%
\mathbb{R}^{+}\doteq \left\{ x\in \mathbb{R}:x>0\right\} $\ while $\mathbb{R}%
_{0}^{+}\doteq \mathbb{R}^{+}\cup \{0\}$.
\end{notation}

\section{Setup of the Problem\label{sec:SetPro}}

We use the mathematical framework of \cite{OhmVI,brupedraLR} to study
fermions on the lattice. For simplicity we take a cubic lattice $\mathbb{Z}%
^{d}$, even if other types of lattices can certainly be considered with the
same, albeit adapted, methods. Disorder within the conductive material, due
to impurities, crystal lattice defects, etc., is modeled by (a) a random
external potential, like in the celebrated Anderson model, and (b) a random
Laplacian, i.e., a self-adjoint operator defined by a next-nearest neighbor
hopping term with random complex-valued amplitudes. In particular, random
vector potentials can also be implemented.

Altogether, this yields the random tight-binding model mathematically
described in Section \ref{Section impurities}: The underlying probability
space is defined in Part (ii) of that subsection, while the one-particle
Hamiltonian driven the non-interacting (or quasi-free) lattice-fermion
system is explained in Part (iii), see in particular Equation (\ref%
{eq:Ham_lap_pot}). Then, we apply on the quasi-free fermion system in
disordered media some time-dependent electromagnetic fields and look at the
linear response current density in the thermodynamic limit of macroscopic
electromagnetic fields. This study is already done in great generality in
\cite{OhmV,OhmVI,brupedraLR} and we shortly explain it in Section \ref%
{Current Densities}, with complementary explanations postponed to Appendix %
\ref{linear_response}. Then, we will be in a position to state the main
results of the paper about the exponential rate of convergence of current
densities in the limit of macroscopic electromagnetic fields.

Observe that \emph{no} interaction between fermions are considered in the
sequel and one can do all our study on the one-particle Hilbert space, as
illustrated in Appendix \ref{One-particle formulation}. Despite this, our
approach is based on the algebraic formulation of fermion systems on
lattices explained in Section \ref{subsec:algebraic} because it makes the
role played by many-fermion correlations due to the Pauli exclusion
principle, i.e., the antisymmetry of the many-body wave function, more
transparent. For instance, the conductivity is naturally defined from
current-current correlations, that is,\ four-point correlation functions, in
this framework. The algebraic formulation also allows a clear link between
transport properties of fermion systems and the CCR algebra of current
fluctuations \cite{OhmIV}. The latter is related to non-commutative central
limit theorems (see, e.g., \cite{CCR fluctuations}). On top of this, the
approach ensures a continuity with our previous results while making much
clearer its extension to a study of \emph{interacting} fermions for which
the algebraic formulation is very advantageous. This paper can thus be seen
as a preparation to do a similar study for interacting fermions. Such an
analysis has already started with \cite{ABPM1,ABPM2} via (highly technical)
constructive methods used in quantum field theory, which will allow us to
obtain convergent expansion schemes around the quasi-free case for generic
generating functions.

\subsection{Random Tight-Binding Model\label{Section impurities}}

\noindent \underline{(i):} The host material for conducting fermions is
assumed to be a cubic crystal represented by the $d$-dimensional cubic
lattice $\mathbb{Z}^{d}$ ($d\in \mathbb{N}$). Below, $\mathcal{P}_{\text{f}}(%
\mathbb{Z}^{d})\subset 2^{\mathbb{Z}^{d}}$ is the set of all non-empty \emph{%
finite} subsets of $\mathbb{Z}^{d}$. Further,
\begin{equation*}
\mathbb{D}\doteq \{z\in \mathbb{C}\colon \left\vert z\right\vert \leq 1\}%
\text{\quad and\quad }\mathfrak{b}\doteq \left\{ \{x,x^{\prime }\}\subset
\mathbb{Z}^{d}\colon |x-x^{\prime }|=1\right\}
\end{equation*}
is the set of (non-oriented) edges of the cubic lattice $\mathbb{Z}^{d}$.

\noindent \underline{(ii):} Disorder in the crystal is modeled by a random
variable taking values in the measurable space $(\Omega ,\mathfrak{A}%
_{\Omega })$, with distribution $\mathfrak{a}_{\Omega }$:

\begin{itemize}
\item[$\Omega$:] Elements of $\Omega $ are pairs $\omega =\left( \omega
_{1},\omega _{2}\right) \in \Omega $, where $\omega _{1}$ is a function on
lattice sites with values in the interval $[-1,1]$ and $\omega _{2}$ is a
function on edges with values in the complex closed unit disc $\mathbb{D}$.
I.e.,%
\begin{equation*}
\Omega \doteq \lbrack -1,1]^{\mathbb{Z}^{d}}\times \mathbb{D}^{\mathfrak{b}}.
\end{equation*}

\item[$\mathfrak{A}_{\Omega }$:] Let $\Omega _{x}^{(1)}$, $x\in \mathbb{Z}%
^{d}$, be an arbitrary element of the Borel $\sigma $-algebra $\mathfrak{A}%
_{x}^{(1)}$ of the interval $[-1,1]$ w.r.t. the usual metric topology.
Define
\begin{equation*}
\mathfrak{A}_{[-1,1]^{\mathbb{Z}^{d}}}\doteq \bigotimes\limits_{x\in \mathbb{%
Z}^{d}}\mathfrak{A}_{x}^{(1)},
\end{equation*}%
i.e., $\mathfrak{A}_{[-1,1]^{\mathbb{Z}^{d}}}$ is the $\sigma $-algebra
generated by the cylinder sets $\prod\limits_{x\in \mathbb{Z}^{d}}\Omega
_{x}^{(1)}$, where $\Omega _{x}^{(1)}=[-1,1]$ for all but finitely many $%
x\in \mathbb{Z}^{d}$. In the same way, let
\begin{equation*}
\mathfrak{A}_{\mathbb{D}^{\mathfrak{b}}}\doteq \bigotimes\limits_{\mathbf{x}%
\in \mathfrak{b}}\mathfrak{A}_{\mathbf{x}}^{(2)}\ ,
\end{equation*}%
where $\mathfrak{A}_{\mathbf{x}}^{(2)}$, $\mathbf{x}\in \mathfrak{b}$, is
the Borel $\sigma $-algebra of the complex closed unit disc $\mathbb{D}$
w.r.t. the usual metric topology. Then
\begin{equation*}
\mathfrak{A}_{\Omega }\doteq \mathfrak{A}_{[-1,1]^{\mathbb{Z}^{d}}}\otimes
\mathfrak{A}_{\mathbb{D}^{\mathfrak{b}}}\ .
\end{equation*}

\item[$\mathfrak{a}_{\Omega }$:] The distribution $\mathfrak{a}_{\Omega }$
is an arbitrary \emph{ergodic} probability measure on the measurable space $%
(\Omega ,\mathfrak{A}_{\Omega })$. I.e., it is invariant under the action
\begin{equation}
\left( \omega _{1},\omega _{2}\right) \longmapsto \chi _{x}^{(\Omega
)}\left( \omega _{1},\omega _{2}\right) \doteq \left( \chi _{x}^{(\mathbb{Z}%
^{d})}\left( \omega _{1}\right) ,\chi _{x}^{(\mathfrak{b})}\left( \omega
_{2}\right) \right) \ ,\qquad x\in \mathbb{Z}^{d}\ ,
\label{translation omega}
\end{equation}%
of the group $(\mathbb{Z}^{d},+)$ of translations on $\Omega $ and $%
\mathfrak{a}_{\Omega }(\mathcal{X})\in \{0,1\}$ whenever $\mathcal{X}\in
\mathfrak{A}_{\Omega }$ satisfies $\chi _{x}^{(\Omega )}\left( \mathcal{X}%
\right) =\mathcal{X}$ for all $x\in \mathbb{Z}^{d}$. Here, for any $\omega
=\left( \omega _{1},\omega _{2}\right) \in \Omega $, $x\in \mathbb{Z}^{d}$
and $y,y^{\prime }\in \mathbb{Z}^{d}$ with $|y-y^{\prime }|=1$,%
\begin{equation}
\chi _{x}^{(\mathbb{Z}^{d})}\left( \omega _{1}\right) \left( y\right) \doteq
\omega _{1}\left( y+x\right) ,\ \chi _{x}^{(\mathfrak{b})}\left( \omega
_{2}\right) \left( \{y,y^{\prime }\}\right) \doteq \omega _{2}\left(
\{y+x,y^{\prime }+x\}\right) \ .  \label{translation omegabis}
\end{equation}%
As is usual, $\mathbb{E}\left[ \cdot \right] $ denotes the expectation value
associated with $\mathfrak{a}_{\Omega }$.\medskip
\end{itemize}

\noindent \underline{(iii):} The one-particle Hilbert space is $\mathfrak{h}%
\doteq \ell ^{2}(\mathbb{Z}^{d};\mathbb{C})$ with scalar product $\langle
\cdot ,\cdot \rangle _{\mathfrak{h}}$. Its canonical orthonormal basis is
denoted by $\left\{ \mathfrak{e}_{x}\right\} _{x\in \mathbb{Z}^{d}}$, which
is defined by $\mathfrak{e}_{x}(y)\doteq \delta _{x,y}$ for all $x,y\in
\mathbb{Z}^{d}$. ($\delta _{x,y}$ is the Kronecker delta.) To any $\omega
\in \Omega $ and strength $\vartheta \in \mathbb{R}_{0}^{+}$ of hopping
disorder, we associate a self-adjoint operator $\Delta _{\omega ,\vartheta
}\in \mathcal{B}(\ell ^{2}(\mathbb{Z}^{d}))$ describing the hoppings of a
single particle in the lattice:
\begin{eqnarray}
\lbrack \Delta _{\omega ,\vartheta }(\psi )](x) &\doteq &2d\psi
(x)-\sum_{j=1}^{d}\Big((1+\vartheta \overline{\omega _{2}(\{x,x-e_{j}\})})\
\psi (x-e_{j})  \notag \\
&&+\psi (x+e_{j})(1+\vartheta \omega _{2}(\{x,x+e_{j}\}))\Big)
\label{equation sup}
\end{eqnarray}%
for any $x\in \mathbb{Z}^{d}$ and $\psi \in \ell ^{2}(\mathbb{Z}^{d})$, with
$\{e_{k}\}_{k=1}^{d}$ being the canonical orthonormal basis of the Euclidian
space $\mathbb{R}^{d}$. In the case of vanishing hopping disorder $\vartheta
=0$, (up to a minus sign) $\Delta _{\omega ,0}$ is the usual $d$-dimensional
discrete Laplacian. Since the hopping amplitudes are complex-valued ($\omega
_{2}$ takes values in $\mathbb{D}$), note additionally that random vector
potentials can be implemented in our model. Then, the random tight-binding
model is the one-particle Hamiltonian defined by
\begin{equation}
h^{(\omega )}\doteq \Delta _{\omega ,\vartheta }+\lambda \omega _{1}\ ,\text{%
\qquad }\omega =\left( \omega _{1},\omega _{2}\right) \in \Omega ,\ \lambda
,\vartheta \in \mathbb{R}_{0}^{+},  \label{eq:Ham_lap_pot}
\end{equation}%
where the function $\omega _{1}\colon \mathbb{Z}^{d}\rightarrow \lbrack
-1,1] $ is identified with the corresponding (self-adjoint) multiplication
operator. We use this operator to define a (infinite volume) dynamics, by
the unitary group $\{\mathrm{e}^{ith^{(\omega )}}\}_{t\in \mathbb{R}}$, in
the one-particle Hilbert space $\mathfrak{h}$. Note that the tight-binding
Anderson model corresponds to the special case $\vartheta =0$. \medskip

\noindent \underline{(iv):} Let
\begin{eqnarray*}
\mathfrak{Z} &\doteq &\left\{ \mathcal{Z}\subset 2^{\mathbb{Z}^{d}}\colon
\left( \forall Z_{1},Z_{2}\in \mathcal{Z}\right) \ Z_{1}\neq
Z_{2}\Longrightarrow Z_{1}\cap Z_{2}=\emptyset \text{ }\right\} , \\
\mathfrak{Z}_{\text{f}} &\doteq &\left\{ \mathcal{Z}\in \mathfrak{Z}\colon
\left\vert \mathcal{Z}\right\vert <\infty \text{ and }\left( \forall Z\in
\mathcal{Z}\right) \ 0<\left\vert Z\right\vert <\infty \right\} .
\end{eqnarray*}%
One can restrict the dynamics to collections $\mathcal{Z}\in \mathfrak{Z}$
of disjoint subsets of the lattice by using the orthogonal projections $%
P_{\Lambda }$, $\Lambda \subset \mathbb{Z}^{d}$, defined on $\mathfrak{h}$
by
\begin{equation}
\lbrack P_{\Lambda }(\varphi )](x)\doteq \left\{
\begin{array}{lll}
\varphi (x) & , & \text{if }x\in \Lambda . \\
0 & , & \text{else.}%
\end{array}%
\right.  \label{orthogonal projection}
\end{equation}%
Then, the one-particle Hamiltonian within $\mathcal{Z}\in \mathfrak{Z}$ is%
\begin{equation}
h_{\mathcal{Z}}^{(\omega )}\doteq \sum_{Z\in \mathcal{Z}}P_{Z}h^{(\omega
)}P_{Z},  \label{h finite f}
\end{equation}%
leading to the unitary group $\{\mathrm{e}^{ith_{\mathcal{Z}}^{(\omega
)}}\}_{t\in \mathbb{R}}$. This kind of decomposition over collections of
disjoint subsets of the lattice is important in the technical proofs.
\medskip

\noindent \underline{(v):} By the Combes-Thomas estimate (Appendix \ref%
{Section Combes-Thomas Estimate}),
\begin{equation}
\left\vert \left\langle \mathfrak{e}_{x},\mathrm{e}^{ith_{\mathcal{Z}%
}^{(\omega )}}\mathfrak{e}_{y}\right\rangle _{\mathfrak{h}}\right\vert \leq
36\mathrm{e}^{\left\vert t\eta \right\vert -2\mu _{\eta }|x-y|}
\label{Combes-ThomasCombes-Thomas}
\end{equation}%
for any $\eta ,\mu \in \mathbb{R}^{+}$, $x,y\in \mathbb{Z}^{d}$, $\mathcal{Z}%
\in \mathfrak{Z}$, $\omega \in \Omega $, and $\lambda ,\vartheta \in \mathbb{%
R}_{0}^{+}$, where
\begin{equation}
\mu _{\eta }\doteq \mu \min \left\{ \frac{1}{2},\frac{\eta }{8d\left(
1+\vartheta \right) \mathrm{e}^{\mu }}\right\} .
\label{Combes-ThomasCombes-Thomasbis}
\end{equation}%
See Corollary \ref{Lemma fermi1}, by observing that the parameter $\mathbf{S}
$ defined by (\ref{S}) is bounded in this case by $\mathbf{S}(h_{\mathcal{Z}%
}^{(\omega )},\mu )\leq 2d(1+\vartheta )\mathrm{e}^{\mu }$.

\subsection{Algebraic Setting\label{subsec:algebraic}}

Although all the problem can be formulated, in a mathematically equivalent
way, in the one-particle (or Hilbert space) setting (Appendix \ref%
{One-particle formulation}), since the underlying physical system is a
many-body one, it is conceptually more appropriate to state the large
deviation principle (LDP) related to microscopic current densities within
the algebraic formulation for lattice fermion systems:\medskip

\noindent \underline{(i):} We denote by $\mathcal{U\equiv U}_{\mathfrak{h}}$
the CAR $C^{\ast }$-algebra generated by the identity $\mathfrak{1}$ and
elements $\{a(\psi )\}_{\psi \in \mathfrak{h}}$ satisfying the canonical
anticommutation relations (CAR): For all $\psi ,\varphi \in \mathfrak{h}$,
\begin{equation}
a(\psi )a(\varphi )=-a(\varphi )a(\psi ),\quad a(\psi )a(\varphi )^{\ast
}+a(\varphi )^{\ast }a(\psi )=\left\langle \psi ,\varphi \right\rangle _{%
\mathfrak{h}}\mathfrak{1}.  \label{eq:CAR}
\end{equation}%
Note that CAR imply that, for all $\psi \in \mathfrak{h}$,
\begin{equation}
\left\Vert a(\psi )\right\Vert _{\mathcal{U}}\leq \left\Vert \psi
\right\Vert _{\mathfrak{h}},  \label{eq:bound_CAR}
\end{equation}%
and the map $\psi \mapsto a(\psi )^{\ast }$ from $\mathfrak{h}$ to $\mathcal{%
U}$ is linear. As is usual, $a(\psi )$ and $a(\psi )^{\ast }$ are called,
respectively, annihilation and creation operators.

\noindent \underline{(ii):} For all $\omega \in \Omega $ and $\lambda
,\vartheta \in \mathbb{R}_{0}^{+}$, the dynamics on the CAR $C^{\ast }$%
-algebra $\mathcal{U}$ is defined by a strongly continuous group $\tau
^{(\omega )}\doteq \{\tau _{t}^{(\omega )}\}_{t\in {\mathbb{R}}}$ of
(Bogoliubov) $\ast $-automorphisms of $\mathcal{U}$ satisfying%
\begin{equation}
\tau _{t}^{(\omega )}(a(\psi ))=a(\mathrm{e}^{ith^{(\omega )}}\psi )\ ,\text{%
\qquad }t\in \mathbb{R},\ \psi \in \mathfrak{h}.  \label{rescaledbis}
\end{equation}%
See (\ref{eq:Ham_lap_pot}) as well as \cite[Theorem 5.2.5]%
{bratteli2003operator2} for more details on Bogoliubov automorphisms.
Similarly, for any $\mathcal{Z}\in \mathfrak{Z}$, we define the strongly
continuous group $\tau ^{(\omega ,\mathcal{Z})}$ by replacing $h^{(\omega )}$
in (\ref{rescaledbis}) with $h_{\mathcal{Z}}^{(\omega )}$ (see (\ref{h
finite f})). In order to define the thermodynamic limit, we introduce the
increasing family%
\begin{equation}
\Lambda _{\ell }\doteq \{(x_{1},\ldots ,x_{d})\in \mathbb{Z}^{d}\colon
|x_{1}|,\ldots ,|x_{d}|\leq \ell \},\qquad \ell \in \mathbb{R}_{0}^{+},
\label{eq:boxesl1}
\end{equation}%
in $\mathcal{P}_{\text{f}}(\mathbb{Z}^{d})$. Observe that, for any $t\in {%
\mathbb{R}}$, $\tau _{t}^{(\omega ,\{\Lambda _{\ell }\})}$ converges
strongly to $\tau _{t}^{(\omega )}\equiv \tau _{t}^{(\omega ,\{\mathbb{Z}%
^{d}\})}$, as $\ell \rightarrow \infty $.\medskip

\noindent \underline{(iii):} For any realization $\omega \in \Omega $ and
disorder strengths $\lambda ,\vartheta \in \mathbb{R}_{0}^{+}$, the thermal
equilibrium state of the system at inverse temperature $\beta \in \mathbb{R}%
^{+}$ (i.e., $\beta >0$) is by definition the unique $(\tau ^{(\omega
)},\beta )$-KMS state $\varrho ^{(\omega )}$, see \cite[Example 5.3.2.]%
{bratteli2003operator2} or \cite[Theorem 5.9]{AttalJoyePillet2006a}. It is
well-known that such a state is stationary w.r.t. the dynamics $\tau
^{(\omega )}$, that is,
\begin{equation}
\varrho ^{(\omega )}\circ \tau _{t}^{(\omega )}=\varrho ^{(\omega )}\
,\qquad \omega \in \Omega ,\ t\in \mathbb{R}.  \label{stationary}
\end{equation}%
The state $\varrho ^{(\omega )}$ is also \emph{gauge-invariant} and \emph{%
quasi-free}, and it satisfies%
\begin{equation}
\varrho ^{(\omega )}(a^{\ast }\left( \varphi \right) a\left( \psi \right)
)=\left\langle \psi ,\frac{1}{1+\mathrm{e}^{\beta h^{(\omega )}}}\varphi
\right\rangle _{\mathfrak{h}},\qquad \varphi ,\psi \in \mathfrak{h}.
\label{2-point correlation function}
\end{equation}%
For $\beta =0$, one gets the tracial state (or chaotic state), denoted by $%
\text{tr}\in \mathcal{U}^{\ast }$.

Recall that gauge-invariant quasi--free states are positive linear
functionals $\rho \in \mathcal{U}^{\ast }$ such that $\rho (\mathfrak{1})=1$
and, for all $N_{1},N_{2}\in \mathbb{N}$ and $\psi _{1},\ldots ,\psi
_{N_{1}+N_{2}}\in \mathfrak{h}$,
\begin{equation}
\rho \left( a^{\ast }(\psi _{1})\cdots a^{\ast }(\psi _{N_{1}})a(\psi
_{N_{1}+N_{2}})\cdots a(\psi _{N_{1}+1})\right) =0  \label{ass O0-00}
\end{equation}%
if $N_{1}\neq N_{2}$, while in the case $N_{1}=N_{2}\equiv N$,
\begin{equation}
\rho \left( a^{\ast }(\psi _{1})\cdots a^{\ast }(\psi _{N})a(\psi
_{2N})\cdots a(\psi _{N+1})\right) =\mathrm{det}\left[ \rho \left(
a^{+}(\psi _{k})a(\psi _{N+l})\right) \right] _{k,l=1}^{N}.
\label{ass O0-00bis}
\end{equation}%
See, e.g., \cite[Definition 3.1]{Araki}, which refers to a more general
notion of quasi-free states. The gauge-invariant property corresponds to
Equation (\ref{ass O0-00}) whereas \cite[Definition 3.1, Condition (3.1)]%
{Araki} only imposes the quasi--free state to be even, which is a strictly
weaker property than being gauge-invariant.

Similarly, for any $\mathcal{Z}\in \mathfrak{Z}$, we define the quasi-free
state $\varrho _{\mathcal{Z}}^{(\omega )}$ by replacing $h^{(\omega )}$ in (%
\ref{2-point correlation function}) with $h_{\mathcal{Z}}^{(\omega )}$ (see (%
\ref{h finite f})). In the thermodynamic limit $\ell \rightarrow \infty $, $%
\varrho _{\{\Lambda _{\ell }\}}^{(\omega )}$ converges in the weak$^{\ast }$
topology to $\varrho ^{(\omega )}\equiv \varrho _{\{\mathbb{Z}%
^{d}\}}^{(\omega )}$.

\subsection{Current Densities\label{Current Densities}}

\noindent \underline{(i) Currents:} Fix $\omega \in \Omega $ and $\vartheta
\in \mathbb{R}_{0}^{+}$. For any oriented edge $(x,y)\in \left( \mathbb{Z}%
^{d}\right) ^{2}$, we define the paramagnetic current observable by
\begin{equation}
I_{(x,y)}^{(\omega )}\doteq -2\Im \mathrm{m}\left( \langle \mathfrak{e}%
_{x},\Delta _{\omega ,\vartheta }\mathfrak{e}_{y}\rangle _{\mathfrak{h}}a(%
\mathfrak{e}_{x})^{\ast }a(\mathfrak{e}_{y})\right) .
\label{current observable}
\end{equation}%
It is seen as a current because it satisfies a discrete continuity equation,
as explained in Appendix \ref{One-particle formulation2}. Here, the
self-adjoint operators $\Im \mathrm{m}(A)\in \mathcal{U}$ and $\Re \mathrm{e}%
(A)\in \mathcal{U}$ are the \emph{imaginary} and \emph{real parts} of $A\in
\mathcal{U}$, that are, respectively,
\begin{equation}
\Im \mathrm{m}\left( A\right) \doteq \frac{1}{2i}\left( A-A^{\ast }\right)
\quad \text{and}\quad \Re \mathrm{e}\left( A\right) \doteq \frac{1}{2}\left(
A+A^{\ast }\right) \ .  \label{im and real part}
\end{equation}%
This \textquotedblleft second-quantized\textquotedblright\ definition of
current observable and the usual one in the one-particle setting, like in
\cite{[SBB98],jfa,klm}, are perfectly equivalent, in the case of
non-interacting fermions. See for instance Equation (\ref{equation refereee}%
).

Note that electric fields accelerate charged particles and induce so-called
diamagnetic currents, which correspond to the ballistic movement of
particles. In fact, as explained in \cite[Sections III and IV]{OhmII}, this
component of the total current creates a kind of \textquotedblleft wave
front\textquotedblright\ that destabilizes the whole system by changing its
state. The presence of diamagnetic currents leads then to the progressive
appearance of paramagnetic currents which are responsible for heat
production and the in-phase AC-conductivity of the system. For more details,
see \cite{OhmII,OhmV,OhmVI} as well as Appendix \ref{linear_response} on
linear response currents. \medskip

\noindent \underline{(ii) Conductivity:} As is usual, $\left[ A,B\right]
\doteq AB-BA\in \mathcal{U}$ denotes the commutator between the elements $%
A\in \mathcal{U}$ and $B\in \mathcal{U}$. For any finite subset $\Lambda \in
\mathcal{P}_{\text{f}}(\mathbb{Z}^{d})$, we define the space-averaged
transport coefficient observable $\mathcal{C}_{\Lambda }^{(\omega )}\in
C^{1}(\mathbb{R};\mathcal{B}(\mathbb{R}^{d};\mathcal{U}^{d}))$, w.r.t. the
canonical orthonormal basis $\{e_{q}\}_{q=1}^{d}$ of the Euclidian space $%
\mathbb{R}^{d}$, by the corresponding matrix entries%
\begin{eqnarray}
\left\{ \mathcal{C}_{\Lambda }^{(\omega )}\left( t\right) \right\} _{k,q}
&\doteq &\frac{1}{\left\vert \Lambda \right\vert }\underset{%
x,y,x+e_{k},y+e_{q}\in \Lambda }{\sum }\int\nolimits_{0}^{t}i[\tau _{-\alpha
}^{(\omega )}(I_{\left( y+e_{q},y\right) }^{(\omega )}),I_{\left(
x+e_{k},x\right) }^{(\omega )}]\mathrm{d}\alpha   \notag \\
&&+\frac{2\delta _{k,q}}{\left\vert \Lambda \right\vert }\underset{x\in
\Lambda }{\sum }\Re \mathrm{e}\left( \langle \mathfrak{e}_{x+e_{k}},\Delta
_{\omega ,\vartheta }\mathfrak{e}_{x}\rangle a(\mathfrak{e}_{x+e_{k}})^{\ast
}a(\mathfrak{e}_{x})\right)   \label{defininion para coeff observable}
\end{eqnarray}%
for any $\omega \in \Omega $, $t\in \mathbb{R}$, $\lambda ,\vartheta \in
\mathbb{R}_{0}^{+}$ and $k,q\in \{1,\ldots ,d\}$. This object is the
conductivity observable matrix associated with the lattice region $\Lambda $
and time $t$. See Appendix \ref{linear_response}, in particular Equations (%
\ref{linear response current1})-(\ref{linear response current1bis}). In
fact, the first term in the right-hand side of (\ref{defininion para coeff
observable}) corresponds to the paramagnetic coefficient, whereas the second
one is the diamagnetic component. For more details, see \cite[Theorem 3.7]%
{OhmV}. \medskip

\noindent \underline{(iii) Linear response current density:} Fix a direction
$\vec{w}\in {\mathbb{R}}^{d}$ with $\left\Vert \vec{w}\right\Vert _{\mathbb{R%
}^{d}}=1$ and a (time-dependent) continuous, compactly supported, electric
field $\mathcal{E}\in C_{0}^{0}(\mathbb{R};\mathbb{R}^{d})$, i.e., the
external electric field is a continuous function $t\mapsto \mathcal{E}(t)\in
\mathbb{R}^{d}$ of time $t\in \mathbb{R}$ with compact support. Then, as it
is explained in Appendix \ref{linear_response}, \cite{OhmV,OhmVI}\footnote{%
Strictly speaking, these papers use smooth electric fields, but the
extension to the continuous case is straightforward.} shows that the
space-averaged linear response current observable in the lattice region $%
\Lambda $ and at time $t=0$ in the direction $\vec{w}$ is equal to%
\begin{equation}
\mathbb{I}_{\Lambda }^{(\omega ,\mathcal{E})}\doteq \underset{k,q=1}{\sum^{d}%
}w_{k}\int_{-\infty }^{0}\left\{ \mathcal{E}\left( \alpha \right) \right\}
_{q}\left\{ \mathcal{C}_{\Lambda }^{(\omega )}\left( -\alpha \right)
\right\} _{k,q}\mathrm{d}\alpha \ .  \label{current}
\end{equation}%
To obtain the current density at any time $t\in \mathbb{R}$, it suffices to
replace $\mathcal{E}\in C_{0}^{0}(\mathbb{R};\mathbb{R}^{d})$ in this
equation with
\begin{equation}
\mathcal{E}_{t}(\alpha )\doteq \mathcal{E}\left( \alpha +t\right) ,\qquad
\alpha \in \mathbb{R}.  \label{new field}
\end{equation}%
Compare with Equations (\ref{linear response current1})-(\ref{linear
response current1bis}).

\section{Main Results\label{sec:main}}

We study large deviations (LD) for the microscopic current density produced
by any fixed, time-dependent electric field $\mathcal{E}$. Via the G\"{a}%
rtner-Ellis theorem (see, e.g., \cite[Corollary 4.5.27]{dembo1998large}),
this is a consequence of the following result:

\begin{theorem}[Generating functions for currents]
\label{Theorem main result}\mbox{ }\newline
There is a measurable subset $\tilde{\Omega}\subset \Omega $ of full measure
such that, for all $\beta \in \mathbb{R}^{+}$, $\vartheta ,\lambda \in
\mathbb{R}_{0}^{+}$, $\omega \in \tilde{\Omega}$, $\mathcal{E}\in C_{0}^{0}(%
\mathbb{R};\mathbb{R}^{d})$ and $\vec{w}\in {\mathbb{R}}^{d}$ with $%
\left\Vert \vec{w}\right\Vert _{\mathbb{R}^{d}}=1$, the limit
\begin{equation*}
\lim_{L\to \infty }\frac{1}{\left\vert \Lambda _{L}\right\vert }\ln \varrho
^{(\omega )}\left( \mathrm{e}^{\left\vert \Lambda _{L}\right\vert \mathbb{I}%
_{\Lambda _{L}}^{(\omega ,\mathcal{E})}}\right)
\end{equation*}%
exist and equals%
\begin{equation*}
\mathrm{J}^{(\mathcal{E})}\doteq \lim_{L\to \infty }\frac{1}{\left\vert
\Lambda _{L}\right\vert }\mathbb{E}\left[ \ln \varrho ^{(\cdot )}\left(
\mathrm{e}^{\left\vert \Lambda _{L}\right\vert \mathbb{I}_{\Lambda
_{L}}^{(\cdot ,\mathcal{E})}}\right) \right] .
\end{equation*}%
Moreover, for any $\mathcal{E}\in C_{0}^{0}(\mathbb{R};\mathbb{R}^{d})$, the
map $s\mapsto \mathrm{J}^{(s\mathcal{E})}$ from ${\mathbb{R}}$ to itself is
continuously differentiable and convex.
\end{theorem}

\begin{proof}
The assertions directly follow from Corollaries \ref{Ackoglu-Krengel ergodic
theorem II copy(2)} and \ref{final corrollary}. Note that the map $s\mapsto
\mathrm{J}^{(s\mathcal{E})}$ is a limit of convex functions, and hence, it
is also convex.
\end{proof}

In probability theory, the law of large numbers refers to the convergence
(at least in probability), as $n\rightarrow \infty $, of the average or
empirical mean of $n$ independent identically distributed (i.i.d.) random
variables towards their expected value (assuming it exists). The large
deviation formalism quantitatively describes, for large $n\gg 1$, the
probability of finding an empirical mean that differs from the expected
value. These are \emph{rare} events, by the law of large numbers, and an LD
principle (LDP) gives their probability as exponentially small (w.r.t. some
speed) in the limit $n\rightarrow \infty $.

In the context of the algebraic formulation of quantum mechanics,
observables (i.e., self-adjoint elements of some $C^{\ast }$-algebra, here $%
\mathcal{U}$) generalize the notion of random variables of classical
probability theory. The link between both notions is given via the
Riesz-Markov theorem and functional calculus: The commutative $C^{\ast }$%
-subalgebra of $\mathcal{U}$ generated by any self-adjoint element $A^{\ast
}=A\in \mathcal{U}$ is isomorphic to the algebra of continuous functions on
the compact set $\mathrm{spec}(A)\subset \mathbb{R}$. Hence, by the
Riesz-Markov theorem, for any state $\rho \in \mathcal{U}^{\ast }$, there is
a unique probability measure $\mathfrak{m}_{\rho ,A}$ on $\mathbb{R}$ such
that
\begin{equation}
\mathfrak{m}_{\rho ,A}(\mathrm{spec}(A))=1\qquad \text{and}\qquad \rho
\left( f(A)\right) =\int_{\mathbb{R}}f(x)\mathfrak{m}_{\rho ,A}(\mathrm{d}x)
\label{fluctuation measure}
\end{equation}%
for all complex-valued continuous functions $f\in C(\mathbb{R};\mathbb{C})$.
$\mathfrak{m}_{\rho ,A}$ is called the \emph{distribution} of the observable
$A$ in the state $\rho $. The LD formalism naturally arises also in this
more general framework: A \emph{rate} function is a lower semi-continuous
function $\mathrm{I}:\mathbb{R}\rightarrow \lbrack 0,\infty ]$. If $\mathrm{I%
}$ is not the $\infty $ constant function and has compact level sets, i.e.,
if $\mathrm{I}^{-1}([0,m])=\{x\in \mathbb{R}\colon \mathrm{I}(x)\leq m\}$ is
compact for any $m\geq 0$, then one says that $\mathrm{I}$ is a \emph{good}
rate function. A sequence $(A_{L})_{L\in \mathbb{N}}\subset \mathcal{U}$ of
observables satisfies an LDP, in a state $\rho \in \mathcal{U}^{\ast }$,
with speed $(\mathfrak{n}_{L})_{L\in \mathbb{N}}\subset \mathbb{R}^{+}$ (a
positive, increasing and divergent sequence) and rate function $\mathrm{I}$
if, for any borel subset $\mathcal{G}$ of $\mathbb{R}$,
\begin{equation*}
-\inf_{x\in \mathcal{G}^{\circ }}\mathrm{I}\left( x\right) \leq
\liminf_{L\rightarrow \infty }\frac{1}{\mathfrak{n}_{L}}\ln \mathfrak{m}%
_{\rho ,A_{L}}\left( \mathcal{G}\right) \leq \limsup_{L\rightarrow \infty }%
\frac{1}{\mathfrak{n}_{L}}\ln \mathfrak{m}_{\rho ,A_{L}}\left( \mathcal{G}%
\right) \leq -\inf_{x\in \mathcal{\bar{G}}}\mathrm{I}\left( x\right) .
\end{equation*}%
Here, $\mathcal{G}^{\circ }$ is the interior of $\mathcal{G}$, while $%
\mathcal{\bar{G}}$ is its closure. Compare with Equations (\ref{LDP upp})-(%
\ref{LDP lower}) in Appendix \ref{sec:LDP}.

A sufficient condition to ensure that a sequence of observables satisfies an
LDP is given by the G\"{a}rtner-Ellis theorem. In particular, Theorem \ref%
{Theorem main result} combined with Theorem \ref{prop Gartner--Ellis} yields
the following corollary:

\begin{corollary}[Large deviation principle for currents]
\label{LDP copy(1)}\mbox{ }\newline
Let $\tilde{\Omega}\subset \Omega $ be the measurable subset of full measure
of Theorem \ref{Theorem main result}. Then, for all $\beta \in \mathbb{R}^{+}
$, $\vartheta ,\lambda \in \mathbb{R}_{0}^{+}$, $\omega \in \tilde{\Omega}$,
$l\in \mathbb{N}$, $\mathcal{E}\in C_{0}^{0}(\mathbb{R};\mathbb{R}^{d})$ and
$\vec{w}\in {\mathbb{R}}^{d}$ with $\left\Vert \vec{w}\right\Vert _{\mathbb{R%
}^{d}}=1$, the sequence $(\mathbb{I}_{\Lambda _{L}}^{(\omega ,\mathcal{E}%
)})_{L\in \mathbb{N}}$ of microscopic current densities satisfies an LDP, in
the KMS state $\varrho ^{(\omega )}$, with speed $\left\vert \Lambda
_{L}\right\vert $ and good rate function $\mathrm{I}^{(\mathcal{E)}}$
defined on $\mathbb{R}$ by
\begin{equation*}
\mathrm{I}^{(\mathcal{E)}}(x)\doteq \sup\limits_{s\in \mathbb{R}}\left\{ sx-%
\mathrm{J}^{(s\mathcal{E})}\right\} \geq 0.
\end{equation*}
\end{corollary}

\begin{remark}
\mbox{ }\newline
By direct estimates, one verifies that, for any fixed state $\rho $, $(%
\mathbb{I}_{\Lambda _{L}}^{(\omega ,\mathcal{E})})_{L\in \mathbb{N}}$ yields
an exponentially tight family of probability measures, defined by (\ref%
{fluctuation measure}) for $A=\mathbb{I}_{\Lambda _{L}}^{(\omega ,\mathcal{E}%
)}$. Therefore, by \cite[Lemma 4.1.23]{dembo1998large}, $(\mathbb{I}%
_{\Lambda _{L}}^{(\omega ,\mathcal{E})})_{L\in \mathbb{N}}$ satisfies, along
some subsequence, an LDP, \emph{in any state }$\rho $, with speed $%
\left\vert \Lambda _{L}\right\vert $ and a good rate function. However, it
is not clear whether this rate function depends on the choice of
subsequences and $\omega \in \Omega $. Moreover, no information on
minimizers of the rate function, like in Theorem \ref{Theorem main result
copy(1)}, can be deduced from \cite[Lemma 4.1.23]{dembo1998large}.
\end{remark}

Observe that, if an LDP holds true, then the law of large numbers follows
\cite[Theorem II.6.4]{E85} from the Borel-Cantelli lemma \cite[Lemma A.5.2]%
{E85}. Therefore, by \cite{OhmIII,OhmVI} and Corollary \ref{LDP copy(1)},
the distributions of the microscopic current density observables, in the
state $\varrho ^{(\omega )}$, weak$^{\ast }$ converges, for $\omega \in
\Omega $ almost surely, to the delta distribution at the (classical value of
the) macroscopic current density. Using Theorem \ref{Theorem main result},
we sharpen this result by proving that the microscopic current density
converges \emph{exponentially fast} to the macroscopic one, w.r.t. the
volume $\left\vert \Lambda _{L}\right\vert $ of the region of the lattice
where an external electric field is applied.

To this end, we remark from Corollary \ref{final corrollary} (see (\ref%
{current utiles})) that, for any $\beta \in \mathbb{R}^{+}$, $\vartheta
,\lambda \in \mathbb{R}_{0}^{+}$, $\vec{w}\in {\mathbb{R}}^{d}$ with $%
\left\Vert \vec{w}\right\Vert _{\mathbb{R}^{d}}=1$, the macroscopic current
density is equal to
\begin{equation}
x^{(\mathcal{E})}\doteq \partial _{s}\mathrm{J}^{(s\mathcal{E}%
)}|_{s=0},\qquad \mathcal{E}\in C_{0}^{0}(\mathbb{R};\mathbb{R}^{d}).
\label{def currents}
\end{equation}%
See also (\ref{limit current}). Define%
\begin{equation*}
x_{-}\doteq \inf \left\{ x\leq x^{(\mathcal{E})}\colon \mathrm{I}^{(\mathcal{%
E)}}\left( x\right) <\infty \right\} ,\text{ }x_{+}\doteq \sup \left\{ x\geq
x^{(\mathcal{E})}\colon \mathrm{I}^{(\mathcal{E)}}\left( x\right) <\infty
\right\} .
\end{equation*}%
Obviously, $\mathrm{I}^{(\mathcal{E)}}\left( x\right) =\infty $ for $x\in {%
\mathbb{R}}\backslash \lbrack x_{-},x_{+}]$. We start by giving important
properties of the rate function $\mathrm{I}^{(\mathcal{E)}}$:

\begin{theorem}[Properties of the rate function]
\label{Theorem main result copy(1)}\mbox{ }\newline
Fix $\beta \in \mathbb{R}^{+}$, $\vartheta ,\lambda \in \mathbb{R}_{0}^{+}$,
$\vec{w}\in {\mathbb{R}}^{d}$ with $\left\Vert \vec{w}\right\Vert _{\mathbb{R%
}^{d}}=1$ and $\mathcal{E}\in C_{0}^{0}(\mathbb{R};\mathbb{R}^{d})$. The
rate function $\mathrm{I}^{(\mathcal{E)}}$ is a lower-semicontinuous convex
function satisfying: (i) $\mathrm{I}^{(\mathcal{E)}}(x^{(\mathcal{E})})=0$;
(ii) $\mathrm{I}^{(\mathcal{E)}}(x)>0$ if $x\neq x^{(\mathcal{E})}$; (iii) $%
\mathrm{I}^{(\mathcal{E)}}\left( x\right) <\infty $ for $x\in (x_{-},x_{+})$
with $\mathrm{I}^{(\mathcal{E)}}\left( x\right) \leq \mathrm{I}^{(\mathcal{E)%
}}\left( x_{-}\right) $ for $x\in (x_{-},x^{(\mathcal{E})}]$ and $\mathrm{I}%
^{(\mathcal{E)}}\left( x\right) \leq \mathrm{I}^{(\mathcal{E)}}\left(
x_{+}\right) $ for $x\in \lbrack x^{(\mathcal{E})},x_{+})$; (iv) $\mathrm{I}%
^{(\mathcal{E)}}$ restricted to the interior of its domain, i.e., the
(possibly empty) open interval $(x_{-},x_{+})$, is continuous.
\end{theorem}

\begin{proof}
Fix all parameters of the theorem. Note that $\mathrm{I}^{(\mathcal{E)}}$ is
clearly a lower-semicontinuous convex function, by construction. As the map $%
s\mapsto \mathrm{J}^{(s\mathcal{E})}$ is differentiable and convex (Theorem %
\ref{Theorem main result}), the map $s\mapsto \mathrm{J}^{(s\mathcal{E})}$
is the Legendre-Fenchel transform of $\mathrm{I}^{(\mathcal{E)}}$, i.e.,
\begin{equation*}
\mathrm{J}^{(s\mathcal{E})}=\sup\limits_{x\in \mathbb{R}}\left\{ sx-\mathrm{I%
}^{(\mathcal{E)}}(x)\right\} ,\qquad s\in \mathbb{R},
\end{equation*}%
and $s_{0}$ is a solution of the variational problem
\begin{equation*}
\mathrm{I}^{(\mathcal{E)}}(x)\doteq \sup\limits_{s\in \mathbb{R}}\left\{ sx-%
\mathrm{J}^{(s\mathcal{E})}\right\}
\end{equation*}%
if and only if $s_{0}$ solves $x=\partial _{s}\mathrm{J}^{(s\mathcal{E}%
)}|_{s=s_{0}}$. By (\ref{def currents}), it follows that%
\begin{equation*}
0=\mathrm{J}^{(0)}=\inf_{x\in \mathbb{R}}\mathrm{I}^{(\mathcal{E)}}(x)=%
\mathrm{I}^{(\mathcal{E)}}(x^{(\mathcal{E})}).
\end{equation*}%
This proves Assertion (i). To prove (ii), it suffices to show that $x^{(%
\mathcal{E})}$ is the only minimizer of $\mathrm{I}^{(\mathcal{E)}}$. Note
that $x_{0}$ is a minimizer of $\mathrm{I}^{(\mathcal{E)}}$ if and only if $%
0 $ is a subdifferential of $\mathrm{I}^{(\mathcal{E)}}$ at $x_{0}$
(Fermat's principle). By \cite[Corollary 5.3.3]{Lu} and the
differentiability of the Legendre transform of $\mathrm{I}^{(\mathcal{E)}}$,
which is the map $s\mapsto \mathrm{J}^{(s\mathcal{E})}$, it follows that the
minimizer of $\mathrm{I}^{(\mathcal{E)}}$ is unique and Assertion (ii)
follows. Assertion (iii) is a direct consequence of the fact that $\mathrm{I}%
^{(\mathcal{E)}}$ is a convex function with $x^{(\mathcal{E})}$ as unique
minimizer. Assertion (iv) is deduced from \cite[Corollary 2.1.3]{Lu}.
\end{proof}

\begin{corollary}[Exponentially fast suppression of quantum uncertainty of
currents]
\label{LDP}\mbox{ }\newline
Let $\tilde{\Omega}\subset \Omega $ be the measurable subset of full measure
of Theorem \ref{Theorem main result}. Then, for all $\beta \in \mathbb{R}%
^{+} $, $\vartheta ,\lambda \in \mathbb{R}_{0}^{+}$, $\omega \in \tilde{%
\Omega}$, $l\in \mathbb{N}$, $\mathcal{E}\in C_{0}^{0}(\mathbb{R};\mathbb{R}%
^{d})$, $\vec{w}\in {\mathbb{R}}^{d}$ with $\left\Vert \vec{w}\right\Vert _{%
\mathbb{R}^{d}}=1$, and any open subset $\mathcal{O}\subset \mathbb{R}$ with
$x^{(\mathcal{E})}\notin \mathcal{\bar{O}}$,
\begin{equation*}
\limsup_{L\to \infty }\frac{1}{\left\vert \Lambda _{L}\right\vert }\ln
\mathfrak{m}_{\varrho ^{(\omega )},\mathbb{I}_{\Lambda _{L}}^{(\omega ,%
\mathcal{E})}}\left( \mathcal{O}\right) <0.
\end{equation*}%
The above limit does not depend on the particular realization of $\omega \in
\tilde{\Omega}$. If, additionally, $\mathcal{O\cap }(x_{-},x_{+})\neq
\emptyset $, then
\begin{equation*}
\lim_{L\to \infty }\frac{1}{\left\vert \Lambda _{L}\right\vert }\ln
\mathfrak{m}_{\varrho ^{(\omega )},\mathbb{I}_{\Lambda _{L}}^{(\omega ,%
\mathcal{E})}}\left( \mathcal{O}\right) =-\inf_{x\in \mathcal{O}}\mathrm{I}%
^{(\mathcal{E)}}\left( x\right) <0.
\end{equation*}%
See (\ref{fluctuation measure}) for the definition of the distribution of $%
\mathbb{I}_{\Lambda _{L}}^{(\omega ,\mathcal{E})}$, in the KMS state $%
\varrho ^{(\omega )}$.
\end{corollary}

\begin{proof}
It is a direct consequence of Corollary \ref{LDP copy(1)} and Theorem \ref%
{Theorem main result copy(1)}.
\end{proof}

Corollary \ref{LDP} shows that the microscopic current density converges
\emph{exponentially fast} to the macroscopic one, w.r.t. the volume $%
\left\vert \Lambda _{L}\right\vert $ (in lattice units (l.u.)) of the region
of the lattice where the electric field is applied. As discussed in the
introduction, this is in accordance with the low temperature ($4.2~\mathrm{K}
$) experiment \cite{Ohm-exp} on the resistance of nanowires with lengths
down to approximately $40\ \mathrm{l.u.}$ ($L\simeq 20$). The breakdown of
the classical description of these nanowires is expected \cite%
{Ohm-exp2,16exp,15exp} to be around $20\ \mathrm{l.u.}$ ($L\simeq 10$).

To conclude, note that, in the experimental setting of \cite{ZDAW,Ohm-exp},
contacts are used to impose an electric potential difference to the
nanowires. These contacts yield supplementary resistances to the systems
that are well-described by Landauer's formalism \cite{Lan} when a \emph{%
ballistic} charge transport takes place in the nanowires. In our model, the
purely ballistic charge transport is reached when $\vartheta =0$ and $%
\lambda \rightarrow 0^{+}$, as proven in \cite[Theorem 4.6]{OhmIV}. When the
nanowire resistance becomes relatively small as compared to the contact
resistances, then the charge transport in the nanowire is well-described by
a ballistic approximation and Landauer's formalism applies, as also
experimentally verified in \cite{ZDAW}. This is the reason why \cite{Ohm-exp}
reaches much smaller length scales than \cite{ZDAW}: the material used in
\cite{Ohm-exp} has a much larger linear resistivity (between $112~\mathrm{%
\Omega /nm}$ and $855~\mathrm{\Omega /nm}$, see \cite[Table 1]{Ohm-exp})
than the one of \cite{ZDAW} ($23~\mathrm{\Omega /nm}$, see \cite[discussions
after Eq. (2)]{ZDAW}).

\section{Technical Proofs\label{sec:proofs}}

\subsection{Preliminary Estimates\label{Sectino tech1}}

We start by giving two general estimates which will be used many times
afterwards. The first one is an elementary observation:

\begin{lemma}[Operator norm estimate]
\label{lemma:norm_estimate}\mbox{ }\newline
For any operator $C\in \mathcal{B}(\mathfrak{h})$,%
\begin{equation*}
\left\Vert C\right\Vert _{\mathcal{B}(\mathfrak{h})}\leq \sup_{x\in \mathbb{Z%
}^{d}}\sum_{y\in \mathbb{Z}^{d}}\left\vert \left\langle \mathfrak{e}_{x},C%
\mathfrak{e}_{y}\right\rangle _{\mathfrak{h}}\right\vert .
\end{equation*}
\end{lemma}

\begin{proof}
By the Cauchy-Schwarz inequality, for all $\varphi ,\psi \in \mathfrak{h}$,%
\begin{align*}
\left\vert \left\langle \varphi ,C\psi \right\rangle _{\mathfrak{h}%
}\right\vert &\leq \sum_{x,y\in \mathbb{Z}^{d}}\left\vert \varphi (x)\psi
(y)\left\langle \mathfrak{e}_{x},C\mathfrak{e}_{y}\right\rangle _{\mathfrak{h%
}}\right\vert \\
&=\sum_{x,y\in \mathbb{Z}^{d}}\left( \left\vert \varphi (x)\right\vert
\left\vert \left\langle \mathfrak{e}_{x},C\mathfrak{e}_{y}\right\rangle _{%
\mathfrak{h}}\right\vert ^{1/2}\right) \left( \left\vert \psi (y)\right\vert
\left\vert \left\langle \mathfrak{e}_{x},C\mathfrak{e}_{y}\right\rangle _{%
\mathfrak{h}}\right\vert ^{1/2}\right) \\
&\leq \sqrt{\sum_{x,y\in \mathbb{Z}^{d}}\left( \left\vert \varphi
(x)\right\vert ^{2}\left\vert \left\langle \mathfrak{e}_{x},C\mathfrak{e}%
_{y}\right\rangle _{\mathfrak{h}}\right\vert \right) }\sqrt{\sum_{x,y\in
\mathbb{Z}^{d}}\left\vert \psi (y)\right\vert ^{2}\left\vert \left\langle
\mathfrak{e}_{x},C\mathfrak{e}_{y}\right\rangle _{\mathfrak{h}}\right\vert }
\\
&\leq \left\Vert \varphi \right\Vert _{\mathfrak{h}}\left\Vert \psi
\right\Vert _{\mathfrak{h}}\sup_{x\in \mathbb{Z}^{d}}\sum_{y\in \mathbb{Z}%
^{d}}\left\vert \left\langle \mathfrak{e}_{x},C\mathfrak{e}_{y}\right\rangle
_{\mathfrak{h}}\right\vert.
\end{align*}
\end{proof}

The second one is a version of the Bogoliubov inequality. Recall that the
tracial state $\mathrm{tr}\in \mathcal{U}^{\ast }$ is the quasi-free state
satisfying (\ref{2-point correlation function}) at $\beta =0$.

\begin{lemma}[Bogoliubov-type inequalities]
\label{lemma:suppor_non_int1}\mbox{ }\newline
Let $C\in \mathcal{U}$ be any strictly positive element.\newline
\emph{(i)} For any continuously differentiable family $\left\{ H_{\alpha
}\right\} _{\alpha \in \mathbb{R}}\subset \mathcal{U}$ of self-adjoint
elements,
\begin{equation*}
\left\vert \partial _{\alpha }\ln \mathrm{tr}\left( C\mathrm{e}^{H_{\alpha
}}\right) \right\vert \leq \sup_{u\in \left[ -1/2,1/2\right] }\left\Vert
\mathrm{e}^{uH_{\alpha }}\left\{ \partial _{\alpha }H_{\alpha }\right\}
\mathrm{e}^{-uH_{\alpha }}\right\Vert _{\mathcal{U}}.
\end{equation*}%
\emph{(ii)} Similarly, for any self-adjoint $H_{0},H_{1}\in \mathcal{U}$,
\begin{equation*}
\left\vert \ln \mathrm{tr}\left( C\mathrm{e}^{H_{1}}\right) -\ln \mathrm{tr}%
\left( C\mathrm{e}^{H_{0}}\right) \right\vert \leq \sup_{\alpha \in \left[
0,1\right] }\sup_{u\in \left[ -1/2,1/2\right] }\left\Vert \mathrm{e}%
^{u\left( \alpha H_{1}+\left( 1-\alpha \right) H_{0}\right) }\left(
H_{1}-H_{0}\right) \mathrm{e}^{-u\left( \alpha H_{1}+\left( 1-\alpha \right)
H_{0}\right) }\right\Vert _{\mathcal{U}}.
\end{equation*}
\end{lemma}

\begin{proof}
(i) By Duhamel's formula, for any continuously differentiable family $%
\left\{ H_{\alpha }\right\} _{\alpha \in \mathbb{R}}\subset \mathcal{U}$ of
self-adjoint elements,
\begin{equation*}
\partial _{\alpha }\left\{ \mathrm{e}^{H_{\alpha }}\right\} =\int_{0}^{1}%
\mathrm{e}^{uH_{\alpha }}\left\{ \partial _{\alpha }H_{\alpha }\right\}
\mathrm{e}^{(1-u)H_{\alpha }}\mathrm{d}u,
\end{equation*}%
which implies that
\begin{equation*}
\partial _{\alpha }\ln \mathrm{tr}\left( C\mathrm{e}^{H_{\alpha }}\right)
=\int_{0}^{1}\frac{\mathrm{tr}\left( C\mathrm{e}^{uH_{\alpha }}\left\{
\partial _{\alpha }H_{\alpha }\right\} \mathrm{e}^{(1-u)H_{\alpha }}\right)
}{\mathrm{tr}\left( C\mathrm{e}^{H_{\alpha }}\right) }\mathrm{d}u.
\end{equation*}%
Using the cyclicity of the trace, we then get
\begin{eqnarray*}
\partial _{\alpha }\ln \mathrm{tr}\left( C\mathrm{e}^{H_{\alpha }}\right)
&=&\int_{0}^{1}\frac{\mathrm{tr}\left( \mathrm{e}^{\frac{H_{\alpha }}{2}}C%
\mathrm{e}^{\frac{H_{\alpha }}{2}}\mathrm{e}^{\left( u-\frac{1}{2}\right)
H_{\alpha }}\left\{ \partial _{\alpha }H_{\alpha }\right\} \mathrm{e}^{(%
\frac{1}{2}-u)H_{\alpha }}\right) }{\mathrm{tr}\left( \mathrm{e}^{\frac{%
H_{\alpha }}{2}}C\mathrm{e}^{\frac{H_{\alpha }}{2}}\right) }\mathrm{d}u \\
&=&\int_{-\frac{1}{2}}^{\frac{1}{2}}\frac{\mathrm{tr}\left( \mathrm{e}^{%
\frac{H_{\alpha }}{2}}C\mathrm{e}^{\frac{H_{\alpha }}{2}}\mathrm{e}%
^{uH_{\alpha }}\left\{ \partial _{\alpha }H_{\alpha }\right\} \mathrm{e}%
^{-uH_{\alpha }}\right) }{\mathrm{tr}\left( \mathrm{e}^{\frac{H_{\alpha }}{2}%
}C\mathrm{e}^{\frac{H_{\alpha }}{2}}\right) }\mathrm{d}u,
\end{eqnarray*}%
which yields (i).\newline
(ii) To prove the second assertion, it suffices to apply Assertion (i) to
the family defined by
\begin{equation*}
H_{\alpha }=H_{0}+\alpha \left( H_{1}-H_{0}\right) ,\qquad \alpha \in \left[
0,1\right] .
\end{equation*}
\end{proof}

Observe that Lemma \ref{lemma:suppor_non_int1} (ii) is proven in \cite[Lemma
3.6]{lenci2005large}. Here, we give a proof of this estimate for
completeness. These Bogoliubov-type inequalities are useful because we deal
with quasi-free dynamics. In this case, we have a very good control on the
norm of
\begin{equation*}
\mathrm{e}^{uH_{\alpha }}\left\{ \partial _{\alpha }H_{\alpha }\right\}
\mathrm{e}^{-uH_{\alpha }},
\end{equation*}%
because $H_{\alpha }$ is a bilinear element, as explained in the next
subsection.

\subsection{Bilinear Elements of CAR Algebra\label{Sect Bilinear Elements}}

Similar to \cite{A68}, bilinear elements are defined as follows:

\begin{definition}[Bilinear elements]
\label{def bilineqr}\mbox{
}\newline
Fix an operator $C\in \mathcal{B}(\mathfrak{h})$ whose range $\mathrm{ran}%
(C) $ is finite dimensional. Given any finite-dimensional subspace $\mathcal{%
H\subset \mathfrak{h}}$, with orthonormal basis $\{\psi _{i}\}_{i\in I}$,
such that $\mathcal{H}\supseteq \mathrm{ran}(C)$ and $\mathcal{H}\supseteq
\mathrm{ran}(C^{\ast })$, we define the bilinear element associated with $C$
to be
\begin{equation*}
\langle \mathrm{A},C\mathrm{A}\rangle \doteq \sum\limits_{i,j\in
I}\left\langle \psi _{i},C\psi _{j}\right\rangle _{\mathfrak{h}}a\left( \psi
_{i}\right) ^{\ast }a\left( \psi _{j}\right) .
\end{equation*}
\end{definition}

\noindent Note that such a finite dimensional $\mathcal{H}$ in this
definition always exists, because
\begin{equation*}
\dim \left( \mathrm{ran}(C)\right) =\dim \left( \mathrm{ran}(C^{\ast
})\right) <\infty ,
\end{equation*}%
and is an invariant space of $C$ containing $\left( \ker (C)\right) ^{\bot }$%
. Hence, $\langle \mathrm{A},C\mathrm{A}\rangle $ does not depend on the
particular choice of $\mathcal{H}$ and its orthonormal basis.

Bilinear elements of $\mathcal{U}$ have adjoints equal to%
\begin{equation}
\langle \mathrm{A},C\mathrm{A}\rangle ^{\ast }=\langle \mathrm{A},C^{\ast }%
\mathrm{A}\rangle ,  \label{selfadjoint bilinear}
\end{equation}%
for any $C\in \mathcal{B}(\mathfrak{h})$ whose range is finite dimensional.
In particular,%
\begin{equation}
\Im \mathrm{m}\left\{ \langle \mathrm{A},C\mathrm{A}\rangle \right\}
=\langle \mathrm{A},\Im \mathrm{m}\left\{ C\right\} \mathrm{A}\rangle ,
\label{Im}
\end{equation}%
where we recall that $\Im \mathrm{m}(A)\in \mathcal{U}$ is the imaginary
part of $A\in \mathcal{U}$, see (\ref{im and real part}). For any $C\in
\mathcal{B}(\mathfrak{h})$ whose range is finite dimensional and any $%
\varphi \in \mathfrak{h}$, note that%
\begin{equation*}
\left[ \left\langle \mathrm{A},C\mathrm{A}\right\rangle ,a\left( \varphi
\right) \right] =-a\left( C^{\ast }\varphi \right) \qquad \text{and}\qquad %
\left[ \left\langle \mathrm{A},C\mathrm{A}\right\rangle ,a\left( \varphi
\right) ^{\ast }\right] =a\left( C\varphi \right) ^{\ast }.
\end{equation*}%
In particular, for any $C_{1},C_{2}\in \mathcal{B}(\mathfrak{h})$ whose
ranges are finite dimensional,%
\begin{equation}
\left[ \left\langle \mathrm{A},C_{1}\mathrm{A}\right\rangle ,\left\langle
\mathrm{A},C_{2}\mathrm{A}\right\rangle \right] =\left\langle \mathrm{A},%
\left[ C_{1},C_{2}\right] \mathrm{A}\right\rangle .  \label{commutator al}
\end{equation}%
Moreover, by (\ref{rescaledbis}), for any $\varphi \in \mathfrak{h}$ and $%
C\in \mathcal{B}(\mathfrak{h})$, whose range is finite dimensional,
\begin{equation}
\mathrm{e}^{\langle \mathrm{A},C\mathrm{A}\rangle }a\left( \varphi \right)
\mathrm{e}^{-\langle \mathrm{A},C\mathrm{A}\rangle }=a\left( \mathrm{e}%
^{-C^{\ast }}\varphi \right) \qquad \text{and}\qquad \mathrm{e}^{\langle
\mathrm{A},C\mathrm{A}\rangle }a\left( \varphi \right) ^{\ast }\mathrm{e}%
^{-\langle \mathrm{A},C\mathrm{A}\rangle }=a\left( \mathrm{e}^{C}\varphi
\right) ^{\ast }.  \label{dynamic bilinear}
\end{equation}

Because of the identities (\ref{dynamic bilinear}), bilinear elements can be
used to represent the dynamics $\{\tau _{t}^{(\omega ,\mathcal{Z})}\}_{t\in
\mathbb{R}}$ for any $\omega \in \Omega $ and $\mathcal{Z}\in \mathfrak{Z}_{%
\text{f}}$. See (\ref{rescaledbis}), replacing $h^{(\omega )}$ with $h_{%
\mathcal{Z}}^{(\omega )}$ (cf. (\ref{h finite f})), and observe that the
range of $h_{\mathcal{Z}}^{(\omega )}\in \mathcal{B}(\mathfrak{h})$ is
finite dimensional whenever $\mathcal{Z}\in \mathfrak{Z}_{\text{f}}$.
Additionally, by using the tracial state $\mathrm{tr}\in \mathcal{U}^{\ast }$%
, i.e., the quasi-free state satisfying (\ref{2-point correlation function})
for $\beta =0$, the corresponding KMS\ state defined by (\ref{2-point
correlation function}) by replacing $h^{(\omega )}$ in this equation with $%
h_{\mathcal{Z}}^{(\omega )}$ (see (\ref{h finite f})) is explicitly given by%
\begin{equation}
\varrho _{\mathcal{Z}}^{(\omega )}(B)=\frac{\mathrm{tr}\left( B\mathrm{e}%
^{-\beta \langle \mathrm{A},h_{\mathcal{Z}}^{(\omega )}\mathrm{A}\rangle
}\right) }{\mathrm{tr}\left( \mathrm{e}^{-\beta \langle \mathrm{A},h_{%
\mathcal{Z}}^{(\omega )}\mathrm{A}\rangle }\right) }\ ,\qquad B\in \mathcal{U%
},  \label{quasi free state}
\end{equation}%
for any $\omega \in \Omega $, $\lambda ,\vartheta \in \mathbb{R}_{0}^{+}$, $%
\beta \in \mathbb{R}^{+}$ and $\mathcal{Z}\in \mathfrak{Z}_{\text{f}}$.

We conclude now by an additional observation used later to control quantum
fluctuations:

\begin{lemma}
\label{lemma cool}\mbox{ }\newline
For any self-adjoint operators $C_{1},C_{2}\in \mathcal{B}(\mathfrak{h})$
whose ranges are finite dimensional, let $C\doteq \ln \left( \mathrm{e}%
^{C_{2}}\mathrm{e}^{C_{1}}\mathrm{e}^{C_{2}}\right) $. Then,%
\begin{equation*}
\mathrm{ran}(C)\subset \mathrm{lin}\left\{ \mathrm{ran}(C_{1})\cup \mathrm{%
ran}(C_{2})\right\}
\end{equation*}%
and there is a constant $D\in \mathbb{R}$ such that%
\begin{equation*}
\mathrm{e}^{\langle \mathrm{A},C_{2}\mathrm{A}\rangle }\mathrm{e}^{\langle
\mathrm{A},C_{1}\mathrm{A}\rangle }\mathrm{e}^{\langle \mathrm{A},C_{2}%
\mathrm{A}\rangle }=\mathrm{e}^{\langle \mathrm{A},C\mathrm{A}\rangle +D%
\mathfrak{1}}.
\end{equation*}
\end{lemma}

\begin{proof}
Fix all parameters of the lemma. We give the proof in two steps:\medskip

\noindent \underline{Step 1:} Let
\begin{equation*}
\mathfrak{h}_{0}\doteq \mathrm{lin}\left\{ \mathrm{ran}(C_{1})\cup \mathrm{%
ran}(C_{2})\right\}
\end{equation*}%
and $\mathcal{U}_{\mathfrak{h}_{0}}\subset \mathcal{U}\equiv \mathcal{U}_{%
\mathfrak{h}}$ be the (finite dimensional) CAR $C^{\ast }$-subalgebra
generated by the identity $\mathfrak{1}$ and $\{a(\varphi )\}_{\varphi \in
\mathfrak{h}_{0}}$. Take two strictly positive elements $M_{1},M_{2}$ of $%
\mathcal{U}_{\mathfrak{h}_{0}}$ satisfying the conditions
\begin{equation*}
M_{1}a(\varphi )M_{1}^{-1}=M_{2}a(\varphi )M_{2}^{-1}\qquad \text{and}\qquad
M_{1}a(\varphi )^{\ast }M_{1}^{-1}=M_{2}a(\varphi )^{\ast }M_{2}^{-1}
\end{equation*}%
for any $\varphi \in \mathfrak{h}_{0}$. From this we conclude that
\begin{equation*}
M_{1}AM_{1}^{-1}=M_{2}AM_{2}^{-1},\qquad A\in \mathcal{U}_{\mathfrak{h}_{0}},
\end{equation*}%
because all elements of $\mathcal{U}_{\mathfrak{h}_{0}}$ are polynomials in $%
\{a(\varphi ),a(\varphi )^{\ast }\}_{\varphi \in \mathfrak{h}_{0}}$, by
definition of $\mathcal{U}_{\mathfrak{h}_{0}}$ and finite dimensionality of $%
\mathfrak{h}_{0}$. In particular, by choosing, respectively, $A=M_{2}^{-1}$
and $A=M_{2}^{-1}BM_{2}$ for $B\in \mathcal{U}_{\mathfrak{h}_{0}}$, it
follows that%
\begin{equation*}
M_{1}M_{2}^{-1}=M_{2}^{-1}M_{1}\qquad \text{and}\qquad
M_{1}M_{2}^{-1}B=BM_{1}M_{2}^{-1}.
\end{equation*}%
Hence, since any element of $\mathcal{U}_{\mathfrak{h}_{0}}$ commuting with
all elements of $\mathcal{U}_{\mathfrak{h}_{0}}$ is a multiple of the
identity, there is $D\in \mathbb{C}$ such that
\begin{equation*}
M_{1}M_{2}^{-1}=M_{2}^{-1}M_{1}=D\mathfrak{1}.
\end{equation*}%
The constant $D$ is non-zero because $M_{1},M_{2}$ are assumed to be
invertible. In fact, $M_{1}=DM_{2}$ with $D>0$ because $M_{1},M_{2}>0$%
.\medskip

\noindent \underline{Step 2:} Observe that $\mathrm{e}^{C_{2}}\mathrm{e}%
^{C_{1}}\mathrm{e}^{C_{2}}>0$ because $C_{1},C_{2}$ are both self-adjoint
operators. In particular, $C\doteq \ln \left( \mathrm{e}^{C_{2}}\mathrm{e}%
^{C_{1}}\mathrm{e}^{C_{2}}\right) $ is well-defined as a bounded
self-adjoint operator acting on $\mathfrak{h}$ with $\mathrm{ran}(C)\subset
\mathfrak{h}_{0}$. Using (\ref{dynamic bilinear}), we obtain that
\begin{equation*}
\mathrm{e}^{\langle \mathrm{A},C\mathrm{A}\rangle }a(\varphi )\mathrm{e}%
^{-\langle \mathrm{A},C\mathrm{A}\rangle }=\mathrm{e}^{\langle \mathrm{A}%
,C_{2}\mathrm{A}\rangle }\mathrm{e}^{\langle \mathrm{A},C_{1}\mathrm{A}%
\rangle }\mathrm{e}^{\langle \mathrm{A},C_{2}\mathrm{A}\rangle }a(\varphi )%
\mathrm{e}^{-\langle \mathrm{A},C_{2}\mathrm{A}\rangle }\mathrm{e}^{-\langle
\mathrm{A},C_{1}\mathrm{A}\rangle }\mathrm{e}^{-\langle \mathrm{A},C_{2}%
\mathrm{A}\rangle }
\end{equation*}%
and
\begin{equation*}
\mathrm{e}^{\langle \mathrm{A},C\mathrm{A}\rangle }a(\varphi )^{\ast }%
\mathrm{e}^{-\langle \mathrm{A},C\mathrm{A}\rangle }=\mathrm{e}^{\langle
\mathrm{A},C_{2}\mathrm{A}\rangle }\mathrm{e}^{\langle \mathrm{A},C_{1}%
\mathrm{A}\rangle }\mathrm{e}^{\langle \mathrm{A},C_{2}\mathrm{A}\rangle
}a(\varphi )^{\ast }\mathrm{e}^{-\langle \mathrm{A},C_{2}\mathrm{A}\rangle }%
\mathrm{e}^{-\langle \mathrm{A},C_{1}\mathrm{A}\rangle }\mathrm{e}^{-\langle
\mathrm{A},C_{2}\mathrm{A}\rangle }.
\end{equation*}%
By Step 1, the assertion follows.
\end{proof}

\subsection{Bilinear Elements Associated with Currents\label{Sectino tech3}}

For simplicity, below we fix $\vec{w}\in {\mathbb{R}}^{d}$ with $\left\Vert
\vec{w}\right\Vert _{\mathbb{R}^{d}}=1$, and $\eta ,\mu \in \mathbb{R}^{+}$
once for all. For any $\mathcal{E}\in C_{0}^{0}(\mathbb{R};\mathbb{R}^{d})$,
any collection $\mathcal{Z}^{(\tau )}\in \mathfrak{Z}$, $\mathcal{Z}\in
\mathfrak{Z}_{\text{f}}$, and $\lambda ,\vartheta \in \mathbb{R}_{0}^{+}$, $%
\omega \in \Omega $, we define the observables%
\begin{eqnarray}
\mathfrak{K}_{\mathcal{Z},\mathcal{Z}^{(\tau )}}^{(\omega ,\mathcal{E})}
&\doteq &\underset{k,q=1}{\sum^{d}}w_{k}\sum_{Z\in \mathcal{Z}}\underset{%
x,y,x+e_{k},y+e_{q}\in Z}{\sum }\int_{-\infty }^{0}\left\{ \mathcal{E}\left(
\alpha \right) \right\} _{q}\mathrm{d}\alpha \int\nolimits_{0}^{-\alpha }%
\mathrm{d}s\ i[\tau _{-s}^{(\omega ,\mathcal{Z}^{(\tau )})}(I_{\left(
y+e_{q},y\right) }^{(\omega )}),I_{\left( x+e_{k},x\right) }^{(\omega )}]
\notag \\
&&+2\underset{k=1}{\sum^{d}}w_{k}\sum_{Z\in \mathcal{Z}}\underset{%
x,x+e_{k}\in Z}{\sum }\left( \int_{-\infty }^{0}\left\{ \mathcal{E}\left(
\alpha \right) \right\} _{q}\mathrm{d}\alpha \right) \Re \mathrm{e}\left(
\langle \mathfrak{e}_{x+e_{k}},\Delta _{\omega ,\vartheta }\mathfrak{e}%
_{x}\rangle a(\mathfrak{e}_{x+e_{k}})^{\ast }a(\mathfrak{e}_{x})\right) ,
\notag \\
&&  \label{eq:par-curr-sum-gen}
\end{eqnarray}%
where we recall that $\Re \mathrm{e}(A)\in \mathcal{U}$ is the real part of $%
A\in \mathcal{U}$, see (\ref{im and real part}). Note that
\begin{equation*}
\mathfrak{K}_{\{\Lambda \},\{\mathbb{Z}^{d}\}}^{(\omega ,\mathcal{E}%
)}=\left\vert \Lambda \right\vert \mathbb{I}_{\Lambda }^{(\omega ,\mathcal{E}%
)},\qquad \Lambda \in \mathcal{P}_{\text{f}}(\mathbb{Z}^{d}),
\end{equation*}%
is a current observable (cf. (\ref{current})). These observables are
bilinear elements (Definition \ref{def bilineqr}):\medskip

\noindent \underline{(i) Single-hopping operators:} For any $x\in \mathbb{Z}%
^{d}$, the shift operator $s_{x}\in \mathcal{B}(\mathfrak{h})$ is defined by
\begin{equation}
\left( s_{x}\psi \right) \left( y\right) \doteq \psi \left( x+y\right)
,\qquad y\in \mathbb{Z}^{d},  \label{shift}
\end{equation}%
Note that $s_{x}^{\ast }=s_{-x}=s_{x}^{-1}$ for any $x\in \mathbb{Z}^{d}$.
Then, for any $\omega \in \Omega $ and $\vartheta \in \mathbb{R}_{0}^{+}$,
the single-hopping operators are
\begin{equation}
S_{x,y}^{(\omega )}\doteq \langle \mathfrak{e}_{x},\Delta _{\omega
,\vartheta }\mathfrak{e}_{y}\rangle _{\mathfrak{h}}P_{\left\{ x\right\}
}s_{x-y}P_{\left\{ y\right\} },\qquad x,y\in \mathbb{Z}^{d},  \label{shift2}
\end{equation}%
where $P_{\left\{ x\right\} }$ is the orthogonal projection defined by (\ref%
{orthogonal projection}) for $\Lambda =\left\{ x\right\} $. Observe that
\begin{equation*}
\left\langle \mathrm{A},S_{x,y}^{(\omega )}\mathrm{A}\right\rangle =\langle
\mathfrak{e}_{x},\Delta _{\omega ,\vartheta }\mathfrak{e}_{y}\rangle _{%
\mathfrak{h}}a(\mathfrak{e}_{x})^{\ast }a(\mathfrak{e}_{y}),\qquad x,y\in
\mathbb{Z}^{d}.
\end{equation*}%
Similarly, the paramagnetic current observables defined by (\ref{current
observable}) equal
\begin{equation}
I_{(x,y)}^{(\omega )}=-2\langle \mathrm{A},\Im \mathrm{m}\{S_{x,y}^{(\omega
)}\}\mathrm{A}\rangle ,\qquad x,y\in \mathbb{Z}^{d},  \label{dekldjfskldjf}
\end{equation}%
for any $\omega \in \Omega $ and $\vartheta \in \mathbb{R}_{0}^{+}$. Compare
with (\ref{Im}). \medskip

\noindent \underline{(ii) Local current observables:} By (\ref{commutator al}%
), for any $\mathcal{E}\in C_{0}^{0}(\mathbb{R};\mathbb{R}^{d})$, any
collection $\mathcal{Z}^{(\tau )}\in \mathfrak{Z}$, $\mathcal{Z}\in
\mathfrak{Z}_{\text{f}}$, and $\lambda ,\vartheta \in \mathbb{R}_{0}^{+}$, $%
\omega \in \Omega $,%
\begin{equation}
\mathfrak{K}_{\mathcal{Z},\mathcal{Z}^{(\tau )}}^{(\omega ,\mathcal{E}%
)}=\left\langle \mathrm{A},K_{\mathcal{Z},\mathcal{Z}^{(\tau )}}^{(\omega ,%
\mathcal{E})}\mathrm{A}\right\rangle ,  \label{bilinear idiot}
\end{equation}%
where%
\begin{eqnarray}
K_{\mathcal{Z},\mathcal{Z}^{(\tau )}}^{(\omega ,\mathcal{E})} &\doteq &4%
\underset{k,q=1}{\sum^{d}}w_{k}\sum_{Z\in \mathcal{Z}}\underset{%
x,y,x+e_{k},y+e_{q}\in Z}{\sum }\int_{-\infty }^{0}\left\{ \mathcal{E}\left(
\alpha \right) \right\} _{q}\mathrm{d}\alpha  \notag \\
&&\qquad \qquad \qquad \int\nolimits_{0}^{-\alpha }\mathrm{d}s\ i\left[
\mathrm{e}^{-ish_{\mathcal{Z}^{(\tau )}}^{(\omega )}}\Im \mathrm{m}%
\{S_{y+e_{q},y}^{(\omega )}\}\mathrm{e}^{ish_{\mathcal{Z}^{(\tau
)}}^{(\omega )}},\Im \mathrm{m}\{S_{x+e_{k},x}^{(\omega )}\}\right]  \notag
\\
&&+2\underset{k=1}{\sum^{d}}w_{k}\sum_{Z\in \mathcal{Z}}\underset{%
x,x+e_{k}\in Z}{\sum }\left( \int_{-\infty }^{0}\left\{ \mathcal{E}\left(
\alpha \right) \right\} _{q}\mathrm{d}\alpha \right) \Re \mathrm{e}%
\{S_{x+e_{k},x}^{(\omega )}\}  \label{definition K}
\end{eqnarray}%
is an operator acting on $\mathfrak{h}$ whose range is finite dimensional.
This one-particle operator satisfies the following decay bounds:

\begin{lemma}[Decay of local currents]
\label{def bilineqr copy(1)}\mbox{
}\newline
For any $\mathcal{E}\in C_{0}^{0}(\mathbb{R};\mathbb{R}^{d})$, $\lambda
,\vartheta \in \mathbb{R}_{0}^{+}$, $\omega \in \Omega $, $x,y\in \mathbb{Z}%
^{d}$, and two collections $\mathcal{Z}\in \mathfrak{Z}_{\text{f}}$ and $%
\mathcal{Z}^{(\tau )}\in \mathfrak{Z}_{\text{f}}$,%
\begin{eqnarray*}
\left\vert \left\langle \mathfrak{e}_{x},K_{\mathcal{Z},\mathcal{Z}^{(\tau
)}}^{(\omega ,\mathcal{E})}\mathfrak{e}_{y}\right\rangle _{\mathfrak{h}%
}\right\vert &\leq &D_{\ref{def bilineqr copy(1)}}\left( \int_{\mathbb{R}%
}\left\Vert \mathcal{E}\left( \alpha \right) \right\Vert _{\mathbb{R}^{d}}%
\mathrm{e}^{2\left\vert \alpha \eta \right\vert }\mathrm{d}\alpha \right)
\left( \mathrm{e}^{-\mu _{\eta }|x-y|}+\eta \delta _{1,\left\vert
x-y\right\vert }\right) , \\
\frac{1}{|\cup \mathcal{Z}|}\sum\limits_{x,y\in \mathbb{Z}^{d}}\left\vert
\left\langle \mathfrak{e}_{x},K_{\mathcal{Z},\mathcal{Z}^{(\tau )}}^{(\omega
,\mathcal{E})}\mathfrak{e}_{y}\right\rangle _{\mathfrak{h}}\right\vert &\leq
&D_{\ref{def bilineqr copy(1)}}\left( \int_{\mathbb{R}}\left\Vert \mathcal{E}%
\left( \alpha \right) \right\Vert _{\mathbb{R}^{d}}\mathrm{e}^{2\left\vert
\alpha \eta \right\vert }\mathrm{d}\alpha \right) \sum_{z\in \mathbb{Z}^{d}}%
\mathrm{e}^{-\mu _{\eta }|z|}\left( 1+\eta \right) ,
\end{eqnarray*}%
where%
\begin{equation*}
D_{\ref{def bilineqr copy(1)}}\doteq 4d\eta ^{-1}\times 36^{2}\left(
1+\vartheta \right) ^{2}\sum_{z\in \mathbb{Z}^{d}}\mathrm{e}^{2\mu _{\eta
}\left( 1-|z|\right) }<\infty .
\end{equation*}%
Recall that $\mu _{\eta }$ is defined by (\ref{Combes-ThomasCombes-Thomasbis}%
).
\end{lemma}

\begin{proof}
Fix the parameters of the lemma. By (\ref{Combes-ThomasCombes-Thomas}), note
that for any $z_{1},z_{2},x,y\in \mathbb{Z}^{d}$, $\omega \in \Omega $, $%
\vartheta \in \mathbb{R}_{0}^{+}$ and $s\in \mathbb{R}$,
\begin{equation}
\left\vert \left\langle \mathfrak{e}_{x},\mathrm{e}^{-ish_{\mathcal{Z}%
^{(\tau )}}^{(\omega )}}S_{z_{2}+e_{q},z_{2}}^{(\omega )}\mathrm{e}^{ish_{%
\mathcal{Z}^{(\tau )}}^{(\omega )}}S_{z_{1}+e_{k},z_{1}}^{(\omega )}%
\mathfrak{e}_{y}\right\rangle _{\mathfrak{h}}\right\vert \leq 36^{2}\left(
1+\vartheta \right) ^{2}\mathrm{e}^{2\left\vert s\eta \right\vert -2\mu
_{\eta }\left( |x-z_{2}-e_{q}|+|y-z_{2}+e_{k}|\right) }\delta _{y,z_{1}}.
\label{idiot1}
\end{equation}%
By the Cauchy-Schwarz and triangle inequalities, observe also that%
\begin{equation}
\sum_{z\in \mathbb{Z}^{d}}\mathrm{e}^{-2\mu _{\eta }\left(
|x-z|+|y-z|\right) }\leq \mathrm{e}^{-\mu _{\eta }|x-y|}\sum_{z\in \mathbb{Z}%
^{d}}\mathrm{e}^{-\mu _{\eta }\left( |x-z|+|y-z|\right) }\leq \mathrm{e}%
^{-\mu _{\eta }|x-y|}\left( \sum_{z\in \mathbb{Z}^{d}}\mathrm{e}^{-2\mu
_{\eta }|z|}\right) .  \label{boudn trivial}
\end{equation}%
From (\ref{idiot1})-(\ref{boudn trivial}), we obtain the bound
\begin{eqnarray}
&&\sum_{Z\in \mathcal{Z}}\underset{z_{1},z_{2},z_{1}+e_{k},z_{2}+e_{q}\in Z}{%
\sum }\left\vert \left\langle \mathfrak{e}_{x},\mathrm{e}^{-ish_{\mathcal{Z}%
^{(\tau )}}^{(\omega )}}S_{z_{2}+e_{q},z_{2}}^{(\omega )}\mathrm{e}^{ish_{%
\mathcal{Z}^{(\tau )}}^{(\omega )}}S_{z_{1}+e_{k},z_{1}}^{(\omega )}%
\mathfrak{e}_{y}\right\rangle _{\mathfrak{h}}\right\vert   \notag \\
&\leq &36^{2}\left( 1+\vartheta \right) ^{2}\mathrm{e}^{2\left\vert s\eta
\right\vert -\mu _{\eta }|x-y|}\left( \sum_{z\in \mathbb{Z}^{d}}\mathrm{e}%
^{2\mu _{\eta }\left( 1-|z|\right) }\right) ,  \label{idiot2}
\end{eqnarray}%
using that $|z-e_{k}|\geq |z|-1$ for any $z\in \mathbb{Z}^{d}$ and $k\in
\{1,\ldots ,d\}$. The other terms computed from (\ref{definition K}) are
estimated in the same way. We omit the details. This yields the first bound
of the lemma. The second estimate is also proven in the same way.
\end{proof}

It is convenient to introduce at this point the notation
\begin{equation}
\mathcal{\partial }_{\Lambda }(\tilde{\Lambda})\doteq \left\{ \left\{
x,y\right\} \subset \Lambda \colon \left\vert x-y\right\vert =1,\ \left\{
x,y\right\} \cap \tilde{\Lambda}\neq \emptyset \text{ and }\left\{
x,y\right\} \cap \tilde{\Lambda}^{c}\neq \emptyset \right\}
\label{notation1}
\end{equation}%
for any set $\tilde{\Lambda}\subset \Lambda \subset \mathbb{Z}^{d}$ with
complement $\tilde{\Lambda}^{c}\doteq \mathbb{Z}^{d}\backslash \tilde{\Lambda%
}$, while, for any $\mathcal{Z}\in \mathfrak{Z}$ such that $\cup \mathcal{Z}%
\subset \Lambda $,
\begin{equation*}
\mathcal{\partial }_{\Lambda }(\mathcal{Z})\doteq \left\{ \mathcal{\partial }%
_{\Lambda }(Z):Z\in \mathcal{Z}\right\} .
\end{equation*}%
Then, the one-particle operators (\ref{definition K}) also satisfy the
following bounds:

\begin{lemma}[Box decomposition of local currents - I]
\label{def bilineqr copy(2)}\mbox{
}\newline
For any $\mathcal{E}\in C_{0}^{0}(\mathbb{R};\mathbb{R}^{d})$, $\Lambda ,%
\tilde{\Lambda}\in \mathcal{P}_{\text{f}}(\mathbb{Z}^{d})$, $\lambda
,\vartheta \in \mathbb{R}_{0}^{+}$, $\omega \in \Omega $, and $\mathcal{Z}%
\in \mathfrak{Z}_{\text{f}}$ with $\cup \mathcal{Z}\subset \tilde{\Lambda}$,
\begin{eqnarray*}
&&\sum\limits_{x,y\in \mathbb{Z}^{d}}\left\vert \left\langle \mathfrak{e}%
_{x},\left( K_{\{\Lambda \},\{\tilde{\Lambda}\}}^{(\omega ,\mathcal{E}%
)}-K_{\{\Lambda \},\mathcal{Z}}^{(\omega ,\mathcal{E})}\right) \mathfrak{e}%
_{y}\right\rangle _{\mathfrak{h}}\right\vert \\
&\leq &D_{\ref{def bilineqr copy(2)}}\left( \int_{\mathbb{R}}\left\Vert
\mathcal{E}\left( \alpha \right) \right\Vert _{\mathbb{R}^{d}}\alpha ^{2}%
\mathrm{e}^{2\left\vert \alpha \eta \right\vert }\mathrm{d}\alpha \right)
\left( \sum\limits_{x\in \Lambda }\sum_{z\in \tilde{\Lambda}\backslash \cup
\mathcal{Z}}\mathrm{e}^{-\mu _{\eta }|x-z|}+\sum_{z\in \mathbb{Z}^{d}}%
\mathrm{e}^{-\mu _{\eta }|z|}\sum_{x\in \cup \mathcal{\partial }_{\tilde{%
\Lambda}}(\mathcal{Z})}1\right) ,
\end{eqnarray*}%
where%
\begin{equation*}
D_{\ref{def bilineqr copy(2)}}\doteq 8\times 36^{4}\left( 1+\vartheta
\right) ^{3}\left( 4d+\lambda \right) \mathrm{e}^{3\mu _{\eta }}\left(
\sum_{z\in \mathbb{Z}^{d}}\mathrm{e}^{-\mu _{\eta }|z|}\right) ^{3}<\infty .
\end{equation*}
\end{lemma}

\begin{proof}
Fix all parameters of the lemma. Let%
\begin{equation*}
C_{\mathcal{Z}}^{(\omega )}(z_{1},z_{2},k,q)=\int\nolimits_{0}^{-\alpha }%
\mathrm{d}s\ i\left[ \mathrm{e}^{-ish_{\mathcal{Z}}^{(\omega )}}\Im \mathrm{m%
}\{S_{z_{2}+e_{q},z_{2}}^{(\omega )}\}\mathrm{e}^{ish_{\mathcal{Z}}^{(\omega
)}},\Im \mathrm{m}\{S_{z_{1}+e_{k},z_{1}}^{(\omega )}\}\right]
\end{equation*}%
for any $z_{1},z_{2}\in \mathbb{Z}^{d}$ and $k,q\in \{1,\ldots ,d\}$. By
Duhamel's formula,
\begin{eqnarray*}
&&\mathrm{e}^{-ish_{\mathcal{\{}\tilde{\Lambda}\mathcal{\}}}^{(\omega )}}A%
\mathrm{e}^{ish_{\mathcal{\{}\tilde{\Lambda}\mathcal{\}}}^{(\omega )}}-%
\mathrm{e}^{-ish_{\mathcal{Z}}^{(\omega )}}A\mathrm{e}^{ish_{\mathcal{Z}%
}^{(\omega )}} \\
&=&-i\int_{0}^{s}\mathrm{e}^{-i\left( s-u\right) h_{\mathcal{Z}}^{(\omega )}}%
\left[ h_{\mathcal{\{}\tilde{\Lambda}\mathcal{\}}}^{(\omega )}-h_{\mathcal{Z}%
}^{(\omega )},\mathrm{e}^{-iuh_{\mathcal{\{}\tilde{\Lambda}\mathcal{\}}%
}^{(\omega )}}A\mathrm{e}^{iuh_{\mathcal{\{}\tilde{\Lambda}\mathcal{\}}%
}^{(\omega )}}\right] \mathrm{e}^{i\left( s-u\right) h_{\mathcal{Z}%
}^{(\omega )}}\mathrm{d}u
\end{eqnarray*}%
and hence, for any $z_{1},z_{2}\in \mathbb{Z}^{d}$ and $k,q\in \{1,\ldots
,d\}$,
\begin{multline*}
C_{\mathcal{\{}\tilde{\Lambda}\mathcal{\}}}^{(\omega )}(z_{1},z_{2},k,q)-C_{%
\mathcal{Z}}^{(\omega )}(z_{1},z_{2},k,q)=4\int\nolimits_{0}^{\alpha }%
\mathrm{d}s\int_{0}^{s}\mathrm{d}u \\
\left[ \mathrm{e}^{-i\left( s-u\right) h_{\mathcal{Z}}^{(\omega )}}\left[ h_{%
\mathcal{\{}\tilde{\Lambda}\mathcal{\}}}^{(\omega )}-h_{\mathcal{Z}%
}^{(\omega )},\mathrm{e}^{-iuh_{\mathcal{\{}\tilde{\Lambda}\mathcal{\}}%
}^{(\omega )}}\Im \mathrm{m}\{S_{z_{2}+e_{q},z_{2}}^{(\omega )}\}\mathrm{e}%
^{iuh_{\mathcal{\{}\tilde{\Lambda}\mathcal{\}}}^{(\omega )}}\right] \mathrm{e%
}^{i\left( s-u\right) h_{\mathcal{Z}}^{(\omega )}},\Im \mathrm{m}%
\{S_{z_{1}+e_{k},z_{1}}^{(\omega )}\}\right] .
\end{multline*}%
By developing the commutators and $\Im \mathrm{m}\{\cdot \}$ we get sixteen
terms:%
\begin{equation}
C_{\mathcal{\{}\tilde{\Lambda}\mathcal{\}}}^{(\omega )}(z_{1},z_{2},k,q)-C_{%
\mathcal{Z}}^{(\omega )}(z_{1},z_{2},k,q)=\int\nolimits_{0}^{\alpha }\mathrm{%
d}s\int_{0}^{s}\mathrm{d}u\sum\limits_{j=1}^{16}\mathbf{X}_{j}\left(
s,u,z_{1},z_{2}\right) ,  \label{term other}
\end{equation}%
where, for instance,
\begin{equation}
\mathbf{X}_{1}\left( s,u,z_{1},z_{2}\right) \doteq \mathrm{e}^{-i\left(
s-u\right) h_{\mathcal{Z}}^{(\omega )}}\left( h_{\mathcal{\{}\tilde{\Lambda}%
\mathcal{\}}}^{(\omega )}-h_{\mathcal{Z}}^{(\omega )}\right) \mathrm{e}%
^{-iuh_{\mathcal{\{}\tilde{\Lambda}\mathcal{\}}}^{(\omega
)}}S_{z_{2}+e_{q},z_{2}}\mathrm{e}^{iuh_{\mathcal{\{}\tilde{\Lambda}\mathcal{%
\}}}^{(\omega )}}\mathrm{e}^{i\left( s-u\right) h_{\mathcal{Z}}^{(\omega
)}}S_{z_{1}+e_{k},z_{1}}.  \label{X1}
\end{equation}%
Since $\cup \mathcal{Z}\subset \tilde{\Lambda}$, note that%
\begin{eqnarray}
h_{\{\tilde{\Lambda}\}}^{(\omega )}-h_{\mathcal{Z}}^{(\omega )}
&=&\sum_{z_{3},z_{4}\in \tilde{\Lambda}\backslash \cup \mathcal{Z}\colon \
\left\vert z_{3}-z_{4}\right\vert =1}S_{z_{3},z_{4}}^{(\omega )}+\sum_{Z\in
\mathcal{Z}}\sum_{\left\{ z_{3},z_{4}\right\} \in \mathcal{\partial }_{%
\tilde{\Lambda}}(Z)}\left( S_{z_{3},z_{4}}^{(\omega
)}+S_{z_{4},z_{3}}^{(\omega )}\right)   \notag \\
&&+\sum_{z_{3}\in \tilde{\Lambda}\backslash \cup \mathcal{Z}}\lambda \omega
_{1}\left( z_{3}\right) S_{z_{3},z_{3}}^{(\omega )}.  \label{X2}
\end{eqnarray}%
Meanwhile, for any $z_{1},z_{2},z_{3},z_{4},x,y\in \mathbb{Z}^{d}$ with $%
\left\vert z_{3}-z_{4}\right\vert \leq 1$, and real numbers $s\geq u\geq 0$,
we infer from (\ref{Combes-ThomasCombes-Thomas}) and (\ref{boudn trivial})
that%
\begin{multline*}
\left\vert \left\langle \mathfrak{e}_{x},\mathrm{e}^{-i\left( s-u\right) h_{%
\mathcal{Z}}^{(\omega )}}S_{z_{3},z_{4}}^{\left( \omega \right) }\mathrm{e}%
^{-iuh_{\mathcal{\{}\tilde{\Lambda}\mathcal{\}}}^{(\omega
)}}S_{z_{2}+e_{q},z_{2}}^{(\omega )}\mathrm{e}^{iuh_{\mathcal{\{}\tilde{%
\Lambda}\mathcal{\}}}^{(\omega )}}\mathrm{e}^{i\left( s-u\right) h_{\mathcal{%
Z}}^{(\omega )}}S_{z_{1}+e_{k},z_{1}}^{(\omega )}\mathfrak{e}%
_{y}\right\rangle _{\mathfrak{h}}\right\vert  \\
\leq 36^{4}\left( 1+\vartheta \right) ^{3}\mathrm{e}^{2\left\vert s\eta
\right\vert +3\mu _{\eta }}\left( \sum_{z\in \mathbb{Z}^{d}}\mathrm{e}%
^{-2\mu _{\eta }|z|}\right) \delta _{z_{1},y}\mathrm{e}^{-\mu _{\eta }\left(
|z_{2}-y|+|x-z_{3}|+|z_{3}-z_{2}|\right) }.
\end{multline*}%
By (\ref{X1})-(\ref{X2}), for any $\alpha \geq 0$, it follows that%
\begin{eqnarray*}
&&\sum\limits_{x,y\in \mathbb{Z}^{d}}\underset{%
z_{1},z_{2},z_{1}+e_{k},z_{2}+e_{q}\in \Lambda }{\sum }\int\nolimits_{0}^{%
\alpha }\mathrm{d}s\int_{0}^{s}\mathrm{d}u\left\vert \left\langle \mathfrak{e%
}_{x},\mathbf{X}_{1}\left( s,u,z_{1},z_{2}\right) \mathfrak{e}%
_{y}\right\rangle _{\mathfrak{h}}\right\vert  \\
&\leq &\frac{36^{4}}{2}\left( 1+\vartheta \right) ^{3}\left( 4d+\lambda
\right) \alpha ^{2}\mathrm{e}^{2\left\vert \alpha \eta \right\vert +3\mu
_{\eta }}\left( \sum_{z\in \mathbb{Z}^{d}}\mathrm{e}^{-\mu _{\eta
}|z|}\right) ^{3} \\
&&\times \left( \sum\limits_{x\in \Lambda }\sum_{z\in \tilde{\Lambda}%
\backslash \cup \mathcal{Z}}\mathrm{e}^{-\mu _{\eta }|x-z|}+\sum_{z\in
\mathbb{Z}^{d}}\mathrm{e}^{-\mu _{\eta }|z|}\sum_{x\in \cup \mathcal{%
\partial }_{\tilde{\Lambda}}(\mathcal{Z})}1\right) .
\end{eqnarray*}%
The fifteen other terms $\mathbf{X}_{j}$ in (\ref{term other}) satisfy the
same bound. By (\ref{definition K}), the assertion follows for any $\mathcal{%
E}\in C_{0}^{0}(\mathbb{R};\mathbb{R}^{d})$.
\end{proof}

\begin{lemma}[Box decomposition of local currents - II]
\label{def bilineqr copy(3)}\mbox{
}\newline
For any $\mathcal{E}\in C_{0}^{0}(\mathbb{R};\mathbb{R}^{d})$, $\Lambda \in
\mathcal{P}_{\text{f}}(\mathbb{Z}^{d})$, $\lambda ,\vartheta \in \mathbb{R}%
_{0}^{+}$, $\omega \in \Omega $, $\mathcal{Z}_{\tau }\in \mathfrak{Z}$, and $%
\mathcal{Z}\in \mathfrak{Z}_{\text{f}}$ with $\cup \mathcal{Z}\subset
\Lambda $,
\begin{equation*}
\sum\limits_{x,y\in \mathbb{Z}^{d}}\left\vert \left\langle \mathfrak{e}%
_{x},\left( K_{\{\Lambda \},\mathcal{Z}_{\tau }}^{(\omega ,\mathcal{E})}-K_{%
\mathcal{Z},\mathcal{Z}_{\tau }}^{(\omega ,\mathcal{E})}\right) \mathfrak{e}%
_{y}\right\rangle _{\mathfrak{h}}\right\vert \leq D_{\ref{def bilineqr
copy(3)}}\left( \int_{\mathbb{R}}\left\Vert \mathcal{E}\left( \alpha \right)
\right\Vert _{\mathbb{R}^{d}}^{2}\left\vert \alpha \right\vert \mathrm{e}%
^{2\left\vert \alpha \eta \right\vert }\mathrm{d}\alpha \right) \underset{%
z\in \left( \Lambda \backslash \cup \mathcal{Z}\right) \cup \left( \cup
\mathcal{\partial }_{\Lambda }(\mathcal{Z})\right) }{\sum }1,
\end{equation*}%
where%
\begin{equation*}
D_{\ref{def bilineqr copy(3)}}\doteq 16\times 36^{2}\left( 1+\vartheta
\right) ^{2}d\mathrm{e}^{4\mu _{\eta }}\left( \sum\limits_{z\in \mathbb{Z}%
^{d}}\mathrm{e}^{-2\mu _{\eta }|z|}\right) ^{2}+d\left( 1+\vartheta \right)
<\infty .
\end{equation*}
\end{lemma}

\begin{proof}
Fix all parameters of the lemma. By combining (\ref{idiot1}) with direct
estimates we observe that%
\begin{eqnarray}
&&\sum\limits_{x,y\in \mathbb{Z}^{d}}\underset{%
z_{1},z_{2},z_{1}+e_{k},z_{2}+e_{q}\in \Lambda }{\sum }\left\vert
\left\langle \mathfrak{e}_{x},\mathrm{e}^{-ish_{\mathcal{Z}^{(\tau
)}}^{(\omega )}}S_{z_{2}+e_{q},z_{2}}^{(\omega )}\mathrm{e}^{ish_{\mathcal{Z}%
^{(\tau )}}^{(\omega )}}S_{z_{1}+e_{k},z_{1}}^{(\omega )}\mathfrak{e}%
_{y}\right\rangle _{\mathfrak{h}}\right\vert  \notag \\
&&-\sum\limits_{x,y\in \mathbb{Z}^{d}}\sum_{Z\in \mathcal{Z}}\underset{%
z_{1},z_{2},z_{1}+e_{k},z_{2}+e_{q}\in Z}{\sum }\left\vert \left\langle
\mathfrak{e}_{x},\mathrm{e}^{-ish_{\mathcal{Z}^{(\tau )}}^{(\omega
)}}S_{z_{2}+e_{q},z_{2}}^{(\omega )}\mathrm{e}^{ish_{\mathcal{Z}^{(\tau
)}}^{(\omega )}}S_{z_{1}+e_{k},z_{1}}^{(\omega )}\mathfrak{e}%
_{y}\right\rangle _{\mathfrak{h}}\right\vert  \notag \\
&\leq &2\times 36^{2}\left( 1+\vartheta \right) ^{2}\mathrm{e}^{2\left\vert
s\eta \right\vert +4\mu _{\eta }}\left( \sum\limits_{x\in \mathbb{Z}^{d}}%
\mathrm{e}^{-2\mu _{\eta }|x|}\right) ^{2}\underset{z\in \left( \Lambda
\backslash \cup \mathcal{Z}\right) \cup \left( \cup \mathcal{\partial }%
_{\Lambda }(\mathcal{Z})\right) }{\sum }1  \label{first term}
\end{eqnarray}%
for any $s\in \mathbb{R}$. Similar to (\ref{term other}), the quantity%
\begin{equation*}
\sum\limits_{x,y\in \mathbb{Z}^{d}}\left\vert \left\langle \mathfrak{e}%
_{x},\left( K_{\{\Lambda \},\mathcal{Z}_{\tau }}^{(\omega ,\mathcal{E})}-K_{%
\mathcal{Z},\mathcal{Z}_{\tau }}^{(\omega ,\mathcal{E})}\right) \mathfrak{e}%
_{y}\right\rangle _{\mathfrak{h}}\right\vert
\end{equation*}%
is a sum of nine terms. The first one is (\ref{first term}), the last one is
related to $\Re \mathrm{e}\{S_{x+e_{k},x}^{(\omega )}\}$ and gives the
constant $d\left( 1+\vartheta \right) $ in $D_{\ref{def bilineqr copy(3)}}$.
The seven remaining ones satisfy the same bound as the first one.
\end{proof}

\subsection{Finite-volume Generating Functions\label{Sectino tech4}}

Fix $\beta \in \mathbb{R}^{+}$ and $\lambda ,\vartheta \in \mathbb{R}%
_{0}^{+} $. Given $\mathcal{E}\in C_{0}^{0}(\mathbb{R};\mathbb{R}^{d})$, $%
\omega \in \Omega $ and three finite collections $\mathcal{Z},\mathcal{Z}%
^{(\varrho )},\mathcal{Z}^{(\tau )}\in \mathfrak{Z}_{\text{f}}$, we define
the finite-volume generating function
\begin{equation}
\mathrm{J}_{\mathcal{Z},\mathcal{Z}^{(\varrho )},\mathcal{Z}^{(\tau
)}}^{(\omega ,\mathcal{E})}\doteq g_{\mathcal{Z},\mathcal{Z}^{(\varrho )},%
\mathcal{Z}^{(\tau )}}^{(\omega ,\mathcal{E})}-g_{\mathcal{Z},\mathcal{Z}%
^{(\varrho )},\mathcal{Z}^{(\tau )}}^{(\omega ,0)},
\label{generating functions0}
\end{equation}%
where%
\begin{equation}
g_{\mathcal{Z},\mathcal{Z}^{(\varrho )},\mathcal{Z}^{(\tau )}}^{(\omega ,%
\mathcal{E})}\doteq \frac{1}{|\cup \mathcal{Z}|}\ln \mathrm{tr}\left( \exp
(-\beta \langle \mathrm{A},h_{\mathcal{Z}^{(\varrho )}}^{(\omega )}\mathrm{A}%
\rangle )\exp (\mathfrak{K}_{\mathcal{Z},\mathcal{Z}^{(\tau )}}^{(\omega ,%
\mathcal{E})})\right) .  \label{generating functions}
\end{equation}%
Recall that the tracial state $\mathrm{tr}\in \mathcal{U}^{\ast }$ is the
quasi-free state satisfying (\ref{2-point correlation function}) at $\beta
=0 $, and $h_{\mathcal{Z}^{(\varrho )}}^{(\omega )}$ is the one-particle
Hamiltonian defined by (\ref{h finite f}). See also Definition \ref{def
bilineqr} and (\ref{eq:par-curr-sum-gen}). By construction, note that
\begin{equation}
\frac{1}{\left\vert \Lambda _{L}\right\vert }\ln \varrho ^{(\omega )}\left(
\mathrm{e}^{\left\vert \Lambda _{L}\right\vert \mathbb{I}_{\Lambda
_{L}}^{(\omega ,\mathcal{E})}}\right) =\lim_{L_{\varrho }\to \infty
}\lim_{L_{\tau }\to \infty }\mathrm{J}_{\mathcal{\{}\Lambda _{L}\mathcal{\}},%
\mathcal{\{}\Lambda _{L_{\varrho }}\mathcal{\}},\mathcal{\{}\Lambda
_{L_{\tau }}\mathcal{\}}}^{(\omega ,\mathcal{E})}.  \label{limit cool}
\end{equation}%
The family of functions $\mathcal{E\mapsto }\mathrm{J}_{\mathcal{Z},\mathcal{%
Z}^{(\varrho )},\mathcal{Z}^{(\tau )}}^{(\omega ,\mathcal{E})}$ is
equicontinuous with uniformly bounded second derivative:

\begin{proposition}[Equicontinuity of generating functions]
\label{prop:equic}\mbox{
}\newline
Fix $n\in \mathbb{N}$. The family of maps $\mathcal{E\mapsto }\mathrm{J}_{%
\mathcal{Z},\mathcal{Z}^{(\varrho )},\mathcal{Z}^{(\tau )}}^{(\omega ,%
\mathcal{E})}$ from $C_{0}^{0}([-n,n];\mathbb{R}^{d})\subset C_{0}^{0}(%
\mathbb{R};\mathbb{R}^{d})$ to $\mathbb{R}$, for $\beta \in \mathbb{R}^{+}$,
$\lambda \in \mathbb{R}_{0}^{+}$, $\omega \in \Omega $, $\mathcal{Z},%
\mathcal{Z}^{(\varrho )},\mathcal{Z}^{(\tau )}\in \mathfrak{Z}_{\text{f}}$, $%
\vec{w}\in {\mathbb{R}}^{d}$ with $\left\Vert \vec{w}\right\Vert _{\mathbb{R}%
^{d}}=1$, and $\vartheta $ in a compact set of $\mathbb{R}_{0}^{+}$, is
equicontinuous w.r.t. the sup norm for $\mathcal{E}$ in any bounded set of $%
C_{0}^{0}([-n,n];\mathbb{R}^{d})$.
\end{proposition}

\begin{proof}
Fix $n\in \mathbb{N}$, $\beta \in \mathbb{R}^{+}$, $\lambda ,\vartheta \in
\mathbb{R}_{0}^{+}$, $\omega \in \Omega $, $\mathcal{Z},\mathcal{Z}%
^{(\varrho )},\mathcal{Z}^{(\tau )}\in \mathfrak{Z}_{\text{f}}$. By using
Lemma \ref{lemma:suppor_non_int1} (ii), for any $\mathcal{E}_{0},\mathcal{E}%
_{1}\in C_{0}^{0}([-n,n];\mathbb{R}^{d})$,
\begin{eqnarray}
&&\left\vert g_{\mathcal{Z},\mathcal{Z}^{(\varrho )},\mathcal{Z}^{(\tau
)}}^{(\omega ,\mathcal{E}_{1})}-g_{\mathcal{Z},\mathcal{Z}^{(\varrho )},%
\mathcal{Z}^{(\tau )}}^{(\omega ,\mathcal{E}_{0})}\right\vert
\label{iodiot0} \\
&\leq &\frac{1}{|\cup \mathcal{Z}|}\sup_{\alpha \in \left[ 0,\,1\right]
}\sup_{u\in \left[ -1/2,\,1/2\right] }\left\Vert \mathrm{e}^{u\mathfrak{K}_{%
\mathcal{Z},\mathcal{Z}^{(\tau )}}^{(\omega ,\alpha \mathcal{E}_{1}+\left(
1-\alpha \right) \mathcal{E}_{0})}}\mathfrak{K}_{\mathcal{Z},\mathcal{Z}%
^{(\tau )}}^{(\omega ,\mathcal{E}_{1}\mathcal{-E}_{0})}\mathrm{e}^{-u%
\mathfrak{K}_{\mathcal{Z},\mathcal{Z}^{(\tau )}}^{(\omega ,\alpha \mathcal{E}%
_{1}+\left( 1-\alpha \right) \mathcal{E}_{0})}}\right\Vert _{\mathcal{U}}.
\notag
\end{eqnarray}%
Recall that, for any $\mathcal{E}\in C_{0}^{0}(\mathbb{R};\mathbb{R}^{d})$, $%
\mathfrak{K}_{\mathcal{Z},\mathcal{Z}^{(\tau )}}^{(\omega ,\mathcal{E})}$ is
the bilinear element associated with the operator $K_{\mathcal{Z},\mathcal{Z}%
^{(\tau )}}^{(\omega ,\mathcal{E})}$. See (\ref{bilinear idiot}) and (\ref%
{definition K}). In particular, from (\ref{dynamic bilinear}), we deduce the
inequality%
\begin{equation}
\sup_{u\in \left[ -1/2,\,1/2\right] }\sup_{x,y\in \mathbb{Z}^{d}}\left\Vert
\mathrm{e}^{u\mathfrak{K}_{\mathcal{Z},\mathcal{Z}^{(\tau )}}^{(\omega ,%
\mathcal{E})}}a\left( \mathfrak{e}_{x}\right) ^{\ast }a\left( \mathfrak{e}%
_{y}\right) \mathrm{e}^{-u\mathfrak{K}_{\mathcal{Z},\mathcal{Z}^{(\tau
)}}^{(\omega ,\mathcal{E})}}\right\Vert _{\mathcal{U}}\leq \mathrm{e}^{\Vert
K_{\mathcal{Z},\mathcal{Z}^{(\tau )}}^{(\omega ,\mathcal{E})}\Vert _{%
\mathcal{B}(\mathfrak{h})}}.  \label{idiot4}
\end{equation}%
The assertion then follows by combining (\ref{bilinear idiot}), (\ref%
{iodiot0}) and Definition \ref{def bilineqr} with (\ref{idiot4}) and Lemmata %
\ref{lemma:norm_estimate}, \ref{def bilineqr copy(1)}.
\end{proof}

\begin{proposition}[Uniform boundedness of second derivatives]
\label{prop:equic copy(5)}\mbox{
}\newline
Fix $\mathcal{E}\in C_{0}^{0}\left( \mathbb{R};\mathbb{R}^{d}\right) $ and $%
\beta _{1},s_{1},\vartheta _{1},\lambda _{1}\in \mathbb{R}^{+}$. Then,%
\begin{equation*}
\sup_{\substack{ \beta \in \left( 0,\beta _{1}\right] ,\ \vartheta \in \left[
0,\vartheta _{1}\right] ,\ \lambda \in \left[ 0,\lambda _{1}\right]  \\ %
\omega \in \Omega ,\ s\in \left[ -s_{1},s_{1}\right] ,\ \mathcal{Z},\mathcal{%
Z}^{(\varrho )},\mathcal{Z}^{(\tau )}\in \mathfrak{Z}_{\text{f}}}}\left\{
\left\vert \partial _{s}\mathrm{J}_{\mathcal{Z},\mathcal{Z}^{(\varrho )},%
\mathcal{Z}^{(\tau )}}^{(\omega ,s\mathcal{E})}\right\vert +\left\vert
\partial _{s}^{2}\mathrm{J}_{\mathcal{Z},\mathcal{Z}^{(\varrho )},\mathcal{Z}%
^{(\tau )}}^{(\omega ,s\mathcal{E})}\right\vert \right\} <\infty .
\end{equation*}
\end{proposition}

\begin{proof}
Fix the parameters of the proposition. Then, by cyclicity of the tracial
state,
\begin{equation*}
\partial _{s}\mathrm{J}_{\mathcal{Z},\mathcal{Z}^{(\varrho )},\mathcal{Z}%
^{(\tau )}}^{(\omega ,s\mathcal{E})}=\frac{1}{|\cup \mathcal{Z}|}\varpi
_{s}\left( \mathfrak{K}_{\mathcal{Z},\mathcal{Z}^{(\tau )}}^{(\omega ,%
\mathcal{E})}\right)
\end{equation*}%
and
\begin{equation*}
\partial _{s}^{2}\mathrm{J}_{\mathcal{Z},\mathcal{Z}^{(\varrho )},\mathcal{Z}%
^{(\tau )}}^{(\omega ,s\mathcal{E})}=\frac{1}{|\cup \mathcal{Z}|}\left(
\varpi _{s}\left( \left( \mathfrak{K}_{\mathcal{Z},\mathcal{Z}^{(\tau
)}}^{(\omega ,\mathcal{E})}\right) ^{2}\right) -\varpi _{s}\left( \mathfrak{K%
}_{\mathcal{Z},\mathcal{Z}^{(\tau )}}^{(\omega ,\mathcal{E})}\right)
^{2}\right) ,
\end{equation*}%
where $\varpi _{s}$ is the state defined, for any $B\in \mathcal{U}$, by
\begin{equation*}
\varpi _{s}\left( B\right) =\frac{\mathrm{tr}\left( B\mathrm{e}^{\frac{s}{2}%
\mathfrak{K}_{\mathcal{Z},\mathcal{Z}^{(\tau )}}^{(\omega ,\mathcal{E})}}%
\mathrm{e}^{-\beta \langle \mathrm{A},h_{\mathcal{Z}^{(\varrho )}}^{(\omega
)}\mathrm{A}\rangle }\mathrm{e}^{\frac{s}{2}\mathfrak{K}_{\mathcal{Z},%
\mathcal{Z}^{(\tau )}}^{(\omega ,\mathcal{E})}}\right) }{\mathrm{tr}\left(
\mathrm{e}^{\frac{s}{2}\mathfrak{K}_{\mathcal{Z},\mathcal{Z}^{(\tau
)}}^{(\omega ,\mathcal{E})}}\mathrm{e}^{-\beta \langle \mathrm{A},h_{%
\mathcal{Z}^{(\varrho )}}^{(\omega )}\mathrm{A}\rangle }\mathrm{e}^{\frac{s}{%
2}\mathfrak{K}_{\mathcal{Z},\mathcal{Z}^{(\tau )}}^{(\omega ,\mathcal{E}%
)}}\right) }.
\end{equation*}%
By Lemma \ref{lemma cool} and (\ref{bilinear idiot}), observe that $\varpi
_{s}$ is the quasi-free state satisfying
\begin{equation}
\varpi _{s}(a^{\ast }\left( \varphi \right) a\left( \psi \right)
)=\left\langle \psi ,\frac{1}{1+\mathrm{e}^{-\frac{s}{2}K_{\mathcal{Z},%
\mathcal{Z}^{(\tau )}}^{(\omega ,\mathcal{E})}}\mathrm{e}^{\beta h_{\mathcal{%
Z}^{(\varrho )}}^{(\omega )}}\mathrm{e}^{-\frac{s}{2}K_{\mathcal{Z},\mathcal{%
Z}^{(\tau )}}^{(\omega ,\mathcal{E})}}}\varphi \right\rangle _{\mathfrak{h}%
},\qquad \varphi ,\psi \in \mathfrak{h}.  \label{cool2}
\end{equation}%
Therefore, by (\ref{bilinear idiot}) and Definition \ref{def bilineqr}, we
directly compute that
\begin{equation*}
\partial _{s}\mathrm{J}_{\mathcal{Z},\mathcal{Z}^{(\varrho )},\mathcal{Z}%
^{(\tau )}}^{(\omega ,s\mathcal{E})}=\frac{1}{|\cup \mathcal{Z}|}%
\sum\limits_{x,y\in \mathbb{Z}^{d}}\left\langle \mathfrak{e}_{x},K_{\mathcal{%
Z},\mathcal{Z}^{(\tau )}}^{(\omega ,\mathcal{E})}\mathfrak{e}%
_{y}\right\rangle _{\mathfrak{h}}\varpi _{s}\left( a\left( \mathfrak{e}%
_{x}\right) ^{\ast }a\left( \mathfrak{e}_{y}\right) \right)
\end{equation*}%
and
\begin{multline*}
\partial _{s}^{2}\mathrm{J}_{\mathcal{Z},\mathcal{Z}^{(\varrho )},\mathcal{Z}%
^{(\tau )}}^{(\omega ,s\mathcal{E})}=\frac{1}{|\cup \mathcal{Z}|}%
\sum\limits_{x,y,u,v\in \mathbb{Z}^{d}}\left\langle \mathfrak{e}_{x},K_{%
\mathcal{Z},\mathcal{Z}^{(\tau )}}^{(\omega ,\mathcal{E})}\mathfrak{e}%
_{y}\right\rangle _{\mathfrak{h}}\left\langle \mathfrak{e}_{u},K_{\mathcal{Z}%
,\mathcal{Z}^{(\tau )}}^{(\omega ,\mathcal{E})}\mathfrak{e}_{v}\right\rangle
_{\mathfrak{h}} \\
\times \varpi _{s}\left( a\left( \mathfrak{e}_{y}\right) a\left( \mathfrak{e}%
_{u}\right) ^{\ast }\right) \varpi _{s}\left( a\left( \mathfrak{e}%
_{x}\right) ^{\ast }a\left( \mathfrak{e}_{v}\right) \right) ,
\end{multline*}%
because of the identity
\begin{equation*}
\varpi _{s}\left( a(\mathfrak{e}_{x})^{\ast }a(\mathfrak{e}_{y})a(\mathfrak{e%
}_{u})^{\ast }a(\mathfrak{e}_{v})\right) =\varpi _{s}\left( a(\mathfrak{e}%
_{x})^{\ast }a(\mathfrak{e}_{y})\right) \varpi _{s}\left( a(\mathfrak{e}%
_{u})^{\ast }a(\mathfrak{e}_{v})\right) +\varpi _{s}\left( a(\mathfrak{e}%
_{y})a(\mathfrak{e}_{u})^{\ast }\right) \varpi _{s}\left( a(\mathfrak{e}%
_{x})^{\ast }a(\mathfrak{e}_{v})\right) ,
\end{equation*}%
for $x,y,u,v\in \mathbb{Z}^{d}$, by (\ref{ass O0-00bis}) for $\rho =\varpi
_{s}$. As a consequence,
\begin{equation*}
\left\vert \partial _{s}\mathrm{J}_{\mathcal{Z},\mathcal{Z}^{(\varrho )},%
\mathcal{Z}^{(\tau )}}^{(\omega ,s\mathcal{E})}\right\vert \leq \frac{1}{%
|\cup \mathcal{Z}|}\sum\limits_{x,y\in \mathbb{Z}^{d}}\left\vert
\left\langle \mathfrak{e}_{x},K_{\mathcal{Z},\mathcal{Z}^{(\tau )}}^{(\omega
,\mathcal{E})}\mathfrak{e}_{y}\right\rangle _{\mathfrak{h}}\right\vert
\end{equation*}%
and%
\begin{eqnarray*}
\left\vert \partial _{s}^{2}\mathrm{J}_{\mathcal{Z},\mathcal{Z}^{(\varrho )},%
\mathcal{Z}^{(\tau )}}^{(\omega ,s\mathcal{E})}\right\vert  &\leq
&\sup_{u,v\in \mathbb{Z}^{d}}\left\vert \left\langle \mathfrak{e}_{u},K_{%
\mathcal{Z},\mathcal{Z}^{(\tau )}}^{(\omega ,\mathcal{E})}\mathfrak{e}%
_{v}\right\rangle _{\mathfrak{h}}\right\vert \left( \frac{1}{|\cup \mathcal{Z%
}|}\sum\limits_{x,y\in \mathbb{Z}^{d}}\left\vert \left\langle \mathfrak{e}%
_{x},K_{\mathcal{Z},\mathcal{Z}^{(\tau )}}^{(\omega ,\mathcal{E})}\mathfrak{e%
}_{y}\right\rangle _{\mathfrak{h}}\right\vert \right)  \\
&&\times \sup_{y\in \mathbb{Z}^{d}}\sum\limits_{u\in \mathbb{Z}%
^{d}}\left\vert \varpi _{s}\left( a\left( \mathfrak{e}_{y}\right) a\left(
\mathfrak{e}_{u}\right) ^{\ast }\right) \right\vert \sup_{x\in \mathbb{Z}%
^{d}}\sum\limits_{v\in \mathbb{Z}^{d}}\left\vert \varpi _{s}\left( a\left(
\mathfrak{e}_{x}\right) ^{\ast }a\left( \mathfrak{e}_{v}\right) \right)
\right\vert ,
\end{eqnarray*}%
which, by Lemma \ref{def bilineqr copy(1)}, implies that
\begin{equation}
\left\vert \partial _{s}\mathrm{J}_{\mathcal{Z},\mathcal{Z}^{(\varrho )},%
\mathcal{Z}^{(\tau )}}^{(\omega ,s\mathcal{E})}\right\vert \leq D_{\ref{def
bilineqr copy(1)}}\left( \int_{\mathbb{R}}\left\Vert \mathcal{E}\left(
\alpha \right) \right\Vert _{\mathbb{R}^{d}}\mathrm{e}^{2\left\vert \alpha
\eta \right\vert }\mathrm{d}\alpha \right) \underset{z\in \mathbb{Z}^{d}}{%
\sum }\mathrm{e}^{-\mu _{\eta }|z|}\left( 1+\eta \right)
\label{punchinle00}
\end{equation}%
as well as%
\begin{eqnarray}
\left\vert \partial _{s}^{2}\mathrm{J}_{\mathcal{Z},\mathcal{Z}^{(\varrho )},%
\mathcal{Z}^{(\tau )}}^{(\omega ,s\mathcal{E})}\right\vert  &\leq &D_{\ref%
{def bilineqr copy(1)}}^{2}\left( \int_{\mathbb{R}}\left\Vert \mathcal{E}%
\left( \alpha \right) \right\Vert _{\mathbb{R}^{d}}\mathrm{e}^{2\left\vert
\alpha \eta \right\vert }\mathrm{d}\alpha \right) ^{2}\left( 1+\eta \right)
^{2}\underset{z\in \mathbb{Z}^{d}}{\sum }\mathrm{e}^{-\mu _{\eta }|z|}
\notag \\
&&\times \sup_{y\in \mathbb{Z}^{d}}\sum\limits_{u\in \mathbb{Z}%
^{d}}\left\vert \varpi _{s}\left( a\left( \mathfrak{e}_{y}\right) a\left(
\mathfrak{e}_{u}\right) ^{\ast }\right) \right\vert \sup_{x\in \mathbb{Z}%
^{d}}\sum\limits_{v\in \mathbb{Z}^{d}}\left\vert \varpi _{s}\left( a\left(
\mathfrak{e}_{x}\right) ^{\ast }a\left( \mathfrak{e}_{v}\right) \right)
\right\vert .  \label{punchinle}
\end{eqnarray}%
Again by Lemma \ref{def bilineqr copy(1)} together with (\ref%
{Combes-ThomasCombes-Thomas})-(\ref{Combes-ThomasCombes-Thomasbis}), for any
$\mu >\mu _{\eta }$,
\begin{equation*}
\sup_{\substack{ \beta \in \left( 0,\beta _{1}\right] ,\ \vartheta \in \left[
0,\vartheta _{1}\right] ,\ \lambda \in \left[ 0,\lambda _{1}\right]  \\ %
\omega \in \Omega ,\ s\in \left[ -s_{1},s_{1}\right] ,\ \mathcal{Z},\mathcal{%
Z}^{(\varrho )},\mathcal{Z}^{(\tau )}\in \mathfrak{Z}_{\text{f}}}}\left\{
\mathbf{S}_{0}(sK_{\mathcal{Z},\mathcal{Z}^{(\tau )}}^{(\omega ,\mathcal{E}%
)},\mu )+\mathbf{S}_{0}(\beta h_{\mathcal{Z}^{(\varrho )}}^{(\omega )},\mu
)\right\} <\infty .
\end{equation*}%
See (\ref{S0}). We thus infer from (\ref{cool2}) and Corollary \ref{Lemma
fermi1 copy(1)} that there is a constant $\mu _{1}\in \mathbb{R}^{+}$ such
that, for any $x,y\in \mathbb{Z}^{d}$,%
\begin{equation*}
\sup_{\substack{ \beta \in \left( 0,\beta _{1}\right] ,\ \vartheta \in \left[
0,\vartheta _{1}\right] ,\ \lambda \in \left[ 0,\lambda _{1}\right]  \\ %
\omega \in \Omega ,\ s\in \left[ -s_{1},s_{1}\right] ,\ \mathcal{Z},\mathcal{%
Z}^{(\varrho )},\mathcal{Z}^{(\tau )}\in \mathfrak{Z}_{\text{f}}}}\left\vert
\varpi _{s}\left( a\left( \mathfrak{e}_{x}\right) ^{\ast }a\left( \mathfrak{e%
}_{y}\right) \right) \right\vert \leq 2\mathrm{e}^{-\mu _{1}|x-y|}.
\end{equation*}%
Combining this estimate with (\ref{punchinle00})-(\ref{punchinle}), one gets
the assertion.
\end{proof}

The local generating functionals (\ref{generating functions0}) can be
approximately decomposed into boxes of fixed volume: By using the boxes (\ref%
{eq:boxesl1}), for any subset $\Lambda \subset \mathbb{Z}^{d}$ and $l\in
\mathbb{N}$, we define the $l$-th box decomposition $\mathcal{Z}^{(\Lambda
,l)}$ of $\Lambda $ by
\begin{equation*}
\mathcal{Z}^{(\Lambda ,l)}\doteq \left\{ \Lambda _{l}+\left( 2l+1\right)
x:x\in \mathbb{Z}^{d}\text{ with }(\Lambda _{l}+\left( 2l+1\right) x)\subset
\Lambda \right\} \in \mathfrak{Z}.
\end{equation*}%
Then, we get the following assertion:

\begin{proposition}[Box decomposition of generating functions]
\label{Main proposition}\mbox{
}\newline
Fix $n\in \mathbb{N}$ and $\beta _{1},\lambda _{1},\vartheta _{1}\in \mathbb{%
R}^{+}$. Then,
\begin{equation*}
\lim_{l\to \infty }\limsup_{L_{\tau }\geq L_{\varrho }\geq L\to \infty
}\left\vert \mathrm{J}_{\mathcal{\{}\Lambda _{L}\mathcal{\}},\mathcal{\{}%
\Lambda _{L_{\varrho }}\mathcal{\}},\mathcal{\{}\Lambda _{L_{\tau }}\mathcal{%
\}}}^{(\omega ,\mathcal{E})}-\frac{1}{\left\vert \mathcal{Z}^{(\Lambda
_{L},l)}\right\vert }\sum_{Z\in \mathcal{Z}^{(\Lambda _{L},l)}}\mathrm{J}_{%
\mathcal{\{}Z\mathcal{\}},\mathcal{\{}Z\mathcal{\}},\mathcal{\{}Z\mathcal{\}}%
}^{(\omega ,\mathcal{E})}\right\vert =0,
\end{equation*}%
uniformly w.r.t. $\beta \in \left[ 0,\beta _{1}\right] $, $\vartheta \in %
\left[ 0,\vartheta _{1}\right] $, $\lambda \in \left[ 0,\lambda _{1}\right] $%
, $\omega \in \Omega $ and $\mathcal{E}$ in any bounded set of $%
C_{0}^{0}([-n,n];\mathbb{R}^{d})$.
\end{proposition}

\noindent The proof of this statement is divided in a series of Lemmata:

\begin{lemma}[Box decomposition of generating functions - I]
\label{prop:equic copy(4)}\mbox{
}\newline
Fix $\beta _{1},\lambda _{1},\vartheta _{1}\in \mathbb{R}^{+}$. Then,
\begin{equation*}
\limsup_{L_{\tau }\geq L_{\varrho }\geq L\to \infty }\left\vert g_{\mathcal{%
\{}\Lambda _{L}\mathcal{\}},\mathcal{\{}\Lambda _{L_{\varrho }}\mathcal{\}},%
\mathcal{\{}\Lambda _{L_{\tau }}\mathcal{\}}}^{(\omega ,\mathcal{E})}-g_{%
\mathcal{\{}\Lambda _{L}\mathcal{\}},\mathcal{\{}\Lambda _{L_{\varrho
}}\backslash \Lambda _{L},\Lambda _{L}\mathcal{\}},\mathcal{\{}\Lambda
_{L_{\tau }}\mathcal{\}}}^{(\omega ,\mathcal{E})}\right\vert =0,
\end{equation*}%
uniformly w.r.t. $\beta \in \left[ 0,\beta _{1}\right] $, $\vartheta \in %
\left[ 0,\vartheta _{1}\right] $, $\lambda \in \left[ 0,\lambda _{1}\right] $%
, $\omega \in \Omega $ and $\mathcal{E}\in C_{0}^{0}(\mathbb{R};\mathbb{R}%
^{d})$.
\end{lemma}

\begin{proof}
Fix all parameters of the lemma. By Lemma \ref{lemma:suppor_non_int1} (ii),%
\begin{multline*}
\left\vert g_{\mathcal{\{}\Lambda _{L}\mathcal{\}},\mathcal{\{}\Lambda
_{L_{\varrho }}\mathcal{\}},\mathcal{\{}\Lambda _{L_{\tau }}\mathcal{\}}%
}^{(\omega ,\mathcal{E})}-g_{\mathcal{\{}\Lambda _{L}\mathcal{\}},\mathcal{\{%
}\Lambda _{L_{\varrho }}\backslash \Lambda _{L},\Lambda _{L}\mathcal{\}},%
\mathcal{\{}\Lambda _{L_{\tau }}\mathcal{\}}}^{(\omega ,\mathcal{E}%
)}\right\vert \\
\leq \frac{\beta }{|\Lambda _{L}|}\sup_{\alpha \in \left[ 0,1\right]
}\sup_{u\in \left[ -1/2,1/2\right] }\left\Vert \mathrm{e}^{u\beta
\left\langle \mathrm{A},h_{\alpha }\mathrm{A}\right\rangle }\left\langle
\mathrm{A},\left( h_{1}-h_{0}\right) \mathrm{A}\right\rangle \mathrm{e}%
^{-u\beta \left\langle \mathrm{A},h_{\alpha }\mathrm{A}\right\rangle
}\right\Vert _{\mathcal{U}},
\end{multline*}%
where
\begin{equation*}
h_{\alpha }\doteq \alpha h_{\mathcal{\{}\Lambda _{L_{\varrho }}\mathcal{\}}%
}^{(\omega )}+\left( 1-\alpha \right) h_{\{\Lambda _{L_{\varrho }}\backslash
\Lambda _{L},\Lambda _{L}\}}^{(\omega )},\qquad \alpha \in \left[ 0,1\right]
.
\end{equation*}%
By using estimates similar to (\ref{idiot4}), we get%
\begin{align}
\left\vert g_{\mathcal{\{}\Lambda _{L}\mathcal{\}},\mathcal{\{}\Lambda
_{L_{\varrho }}\mathcal{\}},\mathcal{\{}\Lambda _{L_{\tau }}\mathcal{\}}%
}^{(\omega ,\mathcal{E})}-g_{\mathcal{\{}\Lambda _{L}\mathcal{\}},\mathcal{\{%
}\Lambda _{L_{\varrho }}\backslash \Lambda _{L},\Lambda _{L}\mathcal{\}},%
\mathcal{\{}\Lambda _{L_{\tau }}\mathcal{\}}}^{(\omega ,\mathcal{E}%
)}\right\vert & \leq \frac{\beta \mathrm{e}^{\beta \left( \lambda +2d\right)
\left( 1+\vartheta \right) }}{|\Lambda _{L}|}\underset{x,y\in \mathbb{Z}^{d}}%
{\sum }\left\vert \left\langle \mathfrak{e}_{x},\left( h_{1}-h_{0}\right)
\mathfrak{e}_{y}\right\rangle _{\mathfrak{h}}\right\vert  \notag \\
& \leq 4d\left( 1+\vartheta \right) \beta \mathrm{e}^{\beta \left( \lambda
+2d\right) \left( 1+\vartheta \right) }\frac{1}{|\Lambda _{L}|}\sum_{z\in
\cup \mathcal{\partial }_{\Lambda _{L_{\varrho }}}(\Lambda _{L})}1.
\label{idiot5}
\end{align}%
See (\ref{X2}). Since
\begin{equation*}
\limsup_{L_{\varrho }\geq L\to \infty }\frac{1}{|\Lambda _{L}|}\sum_{z\in
\cup \mathcal{\partial }_{\Lambda _{L_{\varrho }}}(\Lambda _{L})}1=0,
\end{equation*}%
the assertion follows.
\end{proof}

\begin{lemma}[Box decomposition of generating functions - II]
\label{prop:equic copy(1)}\mbox{
}\newline
Fix $n\in \mathbb{N}$ and $\vartheta _{1},\lambda _{1}\in \mathbb{R}^{+}$.
Then,
\begin{equation*}
\lim_{l\to \infty }\limsup_{L_{\tau }\geq L_{\rho }\geq L\to \infty
}\left\vert g_{\mathcal{\{}\Lambda _{L}\mathcal{\}},\mathcal{\{}\Lambda
_{L_{\varrho }}\backslash \Lambda _{L},\Lambda _{L}\mathcal{\}},\mathcal{\{}%
\Lambda _{L_{\tau }}\mathcal{\}}}^{(\omega ,\mathcal{E})}-g_{\mathcal{\{}%
\Lambda _{L}\mathcal{\}},\mathcal{\{}\Lambda _{L_{\varrho }}\backslash
\Lambda _{L},\Lambda _{L}\mathcal{\}},\mathcal{Z}^{(\Lambda
_{L},l)}}^{(\omega ,\mathcal{E})}\right\vert =0,
\end{equation*}%
uniformly w.r.t. $\vartheta \in \left[ 0,\vartheta _{1}\right] $, $\lambda
\in \left[ 0,\lambda _{1}\right] $, $\omega \in \Omega $ and $\mathcal{E}$
in any bounded set of $C_{0}^{0}([-n,n];\mathbb{R}^{d})$.
\end{lemma}

\begin{proof}
Fix all parameters of the lemma, in particular $L_{\tau }\geq L_{\rho }\geq
L\geq l$, $\omega \in \Omega $ and $\lambda \in \left[ 0,\lambda _{1}\right]
$. By Lemma \ref{lemma:suppor_non_int1} (ii) and (\ref{bilinear idiot}),%
\begin{eqnarray*}
&&\left\vert g_{\mathcal{\{}\Lambda _{L}\mathcal{\}},\mathcal{\{}\Lambda
_{L_{\varrho }}\backslash \Lambda _{L}\mathcal{\}},\mathcal{\{}\Lambda
_{L_{\tau }}\mathcal{\}}}^{(\omega ,\mathcal{E})}-g_{\mathcal{\{}\Lambda _{L}%
\mathcal{\}},\mathcal{\{}\Lambda _{L_{\varrho }}\backslash ,\Lambda _{L}%
\mathcal{\}},\mathcal{Z}^{(\Lambda _{L},l)}}^{(\omega ,\mathcal{E}%
)}\right\vert \\
&\leq &\frac{1}{|\Lambda _{L}|}\sup_{\alpha \in \left[ 0,1\right]
}\sup_{u\in \left[ -1/2,1/2\right] }\left\Vert \mathrm{e}^{u\left\langle
\mathrm{A},K_{\alpha }\mathrm{A}\right\rangle }\left\langle \mathrm{A}%
,\left( K_{1}-K_{0}\right) \mathrm{A}\right\rangle \mathrm{e}%
^{-u\left\langle \mathrm{A},K_{\alpha }\mathrm{A}\right\rangle }\right\Vert
_{\mathcal{U}},
\end{eqnarray*}%
where%
\begin{equation*}
K_{\alpha }\doteq \alpha K_{\mathcal{\{}\Lambda _{L}\mathcal{\}},\mathcal{\{}%
\Lambda _{L_{\tau }}\mathcal{\}}}^{(\omega ,\mathcal{E})}+\left( 1-\alpha
\right) K_{\mathcal{\{}\Lambda _{L}\mathcal{\}},\mathcal{Z}^{(\Lambda
_{L},l)}}^{(\omega ,\mathcal{E})},\qquad \alpha \in \left[ 0,1\right] .
\end{equation*}%
Like in the proof of Lemma \ref{prop:equic copy(4)}, by (\ref{bilinear idiot}%
) and Lemma \ref{def bilineqr copy(2)},%
\begin{eqnarray}
&&\left\vert g_{\mathcal{\{}\Lambda _{L}\mathcal{\}},\mathcal{\{}\Lambda
_{L_{\varrho }}\backslash \Lambda _{L}\mathcal{\}},\mathcal{\{}\Lambda
_{L_{\tau }}\mathcal{\}}}^{(\omega ,\mathcal{E})}-g_{\mathcal{\{}\Lambda _{L}%
\mathcal{\}},\mathcal{\{}\Lambda _{L_{\varrho }}\backslash ,\Lambda _{L}%
\mathcal{\}},\mathcal{Z}^{(\Lambda _{L},l)}}^{(\omega ,\mathcal{E}%
)}\right\vert  \label{totottotot} \\
&\leq &D_{\ref{def bilineqr copy(2)}}\left( \int_{\mathbb{R}}\left\Vert
\mathcal{E}\left( \alpha \right) \right\Vert _{\mathbb{R}^{d}}\alpha ^{2}%
\mathrm{e}^{2\left\vert \alpha \eta \right\vert }\mathrm{d}\alpha \right)
\mathrm{e}^{\sup_{\alpha \in \left[ 0,1\right] }\left\Vert K_{\alpha
}\right\Vert _{\mathcal{B}(\mathfrak{h})}}  \notag \\
&&\times \frac{1}{|\Lambda _{L}|}\left( \sum\limits_{x\in \Lambda
_{L}}\sum_{z\in \Lambda _{L_{\tau }}\backslash \cup \mathcal{Z}^{(\Lambda
_{L},l)}}\mathrm{e}^{-\frac{\mu _{\eta }}{2}|x-z|}+\left( \sum_{z\in \mathbb{%
Z}^{d}}\mathrm{e}^{-\frac{\mu _{\eta }}{2}|z|}\right) \sum_{x\in \cup
\mathcal{\partial }_{\Lambda _{L_{\tau }}}(\mathcal{Z}^{(\Lambda
_{L},l)})}1\right) .  \notag
\end{eqnarray}%
By Lemmata \ref{lemma:norm_estimate} and \ref{def bilineqr copy(1)}, for any
$n\in \mathbb{N}$, observe that the operator norms of $K_{\alpha }$ is
uniformly bounded for $\alpha \in \left[ 0,1\right] $, $\vartheta \in
\lbrack 0,\vartheta _{1}]$, $\lambda \in \mathbb{R}_{0}^{+}$, $\omega \in
\Omega $, $L,L_{\tau },l\in \mathbb{N}$ and $\mathcal{E}$ in any bounded set
of $C_{0}^{0}([-n,n];\mathbb{R}^{d})$. Note additionally that
\begin{equation*}
\limsup_{L_{\tau }\geq L\to \infty }\frac{1}{|\Lambda _{L}|}%
\sum\limits_{x\in \Lambda _{L}}\sum_{z\in \Lambda _{L_{\tau }}\backslash
\cup \mathcal{Z}^{(\Lambda _{L},l)}}\mathrm{e}^{-\frac{\mu _{\eta }}{2}%
|x-z|}=0,
\end{equation*}%
whereas
\begin{equation*}
\limsup_{L_{\tau }\geq L\to \infty }\frac{1}{|\Lambda _{L}|}\sum_{x\in \cup
\mathcal{\partial }_{\Lambda _{L_{\tau }}}(\mathcal{Z}^{(\Lambda _{L},l)})}1=%
\mathcal{O}\left( l^{-1}\right) .
\end{equation*}%
From these last observations combined with (\ref{totottotot}), the assertion
follows.
\end{proof}

\begin{lemma}[Box decomposition of generating functions - III]
\label{prop:equic copy(2)}\mbox{
}\newline
Fix $\beta _{1},\vartheta _{1},\lambda _{1}\in \mathbb{R}^{+}$. Then,%
\begin{equation*}
\lim_{l\to \infty }\limsup_{L_{\tau }\geq L_{\varrho }\geq L\to \infty
}\left\vert g_{\mathcal{\{}\Lambda _{L}\mathcal{\}},\mathcal{\{}\Lambda
_{L_{\varrho }}\backslash \Lambda _{L},\Lambda _{L}\mathcal{\}},\mathcal{Z}%
^{(\Lambda _{L},l)}}^{(\omega ,\mathcal{E})}-g_{\mathcal{\{}\Lambda _{L}%
\mathcal{\}},\{\Lambda _{L_{\varrho }}\backslash \Lambda _{L}\}\cup \mathcal{%
Z}^{(\Lambda _{L},l)},\mathcal{Z}^{(\Lambda _{L},l)}}^{(\omega ,\mathcal{E}%
)}\right\vert =0,
\end{equation*}%
uniformly w.r.t. $\beta \in \left[ 0,\beta _{1}\right] $, $\vartheta \in %
\left[ 0,\vartheta _{1}\right] $, $\lambda \in \left[ 0,\lambda _{1}\right] $%
, $\omega \in \Omega $ and $\mathcal{E}\in C_{0}^{0}(\mathbb{R};\mathbb{R}%
^{d})$.
\end{lemma}

\begin{proof}
This lemma is proven exactly in the same way as Lemmata \ref{prop:equic
copy(4)} and \ref{prop:equic copy(1)}: Fix all parameters of the lemma and
observe that
\begin{eqnarray*}
&&\left\vert \left\langle \mathfrak{e}_{x},\left( h_{\mathcal{\{}\Lambda
_{L_{\varrho }}\backslash \Lambda _{L},\Lambda _{L}\mathcal{\}}}^{(\omega
)}-h_{\{\Lambda _{L_{\varrho }}\backslash \Lambda _{L}\}\cup \mathcal{Z}%
^{(\Lambda _{L},l)}}^{(\omega )}\right) \mathfrak{e}_{y}\right\rangle _{%
\mathfrak{h}}\right\vert \\
&\leq &\left( 1+\vartheta \right) \sum_{z_{3},z_{4}\in \Lambda
_{L}\backslash \cup \mathcal{Z}^{(\Lambda _{L},l)}\ \colon \ \left\vert
z_{3}-z_{4}\right\vert =1}\delta _{z_{3},y}\delta _{z_{4},x}+\lambda
\sum_{z_{3}\in \Lambda _{L}\backslash \cup \mathcal{Z}^{(\Lambda
_{L},l)}}\delta _{z_{3},x}\delta _{z_{3},y} \\
&&+\left( 1+\vartheta \right) \sum_{Z\in \mathcal{Z}^{(\Lambda
_{L},l)}}\sum_{\left\{ z_{3},z_{4}\right\} \in \mathcal{\partial }_{\Lambda
_{L}}(Z)}\left( \delta _{z_{3},y}\delta _{z_{4},x}+\delta _{z_{4},y}\delta
_{z_{3},x}\right) .
\end{eqnarray*}%
See (\ref{X2}). Then, similar to (\ref{idiot5}), we get the bound
\begin{multline*}
\left\vert g_{\mathcal{\{}\Lambda _{L}\mathcal{\}},\mathcal{\{}\Lambda
_{L_{\varrho }}\mathcal{\}},\mathcal{\{}\Lambda _{L_{\tau }}\mathcal{\}}%
}^{(\omega ,\mathcal{E})}-g_{\mathcal{\{}\Lambda _{L}\mathcal{\}},\mathcal{\{%
}\Lambda _{L_{\varrho }}\backslash \Lambda _{L},\Lambda _{L}\mathcal{\}},%
\mathcal{\{}\Lambda _{L_{\tau }}\mathcal{\}}}^{(\omega ,\mathcal{E}%
)}\right\vert \\
\leq \left( 4d+\lambda \right) \left( 1+\vartheta \right) \beta \mathrm{e}%
^{\beta \left( \lambda +2d\right) \left( 1+\vartheta \right) }\frac{1}{%
|\Lambda _{L}|}\left( \sum_{z\in \Lambda _{L}\backslash \cup \mathcal{Z}%
^{(\Lambda _{L},l)}}1+\sum_{Z\in \mathcal{Z}^{(\Lambda _{L},l)}}\sum_{z\in
\cup \mathcal{\partial }_{\Lambda _{L}}(Z)}1\right) ,
\end{multline*}%
where
\begin{equation*}
\limsup_{L\to \infty }\frac{1}{|\Lambda _{L}|}\left( \sum_{z\in \Lambda
_{L}\backslash \cup \mathcal{Z}^{(\Lambda _{L},l)}}1+\sum_{Z\in \mathcal{Z}%
^{(\Lambda _{L},l)}}\sum_{z\in \cup \mathcal{\partial }_{\Lambda
_{L}}(Z)}1\right) =\mathcal{O}\left( l^{-1}\right) .
\end{equation*}
\end{proof}

\begin{lemma}[Box decomposition of generating functions - IV]
\label{prop:equic copy(3)}\mbox{
}\newline
Fix $n\in \mathbb{N}$ and $\vartheta _{1}\in \mathbb{R}^{+}$. Then,%
\begin{equation*}
\lim_{l\to \infty }\limsup_{L_{\tau }\geq L_{\varrho }\geq L\to \infty
}\left\vert g_{\mathcal{\{}\Lambda _{L}\mathcal{\}},\{\Lambda _{L_{\varrho
}}\backslash \Lambda _{L}\}\cup \mathcal{Z}^{(\Lambda _{L},l)},\mathcal{Z}%
^{(\Lambda _{L},l)}}^{(\omega ,\mathcal{E})}-g_{\mathcal{Z}^{(\Lambda
_{L},l)},\{\Lambda _{L_{\varrho }}\backslash \Lambda _{L}\}\cup \mathcal{Z}%
^{(\Lambda _{L},l)},\mathcal{Z}^{(\Lambda _{L},l)}}^{(\omega ,\mathcal{E}%
)}\right\vert =0,
\end{equation*}%
uniformly w.r.t. $\beta \in \mathbb{R}^{+}$, $\vartheta \in \left[
0,\vartheta _{1}\right] $, $\lambda \in \mathbb{R}_{0}^{+}$, $\omega \in
\Omega $ and $\mathcal{E}$ in any bounded set of $C_{0}^{0}([-n,n];\mathbb{R}%
^{d})$.
\end{lemma}

\begin{proof}
Fix all parameters of the lemma. Then, like for previous lemmata, we use
again Lemma \ref{lemma:suppor_non_int1} (ii) and (\ref{bilinear idiot}) to
obtain the bound%
\begin{multline*}
\left\vert g_{\mathcal{\{}\Lambda _{L}\mathcal{\}},\{\Lambda _{L_{\varrho
}}\backslash \Lambda _{L}\}\cup \mathcal{Z}^{(\Lambda _{L},l)},\mathcal{Z}%
^{(\Lambda _{L},l)}}^{(\omega ,\mathcal{E})}-g_{\mathcal{Z}^{(\Lambda
_{L},l)},\{\Lambda _{L_{\varrho }}\backslash \Lambda _{L}\}\cup \mathcal{Z}%
^{(\Lambda _{L},l)},\mathcal{Z}^{(\Lambda _{L},l)}}^{(\omega ,\mathcal{E}%
)}\right\vert \\
\leq \frac{1}{|\Lambda _{L}|}\sup_{\alpha \in \left[ 0,1\right] }\sup_{u\in %
\left[ -1/2,1/2\right] }\left\Vert \mathrm{e}^{u\left\langle \mathrm{A}%
,K_{\alpha }\mathrm{A}\right\rangle }\left\langle \mathrm{A},\left(
K_{1}-K_{0}\right) \mathrm{A}\right\rangle \mathrm{e}^{-u\left\langle
\mathrm{A},K_{\alpha }\mathrm{A}\right\rangle }\right\Vert _{\mathcal{U}},
\end{multline*}%
where
\begin{equation*}
K_{\alpha }\doteq \alpha K_{\{\Lambda _{L}\mathcal{\}},\mathcal{Z}^{(\Lambda
_{L},l)}}^{(\omega ,\mathcal{E})}+\left( 1-\alpha \right) K_{\mathcal{Z}%
^{(\Lambda _{L},l)},\mathcal{Z}^{(\Lambda _{L},l)}}^{(\omega ,\mathcal{E}%
)},\qquad \alpha \in \left[ 0,1\right] .
\end{equation*}%
Therefore, by Lemmata \ref{lemma:norm_estimate}, \ref{def bilineqr copy(1)}
and \ref{def bilineqr copy(3)}, the assertion follows.
\end{proof}

\noindent We are now in a position to prove Proposition \ref{Main
proposition}:\bigskip

\begin{proof}
Fix all parameters of Proposition \ref{Main proposition}. By\ Lemmata \ref%
{prop:equic copy(4)}-\ref{prop:equic copy(3)},%
\begin{equation}
\limsup_{L_{\tau }\geq L_{\varrho }\geq L\rightarrow \infty }\left\vert
\mathrm{J}_{\mathcal{\{}\Lambda _{L}\mathcal{\}},\mathcal{\{}\Lambda
_{L_{\varrho }}\mathcal{\}},\mathcal{\{}\Lambda _{L_{\tau }}\mathcal{\}}%
}^{(\omega ,\mathcal{E})}-\mathrm{J}_{\mathcal{Z}^{(\Lambda
_{L},l)},\{\Lambda _{L_{\varrho }}\backslash \Lambda _{L}\}\cup \mathcal{Z}%
^{(\Lambda _{L},l)},\mathcal{Z}^{(\Lambda _{L},l)}}^{(\omega ,\mathcal{E}%
)}\right\vert =0,  \label{idiot encore}
\end{equation}%
uniformly w.r.t. $\beta \in \left[ 0,\beta _{1}\right] $, $\vartheta \in %
\left[ 0,\vartheta _{1}\right] $, $\lambda \in \left[ 0,\lambda _{1}\right] $%
, $\omega \in \Omega $ and $\mathcal{E}$ in any bounded set of $%
C_{0}^{0}([-n,n];\mathbb{R}^{d})$. To conclude the proof, observe that%
\begin{equation}
\mathrm{J}_{\mathcal{Z}^{(\Lambda _{L},l)},\{\Lambda _{L_{\varrho
}}\backslash \Lambda _{L}\}\cup \mathcal{Z}^{(\Lambda _{L},l)},\mathcal{Z}%
^{(\Lambda _{L},l)}}^{(\omega ,\mathcal{E})}=\mathrm{J}_{\mathcal{Z}%
^{(\Lambda _{L},l)},\mathcal{Z}^{(\Lambda _{L},l)},\mathcal{Z}^{(\Lambda
_{L},l)}}^{(\omega ,\mathcal{E})}=\frac{1}{\left\vert \mathcal{Z}^{(\Lambda
_{L},l)}\right\vert }\sum_{Z\in \mathcal{Z}^{(\Lambda _{L},l)}}\mathrm{J}_{%
\mathcal{\{}Z\mathcal{\}},\mathcal{\{}Z\mathcal{\}},\mathcal{\{}Z\mathcal{\}}%
}^{(\omega ,\mathcal{E})}.  \label{punchline0}
\end{equation}%
This follows from the fact that the tracial state $\mathrm{tr}\in \mathcal{U}%
^{\ast }$ is a product of single-site states. See, e.g., \cite{Araki-Moriya}.
\end{proof}

\subsection{Akcoglu-Krengel Ergodic Theorem and Existence of Generating
Functions\label{Sectino tech5}}

For convenience, we shortly recall the Akcoglu-Krengel ergodic theorem. We
restrict ourselves to \emph{additive} processes associated with the
probability space $(\Omega ,\mathfrak{A}_{\Omega },\mathfrak{a}_{\Omega })$
defined in Section \ref{Section impurities}, even if the Akcoglu-Krengel
ergodic theorem holds for superadditive or subadditive ones (cf. \cite[%
Definition VI.1.6]{birkoff}).

\begin{definition}[Additive processes associated with random variables]
\label{Additive process}\mbox{ }\newline
$\{\mathfrak{F}^{(\omega )}\left( \Lambda \right) \}_{\Lambda \in \mathcal{P}%
_{\text{f}}(\mathbb{Z}^{d})}$ is an additive process associated with the
probability space $(\Omega ,\mathfrak{A}_{\Omega },\mathfrak{a}_{\Omega })$
if:\newline
\noindent (i) the map $\omega \mapsto \mathfrak{F}^{(\omega )}\left( \Lambda
\right) $ is bounded and measurable w.r.t. the $\sigma $-algebra $\mathfrak{A%
}_{\Omega }$ for any $\Lambda \in \mathcal{P}_{\text{f}}(\mathbb{Z}^{d})$.%
\newline
\noindent (ii) For all disjoint $\Lambda _{1},\Lambda _{2}\in \mathcal{P}_{%
\text{f}}(\mathbb{Z}^{d})$,
\begin{equation*}
\mathfrak{F}^{(\omega )}\left( \Lambda _{1}\cup \Lambda _{2}\right) =%
\mathfrak{F}^{(\omega )}\left( \Lambda _{1}\right) +\mathfrak{F}^{(\omega
)}\left( \Lambda _{2}\right) \ ,\qquad \omega \in \Omega \ .
\end{equation*}%
\noindent (iii) For all $\Lambda \in \mathcal{P}_{\text{f}}(\mathbb{Z}^{d})$
and any space shift $x\in \mathbb{Z}^{d}$,
\begin{equation}
\mathbb{E}\left[ \mathfrak{F}^{(\cdot )}\left( \Lambda \right) \right] =%
\mathbb{E}\left[ \mathfrak{F}^{(\cdot )}\left( x+\Lambda \right) \right] \ .
\label{equality a la con}
\end{equation}%
Recall that $\mathbb{E}[\ \cdot \ ]$ is the expectation value associated
with the distribution $\mathfrak{a}_{\Omega }$.
\end{definition}

\noindent We now define \emph{regular }sequences (cf. \cite[Remark VI.1.8]%
{birkoff}) as follows:

\begin{definition}[Regular sequences]
\label{regular sequences}\mbox{ }\newline
The non-decreasing sequence $(\Lambda ^{(L)})_{L\in \mathbb{N}}\subset
\mathcal{P}_{\text{f}}(\mathbb{Z}^{d})$ of (possibly non-cubic) boxes in $%
\mathbb{Z}^{d}$ is a regular sequence if there is a finite constant $D\in
(0,1]$ and a diverging sequence $(\ell _{L})_{L\in \mathbb{N}}\subset
\mathbb{N}$ such that $\Lambda ^{(L)}\subset \Lambda _{\ell _{L}}$ and $%
0<|\Lambda _{\ell _{L}}|\leq D|\Lambda ^{(L)}|$ for all $L\in \mathbb{N}$.
Here, $\Lambda _{\ell }$, $\ell \in \mathbb{R}^{+}$, is the family of boxes
defined by (\ref{eq:boxesl1}).
\end{definition}

Then, the form of Akcoglu-Krengel ergodic theorem we use in the sequel is
the lattice version of \cite[Theorem VI.1.7, Remark VI.1.8]{birkoff} for
additive processes associated with the probability space $(\Omega ,\mathfrak{%
A}_{\Omega },\mathfrak{a}_{\Omega })$:

\begin{theorem}[Akcoglu-Krengel ergodic theorem]
\label{Ackoglu-Krengel ergodic theorem II copy(1)}\mbox{
}\newline
Let $\{\mathfrak{F}^{(\omega )}\left( \Lambda \right) \}_{\Lambda \in
\mathcal{P}_{f}(\mathbb{Z}^{d})}$ be an additive process. Then, for any
regular sequence $(\Lambda ^{(L)})_{L\in \mathbb{N}}\subset \mathcal{P}_{%
\text{f}}(\mathbb{Z}^{d})$, there is a measurable subset $\tilde{\Omega}%
\subset \Omega $ of full measure such that, for all $\omega \in \tilde{\Omega%
}$,%
\begin{equation*}
\underset{L\rightarrow \infty }{\lim }\left\{ \left\vert \Lambda
^{(L)}\right\vert ^{-1}\mathfrak{F}^{(\omega )}\left( \Lambda ^{(L)}\right)
\right\} =\mathbb{E}\left[ \mathfrak{F}^{(\cdot )}\left( \left\{ 0\right\}
\right) \right] \ .
\end{equation*}
\end{theorem}

\noindent See also \cite{bikfoffbis}.

The Ackoglu-Krengel (superadditive) ergodic theorem, cornerstone of ergodic
theory, generalizes the celebrated Birkhoff additive ergodic theorem. It is
used to deduce, via Proposition \ref{prop:equic}, the following Corollary:

\begin{corollary}[Akcoglu-Krengel ergodic theorem for generating functions]
\label{Ackoglu-Krengel ergodic theorem II}\mbox{
}\newline
There is a measurable subset $\tilde{\Omega}\subset \Omega $ of full measure
such that, for all $\beta \in \mathbb{R}^{+}$, $\vartheta ,\lambda \in
\mathbb{R}_{0}^{+}$, $\omega \in \tilde{\Omega}$, $l\in \mathbb{N}$, $%
\mathcal{E}\in C_{0}^{0}(\mathbb{R};\mathbb{R}^{d})$ and $\vec{w}\in {%
\mathbb{R}}^{d}$ with $\left\Vert \vec{w}\right\Vert _{\mathbb{R}^{d}}=1$,%
\begin{equation*}
\lim_{L\rightarrow \infty }\frac{1}{\left\vert \mathcal{Z}^{(\Lambda
_{L},l)}\right\vert }\sum_{Z\in \mathcal{Z}^{(\Lambda _{L},l)}}\mathrm{J}_{%
\mathcal{\{}Z\mathcal{\}},\mathcal{\{}Z\mathcal{\}},\mathcal{\{}Z\mathcal{\}}%
}^{(\omega ,\mathcal{E})}=\mathbb{E}\left[ \mathrm{J}_{\mathcal{\{}\Lambda
_{l}\mathcal{\}},\mathcal{\{}\Lambda _{l}\mathcal{\}},\mathcal{\{}\Lambda
_{l}\mathcal{\}}}^{(\cdot ,\mathcal{E})}\right] .
\end{equation*}
\end{corollary}

\begin{proof}
Fix $\beta \in \mathbb{R}^{+}$, $\vartheta ,\lambda \in \mathbb{R}_{0}^{+}$,
$\omega \in \Omega $, $l\in \mathbb{N}$, $\mathcal{E}\in C_{0}^{0}(\mathbb{R}%
;\mathbb{R}^{d})$ and $\vec{w}\in {\mathbb{R}}^{d}$ with $\left\Vert \vec{w}%
\right\Vert _{\mathbb{R}^{d}}=1$. For any $\Gamma \in \mathcal{P}_{\text{f}}(%
\mathbb{Z}^{d})$, let%
\begin{equation*}
\mathfrak{F}_{l}^{(\omega ,\mathcal{E})}\left( \Gamma \right) \doteq
\sum_{x\in \Gamma }\mathrm{J}_{\mathcal{\{}\Lambda _{l}+\left( 2l+1\right) x%
\mathcal{\}},\mathcal{\{}\Lambda _{l}+\left( 2l+1\right) x\mathcal{\}},%
\mathcal{\{}\Lambda _{l}+\left( 2l+1\right) x\mathcal{\}}}^{(\omega ,%
\mathcal{E})}.
\end{equation*}%
Then, if
\begin{equation*}
\Lambda ^{(L)}\equiv \Lambda ^{(L,l)}\doteq \left\{ x\in \mathbb{Z}%
^{d}\colon (\Lambda _{l}+\left( 2l+1\right) x)\subset \Lambda _{L}\right\}
\subset \Lambda _{L},
\end{equation*}%
observe that%
\begin{equation*}
\left\vert \Lambda ^{(L)}\right\vert ^{-1}\mathfrak{F}_{l}^{(\omega ,%
\mathcal{E})}\left( \Lambda ^{(L)}\right) =\frac{1}{\left\vert \mathcal{Z}%
^{(\Lambda _{L},l)}\right\vert }\sum_{Z\in \mathcal{Z}^{(\Lambda _{L},l)}}%
\mathrm{J}_{\mathcal{\{}Z\mathcal{\}},\mathcal{\{}Z\mathcal{\}},\mathcal{\{}Z%
\mathcal{\}}}^{(\omega ,\mathcal{E})}.
\end{equation*}%
Therefore, since $(\Lambda ^{(L)})_{L\in \mathbb{N}}$ is clearly a regular
sequence, by Theorem \ref{Ackoglu-Krengel ergodic theorem II copy(1)}, for
any $\beta \in \mathbb{R}^{+}$, $\vartheta ,\lambda \in \mathbb{R}_{0}^{+}$,
$l\in \mathbb{N}$, $\mathcal{E}\in C_{0}^{0}(\mathbb{R};\mathbb{R}^{d})$ and
$\vec{w}\in {\mathbb{R}}^{d}$ with $\left\Vert \vec{w}\right\Vert _{\mathbb{R%
}^{d}}=1$, there is a measurable subset $\hat{\Omega}\equiv \hat{\Omega}%
^{(\beta ,\vartheta ,\lambda ,l,\mathcal{E},\vec{w})}\subset \Omega $ of
full measure such that, for all $\omega \in \hat{\Omega}$,%
\begin{equation*}
\lim_{L\rightarrow \infty }\frac{1}{\left\vert \mathcal{Z}^{(\Lambda
_{L},l)}\right\vert }\sum_{Z\in \mathcal{Z}^{(\Lambda _{L},l)}}\mathrm{J}_{%
\mathcal{\{}Z\mathcal{\}},\mathcal{\{}Z\mathcal{\}},\mathcal{\{}Z\mathcal{\}}%
}^{(\omega ,\mathcal{E})}=\mathbb{E}\left[ \mathrm{J}_{\mathcal{\{}\Lambda
_{l}\mathcal{\}},\mathcal{\{}\Lambda _{l}\mathcal{\}},\mathcal{\{}\Lambda
_{l}\mathcal{\}}}^{(\cdot ,\mathcal{E})}\right] .
\end{equation*}%
Observe that, for any $n\in \mathbb{N}$, there is a countable dense set $%
\mathcal{D}_{n}\subset C_{0}^{0}(\mathbb{R};\mathbb{R}^{d})$. Let $\mathbb{S}%
^{d-1}$ be a dense countable subset of the $(d-1)$-dimensional sphere.
Hence, by Proposition \ref{prop:equic}, we arrive at the assertion for any
realization $\omega \in \tilde{\Omega}\subset \Omega $, where%
\begin{equation*}
\tilde{\Omega}\doteq \underset{\vartheta ,\lambda \in \mathbb{Q}\cap \mathbb{%
R}_{0}^{+}}{\bigcap }\underset{\beta \in \mathbb{Q}\cap \mathbb{R}^{+}}{%
\bigcap }\underset{\vec{w}\in \mathbb{S}^{d-1}}{\bigcap }\underset{n\in
\mathbb{N}}{\bigcap }\underset{\mathcal{E}\in \mathcal{D}_{n}}{\bigcap }%
\underset{l\in \mathbb{N}}{\bigcap }\hat{\Omega}^{(\beta ,\vartheta ,\lambda
,l,\mathcal{E},\vec{w})}\ .
\end{equation*}%
[Recall that any countable intersection of measurable sets of full measure
has full measure].
\end{proof}

\begin{corollary}[Almost surely existence of generating functions]
\label{Ackoglu-Krengel ergodic theorem II copy(2)}\mbox{
}\newline
Let $\tilde{\Omega}\subset \Omega $ be the measurable subset of Corollary %
\ref{Ackoglu-Krengel ergodic theorem II}. Then, for all $\beta \in \mathbb{R}%
^{+}$, $\vartheta ,\lambda \in \mathbb{R}_{0}^{+}$, $\omega \in \tilde{\Omega%
}$, $l\in \mathbb{N}$, $\mathcal{E}\in C_{0}^{0}(\mathbb{R};\mathbb{R}^{d})$
and $\vec{w}\in {\mathbb{R}}^{d}$ with $\left\Vert \vec{w}\right\Vert _{%
\mathbb{R}^{d}}=1$,
\begin{equation*}
\lim_{L\to \infty }\frac{1}{\left\vert \Lambda _{L}\right\vert }\mathbb{E}%
\left[ \ln \varrho ^{(\cdot )}\left( \mathrm{e}^{\left\vert \Lambda
_{L}\right\vert \mathbb{I}_{\Lambda _{L}}^{(\cdot ,\mathcal{E})}}\right) %
\right] =\lim_{L_{\tau }\geq L_{\varrho }\geq L\to \infty }\mathrm{J}_{%
\mathcal{\{}\Lambda _{L}\mathcal{\}},\mathcal{\{}\Lambda _{L_{\varrho }}%
\mathcal{\}},\mathcal{\{}\Lambda _{L_{\tau }}\mathcal{\}}}^{(\omega ,%
\mathcal{E})}\doteq \mathrm{J}^{(\mathcal{E})}.
\end{equation*}%
For all $n\in \mathbb{N}$, the convergence is uniform w.r.t. $\beta
,\vartheta ,\lambda $ in compact sets, $\omega \in \tilde{\Omega}$, $\vec{w}%
\in {\mathbb{R}}^{d}$ with $\left\Vert \vec{w}\right\Vert _{\mathbb{R}%
^{d}}=1 $ and $\mathcal{E}$ in any bounded set of $C_{0}^{0}([-n,n];\mathbb{R%
}^{d})$.
\end{corollary}

\begin{proof}
By translation invariance of the distribution $\mathfrak{a}_{\Omega }$,
\begin{equation*}
\mathbb{E}\left[ \mathrm{J}_{\mathcal{\{}\Lambda _{l}\mathcal{\}},\mathcal{\{%
}\Lambda _{l}\mathcal{\}},\mathcal{\{}\Lambda _{l}\mathcal{\}}}^{(\cdot ,%
\mathcal{E})}\right] =\mathbb{E}\left[ \frac{1}{\left\vert \mathcal{Z}%
^{(\Lambda _{L},l)}\right\vert }\sum_{Z\in \mathcal{Z}^{(\Lambda _{L},l)}}%
\mathrm{J}_{\mathcal{\{}Z\mathcal{\}},\mathcal{\{}Z\mathcal{\}},\mathcal{\{}Z%
\mathcal{\}}}^{(\cdot ,\mathcal{E})}\right] .
\end{equation*}%
Hence,
\begin{equation*}
\left\{ \mathbb{E}\left[ \mathrm{J}_{\mathcal{\{}\Lambda _{l}\mathcal{\}},%
\mathcal{\{}\Lambda _{l}\mathcal{\}},\mathcal{\{}\Lambda _{l}\mathcal{\}}%
}^{(\cdot ,\mathcal{E})}\right] \right\} _{l\in \mathbb{N}}
\end{equation*}%
is a Cauchy sequence, by (\ref{idiot encore}) and (\ref{punchline0}). By
Proposition \ref{Main proposition} and Corollary \ref{Ackoglu-Krengel
ergodic theorem II}, there is a measurable subset $\tilde{\Omega}\subset
\Omega $ of full measure such that, for all $\beta \in \mathbb{R}^{+}$, $%
\vartheta ,\lambda \in \mathbb{R}_{0}^{+}$, $\omega \in \tilde{\Omega}$, $%
l\in \mathbb{N}$, $\mathcal{E}\in C_{0}^{0}(\mathbb{R};\mathbb{R}^{d})$ and $%
\vec{w}\in {\mathbb{R}}^{d}$ with $\left\Vert \vec{w}\right\Vert _{\mathbb{R}%
^{d}}=1$,%
\begin{equation*}
\lim_{L_{\tau }\geq L_{\varrho }\geq L\to \infty }\mathrm{J}_{\mathcal{\{}%
\Lambda _{L}\mathcal{\}},\mathcal{\{}\Lambda _{L_{\varrho }}\mathcal{\}},%
\mathcal{\{}\Lambda _{L_{\tau }}\mathcal{\}}}^{(\omega ,\mathcal{E}%
)}=\lim_{l\to \infty }\mathbb{E}\left[ \mathrm{J}_{\mathcal{\{}\Lambda _{l}%
\mathcal{\}},\mathcal{\{}\Lambda _{l}\mathcal{\}},\mathcal{\{}\Lambda _{l}%
\mathcal{\}}}^{(\cdot ,\mathcal{E})}\right] .
\end{equation*}%
For all $n\in \mathbb{N}$, the convergence is uniform w.r.t. $\beta
,\vartheta ,\lambda $ in compact sets, $\omega \in \tilde{\Omega}$, $\vec{w}%
\in {\mathbb{R}}^{d}$ with $\left\Vert \vec{w}\right\Vert _{\mathbb{R}%
^{d}}=1 $ and $\mathcal{E}$ in any bounded set of $C_{0}^{0}([-n,n];\mathbb{R%
}^{d})$. By (\ref{limit cool}), the assertion then follows.
\end{proof}

\begin{corollary}[Differentiability of generating functions]
\label{final corrollary}\mbox{
}\newline
Fix $\beta ,\lambda ,\vartheta \in \mathbb{R}^{+}$ and $\vec{w}\in {\mathbb{R%
}}^{d}$ with $\left\Vert \vec{w}\right\Vert _{\mathbb{R}^{d}}=1$. For any $%
\mathcal{E}\in C_{0}^{0}(\mathbb{R};\mathbb{R}^{d})$, the map $s\mapsto
\mathrm{J}^{(s\mathcal{E})}$ from ${\mathbb{R}}$ to itself is continuously
differentiable, so that%
\begin{equation}
\partial _{s}\mathrm{J}^{(s\mathcal{E})}=\lim_{L\rightarrow \infty }\frac{%
\varrho ^{(\omega )}\left( \mathbb{I}_{\Lambda _{L}}^{(\omega ,\mathcal{E})}%
\mathrm{e}^{s\left\vert \Lambda _{L}\right\vert \mathbb{I}_{\Lambda
_{L}}^{(\omega ,\mathcal{E})}}\right) }{\varrho ^{(\omega )}\left( \mathrm{e}%
^{s\left\vert \Lambda _{L}\right\vert \mathbb{I}_{\Lambda _{L}}^{(\omega ,%
\mathcal{E})}}\right) }.  \label{current utiles}
\end{equation}
\end{corollary}

\begin{proof}
Take any $\mathcal{E}\in C_{0}^{0}(\mathbb{R};\mathbb{R}^{d})$ and $\omega
\in \tilde{\Omega}$. See Corollary \ref{Ackoglu-Krengel ergodic theorem II
copy(2)}. Then, for any $s\in {\mathbb{R}}$,
\begin{equation*}
\mathrm{J}^{(s\mathcal{E})}=\lim_{L_{\tau }\geq L_{\varrho }\geq
L\rightarrow \infty }\mathrm{J}_{\mathcal{\{}\Lambda _{L}\mathcal{\}},%
\mathcal{\{}\Lambda _{L_{\varrho }}\mathcal{\}},\mathcal{\{}\Lambda
_{L_{\tau }}\mathcal{\}}}^{(\omega ,s\mathcal{E})}.
\end{equation*}%
By Proposition \ref{prop:equic copy(5)} combined with the mean value theorem
and the (Arzel\`{a}-) Ascoli theorem \cite[Theorem A5]{Rudin}, there are
three sequences $(L_{\tau }^{(n)})_{n\in \mathbb{N}}$, $(L_{\varrho
}^{(n)})_{n\in \mathbb{N}}$, $(L^{(n)})_{n\in \mathbb{N}}\subset \mathbb{R}%
_{0}^{+}$, with $L_{\tau }^{(n)}\geq L_{\varrho }^{(n)}\geq L^{(n)}$, such
that the maps
\begin{equation*}
s\mapsto \mathrm{J}_{\mathcal{\{}\Lambda _{L^{(n)}}\mathcal{\}},\mathcal{\{}%
\Lambda _{L_{\varrho }^{(n)}}\mathcal{\}},\mathcal{\{}\Lambda _{L_{\tau
}^{(n)}}\mathcal{\}}}^{(\omega ,s\mathcal{E})}\qquad \text{and}\qquad
s\mapsto \partial _{s}\mathrm{J}_{\mathcal{\{}\Lambda _{L^{(n)}}\mathcal{\}},%
\mathcal{\{}\Lambda _{L_{\varrho }^{(n)}}\mathcal{\}},\mathcal{\{}\Lambda
_{L_{\tau }^{(n)}}\mathcal{\}}}^{(\omega ,s\mathcal{E})}
\end{equation*}%
converge uniformly for $s$ in any compact set of $\mathbb{R}$. In
particular, the map $s\mapsto \mathrm{J}^{(s\mathcal{E})}$ from ${\mathbb{R}}
$ to itself is continuously differentiable with
\begin{equation*}
\partial _{s}\mathrm{J}^{(s\mathcal{E})}=\lim_{L_{\tau }\geq L_{\varrho
}\geq L\rightarrow \infty }\partial _{s}\mathrm{J}_{\mathcal{\{}\Lambda
_{L^{(n)}}\mathcal{\}},\mathcal{\{}\Lambda _{L_{\varrho }^{(n)}}\mathcal{\}},%
\mathcal{\{}\Lambda _{L_{\tau }^{(n)}}\mathcal{\}}}^{(\omega ,s\mathcal{E}%
)}=\lim_{L\rightarrow \infty }\frac{\varrho ^{(\omega )}\left( \mathbb{I}%
_{\Lambda _{L}}^{(\omega ,\mathcal{E})}\mathrm{e}^{s\left\vert \Lambda
_{L}\right\vert \mathbb{I}_{\Lambda _{L}}^{(\omega ,\mathcal{E})}}\right) }{%
\varrho ^{(\omega )}\left( \mathrm{e}^{s\left\vert \Lambda _{L}\right\vert
\mathbb{I}_{\Lambda _{L}}^{(\omega ,\mathcal{E})}}\right) }.
\end{equation*}
\end{proof}

\appendix

\section{Combes-Thomas Estimates\label{Section Combes-Thomas Estimate}}

For any operator $h\in \mathcal{B}(\mathfrak{h})$ and $\mu \in \mathbb{R}%
_{0}^{+}$, let
\begin{equation}
\mathbf{S}_{0}(h,\mu )\doteq \sup_{x\in \mathbb{Z}^{d}}\sum_{y\in \mathbb{Z}%
^{d}}\mathrm{e}^{\mu |x-y|}\left\vert \left\langle \mathfrak{e}_{x},h%
\mathfrak{e}_{y}\right\rangle _{\mathfrak{h}}\right\vert \in \mathbb{R}%
_{0}^{+}\cup \left\{ \infty \right\} .  \label{S0}
\end{equation}%
Note that%
\begin{equation}
\mathbf{S}_{0}(h_{1}h_{2},\mu )\leq \mathbf{S}_{0}(h_{1},\mu )\mathbf{S}%
_{0}(h_{2},\mu ),  \label{produit}
\end{equation}%
for any $h_{1},h_{2}\in \mathcal{B}(\mathfrak{h})$ and $\mu \in \mathbb{R}%
_{0}^{+}$. In particular, for any $z\in \mathbb{C}$, $h\in \mathcal{B}(%
\mathfrak{h})$ and $\mu \in \mathbb{R}_{0}^{+}$,
\begin{equation}
\mathbf{S}_{0}(\mathrm{e}^{zh},\mu )\leq \mathrm{e}^{\mathbf{S}_{0}(zh,\mu
)}=\mathrm{e}^{|z|\mathbf{S}_{0}(h,\mu )}  \label{produit2}
\end{equation}%
and hence,%
\begin{equation*}
\left\vert \left\langle \mathfrak{e}_{x},\mathrm{e}^{zh}\mathfrak{e}%
_{y}\right\rangle _{\mathfrak{h}}\right\vert \leq \mathrm{e}^{|z|\mathbf{S}%
_{0}(h,\mu )}\mathrm{e}^{-\mu |x-y|}.
\end{equation*}

The above bound can be sharpened if $z=it$ is imaginary by using the
Combes-Thomas estimate, first proven in \cite{CT73}. We give a version of
this estimate that is adapted to the present setting: Given a self-adjoint
operator $h=h^{\ast }\in \mathcal{B}(\mathfrak{h})$ whose spectrum is
denoted by $\mathrm{spec}(h)$, we define the constants
\begin{equation}
\mathbf{S}(h,\mu )\doteq \sup_{x\in \mathbb{Z}^{d}}\sum_{y\in \mathbb{Z}%
^{d}}\left( \mathrm{e}^{\mu |x-y|}-1\right) \left\vert \left\langle
\mathfrak{e}_{x},h\mathfrak{e}_{y}\right\rangle _{\mathfrak{h}}\right\vert
\in \mathbb{R}_{0}^{+}\cup \left\{ \infty \right\} ,  \label{S}
\end{equation}%
for $\mu \in \mathbb{R}_{0}^{+}$, and
\begin{equation*}
\Delta (h,z)\doteq \inf \left\{ \left\vert z-\lambda \right\vert \colon
\lambda \in \mathrm{spec}(h)\right\} ,\qquad z\in \mathbb{C},
\end{equation*}%
as being the distance from the point $z$ to the spectrum of $h$. Since the
function $x\mapsto (e^{xr}-1)/x$ is increasing on $\mathbb{R}^{+}$ for any
fixed $r\geq 0$, it follows that
\begin{equation}
\mathbf{S}(h,\mu _{1})\leq \frac{\mu _{1}}{\mu _{2}}\mathbf{S}(h,\mu _{2})\
,\qquad \mu _{2}\geq \mu _{1}\geq 0.  \label{inequality combes easy}
\end{equation}%
The version of the Combes-Thomas estimate that is most convenient for the
current study is the following:

\begin{theorem}[Combes-Thomas]
\label{Combes-Thomas}\mbox{ }\newline
Let $h=h^{\ast }\in \mathcal{B}(\mathfrak{h})$, $\mu \in \mathbb{R}_{0}^{+}$
and $z\in \mathbb{C}$. If $\Delta (h,z)>\mathbf{S}(h,\mu )$ then, for all $%
x,y\in \mathbb{Z}^{d}$,
\begin{equation*}
\left\vert \left\langle \mathfrak{e}_{x},(z-h)^{-1}\mathfrak{e}%
_{y}\right\rangle \right\vert \leq \frac{\mathrm{e}^{-\mu |x-y|}}{\Delta
(h,z)-\mathbf{S}(h,\mu )}.
\end{equation*}
\end{theorem}

\begin{proof}
This theorem is an instance of the first part of \cite[Theorem 10.5]%
{AizenmanWarzel} and is proven in the same way.
\end{proof}

\noindent The Combes-Thomas estimate yields the following bound \cite[Lemma 3%
]{AG98}:

\begin{proposition}[Bound on differences of resolvents]
\label{Lemma AG98}\mbox{ }\newline
Let $h=h^{\ast }\in \mathcal{B}(\mathfrak{h})$, $\mu \in \mathbb{R}_{0}^{+}$
and $\eta \in \mathbb{R}^{+}$ such that $\mathbf{S}(h,\mu )\leq \eta /2$.
Then, for all $x,y\in \mathbb{Z}^{d}$ and $u\in \mathbb{R}$,
\begin{eqnarray*}
&&\left\vert \left\langle \mathfrak{e}_{x},((h-u)^{2}+\eta ^{2})^{-1}%
\mathfrak{e}_{y}\right\rangle _{\mathfrak{h}}\right\vert \\
&\leq &12\mathrm{e}^{-\mu |x-y|}\left\langle \mathfrak{e}_{x},((h-u)^{2}+%
\eta ^{2})^{-1}\mathfrak{e}_{x}\right\rangle _{\mathfrak{h}%
}^{1/2}\left\langle \mathfrak{e}_{y},((h-u)^{2}+\eta ^{2})^{-1}\mathfrak{e}%
_{y}\right\rangle _{\mathfrak{h}}^{1/2}.
\end{eqnarray*}
\end{proposition}

We are now in a position to prove the space decay of propagators:

\begin{corollary}[Space decay of propagators - I]
\label{Lemma fermi1}\mbox{ }\newline
For any self-adjoint operator $h=h^{\ast }\in \mathcal{B}(\mathfrak{h})$, $%
\eta ,\mu \in \mathbb{R}^{+}$, all $x,y\in \mathbb{Z}^{d}$ and $t\in \mathbb{%
R}$,
\begin{equation*}
\left\vert \left\langle \mathfrak{e}_{x},\mathrm{e}^{ith}\mathfrak{e}%
_{y}\right\rangle _{\mathfrak{h}}\right\vert \leq 36\exp \left( \left\vert
t\eta \right\vert -\mu \min \left\{ 1,\frac{\eta }{2\mathbf{S}(h,\mu )}%
\right\} |x-y|\right) .
\end{equation*}
\end{corollary}

\begin{proof}
The proof is a simple adaptation of the one from \cite[Theorem 3]{AG98}: Fix
all parameters of the lemma and observe that Proposition \ref{Lemma AG98}
combined with Inequality (\ref{inequality combes easy}) yields%
\begin{eqnarray}
&&\left\vert \left\langle \mathfrak{e}_{x},((h-u)^{2}+\eta ^{2})^{-1}%
\mathfrak{e}_{y}\right\rangle _{\mathfrak{h}}\right\vert
\label{Lemma AG982} \\
&\leq &12\mathrm{e}^{-\frac{\mu \eta }{2\mathbf{S}(h,\mu )}%
|x-y|}\left\langle \mathfrak{e}_{x},((h-u)^{2}+\eta ^{2})^{-1}\mathfrak{e}%
_{x}\right\rangle _{\mathfrak{h}}^{1/2}\left\langle \mathfrak{e}%
_{y},((h-u)^{2}+\eta ^{2})^{-1}\mathfrak{e}_{y}\right\rangle _{\mathfrak{h}%
}^{1/2}  \notag
\end{eqnarray}%
for $x,y\in \mathbb{Z}^{d}$, $u\in \mathbb{R}$ and $\eta \in \mathbb{R}^{+}$%
. On the other hand, at fixed $\eta \in \mathbb{R}^{+}$, the function
defined by $G\left( z\right) \doteq \mathrm{e}^{itz}$ on the stripe
\begin{equation*}
\mathbb{R}+i\eta \left[ -1,1\right] \subset \mathbb{C}
\end{equation*}%
is analytic and uniformly bounded by $\mathrm{e}^{\left\vert t\eta
\right\vert }$. Using Cauchy's integral formula and translations by $\pm
i\eta $ of the integration variable, $u$, we write the function $G$ as
\begin{eqnarray}
G\left( E\right)  &=&\frac{1}{2\pi i}\int_{\mathbb{R}}\left( \frac{G\left(
u-i\eta \right) }{u-i\eta -E}-\frac{G\left( u+i\eta \right) }{u+i\eta -E}%
\right) \mathrm{d}u  \notag \\
&=&\frac{\eta }{\pi }\int_{\mathbb{R}}\frac{G\left( u-i\eta \right) +G\left(
u+i\eta \right) }{\left( E-u\right) ^{2}+\eta ^{2}}\mathrm{d}u-\frac{2\eta }{%
\pi }\int_{\mathbb{R}}\frac{G\left( u\right) }{\left( E-u\right) ^{2}+4\eta
^{2}}\mathrm{d}u  \label{FE}
\end{eqnarray}%
for all $E\in \mathbb{R}$ and $\eta \in \mathbb{R}^{+}$. By spectral
calculus, together with (\ref{Lemma AG982})-(\ref{FE}) and the
Cauchy-Schwarz inequality, the assertion follows.
\end{proof}

\begin{corollary}[Space decay of propagators - II]
\label{Lemma fermi1 copy(1)}\mbox{ }\newline
For any self-adjoint operators $h_{1},h_{2}\in \mathcal{B}(\mathfrak{h})$
and all $x,y\in \mathbb{Z}^{d}$,%
\begin{equation*}
\left\vert \left\langle \mathfrak{e}_{x},\frac{1}{1+\mathrm{e}^{h_{2}}%
\mathrm{e}^{h_{1}}\mathrm{e}^{h_{2}}}\mathfrak{e}_{y}\right\rangle _{%
\mathfrak{h}}\right\vert \leq 2\inf_{\mu \in \mathbb{R}_{0}^{+}}\exp \left( -%
\frac{\mu }{2}\mathrm{e}^{-\mathbf{S}_{0}(h_{1},\mu )-2\mathbf{S}%
_{0}(h_{2},\mu )}|x-y|\right) .
\end{equation*}
\end{corollary}

\begin{proof}
By (\ref{S0})-(\ref{S}), note that, for any $\mu \in \mathbb{R}_{0}^{+}$,
\begin{equation*}
\mathbf{S}(\mathrm{e}^{h_{2}}\mathrm{e}^{h_{1}}\mathrm{e}^{h_{2}},\mu )\leq
\mathbf{S}_{0}(\mathrm{e}^{h_{2}}\mathrm{e}^{h_{1}}\mathrm{e}^{h_{2}},\mu
)\leq \mathrm{e}^{\mathbf{S}_{0}(h_{1},\mu )+2\mathbf{S}_{0}(h_{2},\mu )}.
\end{equation*}%
Fix $\mu \in \mathbb{R}_{0}^{+}$ and define
\begin{equation*}
\mu _{1}\doteq \frac{\mu }{2}\mathrm{e}^{-\mathbf{S}_{0}(h_{1},\mu )-2%
\mathbf{S}_{0}(h_{2},\mu )}.
\end{equation*}%
By (\ref{inequality combes easy}), $\mathbf{S}(\mathrm{e}^{h_{2}}\mathrm{e}%
^{h_{1}}\mathrm{e}^{h_{2}},\mu _{1})<1/2$. Meanwhile, by using Theorem \ref%
{Combes-Thomas} with $h=\mathrm{e}^{h_{2}}\mathrm{e}^{h_{1}}\mathrm{e}%
^{h_{2}}\geq 0$,
\begin{equation*}
\left\vert \left\langle \mathfrak{e}_{x},\frac{1}{1+\mathrm{e}^{h_{2}}%
\mathrm{e}^{h_{1}}\mathrm{e}^{h_{2}}}\mathfrak{e}_{y}\right\rangle _{%
\mathfrak{h}}\right\vert \leq 2\mathrm{e}^{-\mu _{1}|x-y|}.
\end{equation*}
\end{proof}

\section{Large Deviation Formalism\label{sec:LDP}}

In probability theory, the large deviation (LD) formalism quantitatively
describes, for large $n\gg 1$, the probability of finding an empirical mean
that differs from the expected value, by more than some fixed amount. That's
the reason is why we apply it in\ Section \ref{sec:main} to prove the
exponentially fast convergence of microscopic current densities towards
their (classical) macroscopic values. For completeness, in this appendix, we
present the main result from LD theory used in the current study, namely,
the G\"{a}rtner-Ellis theorem (Theorem \ref{prop Gartner--Ellis} below). For
more details, see \cite{DS89,dembo1998large}. For a historical review of LD
in quantum statistical mechanics, see \cite[Section 7.1]{ABPM1}.

Let $\mathcal{X}$ denote a topological vector space. A lower semi-continuous
function $\mathrm{I}:\mathcal{X}\rightarrow \lbrack 0,\infty ]$ is called a
\textit{good rate function} if $\mathrm{I}$ is not identically $\infty $ and
has compact level sets, i.e., $\mathrm{I}^{-1}([0,m])=\{x\in \mathcal{X}:%
\mathrm{I}(x)\leq m\}$ is compact for any $m\geq 0$. A sequence $%
(X_{L})_{L\in \mathbb{N}}$ of $\mathcal{X}$-valued random variables
satisfies the \textit{LD upper bound} with \textit{speed }$(\mathfrak{n}%
_{L})_{L\in \mathbb{N}}\subset \mathbb{R}^{+}$ (a positive, increasing and
divergent sequence) and rate function $\mathrm{I}$ if, for any closed subset
$F$ of ${\mathcal{X}}$,%
\begin{equation}
\limsup_{L\rightarrow \infty }\frac{1}{\mathfrak{n}_{L}}\ln \mathbb{P}%
(X_{l}\in F)\leq -\inf_{x\in F}\mathrm{I}(x),  \label{LDP upp}
\end{equation}%
and it satisfies the \textit{LD lower bound} if, for any open subset $G$ of $%
{\mathcal{X}}$,
\begin{equation}
\liminf_{L\rightarrow \infty }\frac{1}{\mathfrak{n}_{L}}\ln \mathbb{P}%
(X_{l}\in G)\geq -\inf_{x\in G}\mathrm{I}\left( x\right) .  \label{LDP lower}
\end{equation}%
If both, upper and lower bound, are satisfied, one says that $(X_{L})_{L\in
\mathbb{N}}$ satisfies an \emph{LD}\textit{\ principle} (LDP). The principle
is called \textit{weak} if the upper bound in (\ref{LDP upp}) holds only for
\textit{compact }sets $F$.

A weak LDP can be strengthened to a full one by showing that the sequence $%
(X_{L})_{L\in \mathbb{N}}$ of distributions is \textit{exponentially tight},
i.e., if for any $\alpha \in \mathbb{R}$, there is a compact subset $%
\mathcal{G}_{\alpha }$ of ${\mathcal{X}}$ such that%
\begin{equation}
\limsup_{L\rightarrow \infty }\frac{1}{\mathfrak{n}_{L}}\ln \mathbb{P}%
(X_{L}\in {\mathcal{X}}\backslash \mathcal{G}_{\alpha })<-\alpha \mathfrak{.}
\label{eq:exp_tight}
\end{equation}%
If $\mathcal{X}$ is a locally compact topological space, i.e., every point
possesses a compact neighborhood, then the existence of an LDP with a good
rate function $\mathrm{I}$ for the sequence $(X_{L})_{L\in \mathbb{N}}$
implies its exponential tightness \cite[Exercise 1.2.19]{dembo1998large}.

A sufficient condition to ensure that a sequence $(X_{L})_{L\in \mathbb{N}}$
of $\mathcal{X}$-valued random variables satisfies an LDP is given by the G%
\"{a}rtner-Ellis theorem. It says \cite[Corollary 4.5.27]{dembo1998large}
that an exponentially tight sequence $(X_{L})_{L\in \mathbb{N}}$ of $%
\mathcal{X}$-valued random variables on a Banach space $\mathcal{X}$
satisfies an LDP with the good rate function
\begin{equation}
\mathrm{I}\left( x\right) =\sup\limits_{s\in \mathcal{X}^{\ast }}\left\{
s\left( x\right) -J(s)\right\} \ ,\qquad x\in \mathcal{X},
\label{rate function}
\end{equation}%
whenever the so-called limiting logarithmic moment generating function
\begin{equation}
J(s)\doteq \lim_{L\rightarrow \infty }\frac{1}{\mathfrak{n}_{L}}\ln \mathbb{E%
}\left[ \mathrm{e}^{\mathfrak{n}_{L}s\left( X_{L}\right) }\right] \ ,\qquad
s\in \mathcal{X}^{\ast },  \label{eq:state_LDP}
\end{equation}%
exists as a Gateaux differentiable and weak$^{\ast }$ lower semi-continuous
(finite-valued) function on the dual space $\mathcal{X}^{\ast }$. See also
\cite[Theorem 2.2.4]{DS89}.

The random variables we study in this paper result from bounded sequences $%
(A_{L})_{L\in \mathbb{N}}\subset \mathcal{U}$ of self-adjoint elements of
the CAR $C^{\ast }$-algebra $\mathcal{U}$ along with some fixed state $\rho
\in \mathcal{U}^{\ast }$. In Section \ref{sec:main}, we explain how such a
sequence and state naturally define an exponentially tight sequence of
random variables on the real line $\mathcal{X}=\mathbb{R}$, via the
Riesz-Markov theorem and functional calculus (cf. (\ref{fluctuation measure}%
)). The following simple version of the celebrated G\"{a}rtner-Ellis theorem
of LD theory is sufficient for our purposes:

\begin{theorem}[G\"{a}rtner-Ellis]
\label{prop Gartner--Ellis}\mbox{ }\newline
Take any exponentially tight sequence $(X_{L})_{L\in \mathbb{N}}$ of
real-valued random variables (i.e., $\mathcal{X=X}^{\ast }=\mathbb{R}$) and
assume that the limiting logarithmic moment generating function $J$ defined
by (\ref{eq:state_LDP}) exists for all $s\in \mathbb{R}$. Then: \newline
\emph{(LD1)} $(X_{L})_{L\in \mathbb{N}}$ satisfies the LD upper bound (\ref%
{LDP upp}) with rate function $\mathrm{I}$ given by (\ref{rate function}).
\newline
\emph{(LD2)} If, additionally, $\mathrm{J}$ is differentiable for all $s\in
\mathbb{R}$ then $(X_{L})_{L\in \mathbb{N}}$ satisfies the LD lower bound (%
\ref{LDP lower}) with good rate function $\mathrm{I}$ given again by (\ref%
{rate function}).
\end{theorem}

\begin{proof}
(LD1) and (LD2) are special cases of \cite[Theorem V.6.(a) and (c)]{denHoll}%
, respectively.
\end{proof}

\section{Response of Quasi-Free Fermion Systems to\ Electric Fields\label%
{linear_response}}

\subsection{Linear Response Current\label{Linear Response Current}}

Recall that $\left( \Omega ,\mathfrak{A}_{\Omega }\right) $ is the
measurable space defined in Section \ref{Section impurities}, $\mathfrak{h}%
\doteq \ell ^{2}(\mathbb{Z}^{d};\mathbb{C})$ is the one-particle Hilbert
space with scalar product $\langle \cdot ,\cdot \rangle _{\mathfrak{h}}$ and
canonical orthonormal basis denoted by $\left\{ \mathfrak{e}_{x}\right\}
_{x\in \mathbb{Z}^{d}}$, and the one-particle Hamiltonian of the quasi-free
fermion system equals (\ref{eq:Ham_lap_pot}), i.e.,
\begin{equation*}
h^{(\omega )}\doteq \Delta _{\omega ,\vartheta }+\lambda \omega _{1}\ ,\text{%
\qquad }\omega =\left( \omega _{1},\omega _{2}\right) \in \Omega ,\ \lambda
,\vartheta \in \mathbb{R}_{0}^{+},
\end{equation*}%
with $\Delta _{\omega ,\vartheta }$ being (up to a minus sign) the random
discrete Laplacian. See again Section \ref{Section impurities}. The
associate (quasi-) free dynamics is thus defined from the (random) unitary
group $\{\mathrm{e}^{ith^{(\omega )}}\}_{t\in \mathbb{R}}$.

Then, apply on the fermion system an electromagnetic field resulting%
\footnote{%
We use the Weyl gauge, also named temporal gauge.} from a compactly
supported time--depen%
\-%
dent space-rescaled vector potential $\eta \mathbf{A}_{L}$ defined by
\begin{equation}
\eta \mathbf{A}_{L}(t,x)\doteq \eta \mathbf{A}(t,L^{-1}x),\quad t\in \mathbb{%
R},\,x\in \mathbb{R}^{d},\,\eta \in \mathbb{R}_{0}^{+},  \label{potential}
\end{equation}%
where%
\begin{equation*}
\mathbf{A}\in \mathbf{C}_{0}^{\infty }\doteq \underset{l\in \mathbb{R}^{+}}{%
\mathop{\displaystyle \bigcup}}C_{0}^{\infty }(\mathbb{R}\times \left[ -l,l%
\right] ^{d};({\mathbb{R}}^{d})^{\ast }).
\end{equation*}%
Here, $({\mathbb{R}}^{d})^{\ast }$ is the set of one-forms\footnote{%
In a strict sense, one should take the dual space of the tangent spaces $T({%
\mathbb{R}}^{d})_{x}$, $x\in {\mathbb{R}}^{d}$.} on ${\mathbb{R}}^{d}$ that
take values in $\mathbb{R}$. We see any $\mathbf{A}\in C_{0}^{\infty }(%
\mathbb{R}\times \left[ -l,l\right] ^{d};({\mathbb{R}}^{d})^{\ast
})\subseteq \mathbf{C}_{0}^{\infty }$, $l\in \mathbb{R}^{+}$, as a function $%
\mathbb{R}\times {\mathbb{R}}^{d}\rightarrow ({\mathbb{R}}^{d})^{\ast }$ via
the convention $\mathbf{A}(t,x)\equiv 0$ for $x\notin \lbrack -l,l]^{d}$.
The main reason for not using (the standard choice) $C_{0}^{\infty }(\mathbb{%
R}\times {\mathbb{R}}^{d};({\mathbb{R}}^{d})^{\ast })$ instead of $\mathbf{C}%
_{0}^{\infty }$ as a space of vector potentials, is that we need to include
(in general non-smooth) functions that are constant for $x$ inside cubes $%
[-l,l]^{d}$ and vanish outside. The time derivative of this vector potential
is the (time-dependent) electric field. Since we are interested here in the
linear response current to electromagnetic fields, we use in (\ref{potential}%
) a real parameter $\eta \in \mathbb{R}_{0}^{+}$ to also rescale the
strength of the vector potential $\mathbf{A}_{L}$.

To simplify notation, we consider, without loss of generality, spinless
fermions with negative charge. So, such an electromagnetic field leads to a
time-dependent Hamiltonian defined by
\begin{equation*}
\Delta _{\omega ,\vartheta }^{(\eta \mathbf{A}_{L})}+\lambda \omega
_{1},\quad t\in \mathbb{R},
\end{equation*}%
\ where $\Delta _{\omega ,\vartheta }^{(\mathbf{A})}\equiv \Delta _{\omega
,\vartheta }^{(\mathbf{A}(t,\cdot ))}\in \mathcal{B}(\ell ^{2}(\mathfrak{L}%
)) $ is the time-dependent self-adjoint operator defined\footnote{%
Observe that the sign of the coupling between $\mathbf{A}\in \mathbf{C}%
_{0}^{\infty }$ and the laplacian is wrong in \cite[Eq. (2.8)]{OhmI} for
negatively charged fermions.} by%
\begin{equation}
\langle \mathfrak{e}_{x},\Delta _{\omega ,\vartheta }^{(\mathbf{A})}%
\mathfrak{e}_{y}\rangle _{\mathfrak{h}}=\exp \left( i\int\nolimits_{0}^{1}%
\left[ \mathbf{A}(t,\alpha y+(1-\alpha )x)\right] (y-x)\mathrm{d}\alpha
\right) \langle \mathfrak{e}_{x},\Delta _{\omega ,\vartheta }\mathfrak{e}%
_{y}\rangle _{\mathfrak{h}}  \label{eq discrete lapla A}
\end{equation}%
for $\mathbf{A}\in \mathbf{C}_{0}^{\infty }$, $t\in \mathbb{R}$ and $x,y\in
\mathbb{Z}^{d}$. It is (up to a minus sign) the magnetic Laplacian, as
explained in \cite[Section III, in particular Corollary 3.1]{Ne}. This
yields a dynamics, perturbed by the time-dependent vector potential $\eta
\mathbf{A}_{L}$, given by the (well-defined random) two-parameter family $\{%
\mathrm{U}_{t,t_{0}}^{(\omega )}\}_{t_{0},t\in {\mathbb{R}}}$ of unitary
operators on $\mathfrak{h}$ satisfying the non-autonomous evolution equation
\begin{equation}
\forall t_{0},t\in {\mathbb{R}}:\text{\qquad }\partial _{t}\mathrm{U}%
_{t,t_{0}}^{(\omega )}=-i(\Delta _{\omega ,\vartheta }^{(\eta \mathbf{A}%
_{L})}+\lambda \omega _{1})\mathrm{U}_{t,t_{0}}^{(\omega )},\quad \mathrm{U}%
_{t_{0},t_{0}}^{(\omega )}\doteq \mathbf{1}_{\mathfrak{h}}.
\label{rescaledbis300}
\end{equation}%
In the algebraic formulation, it corresponds to the quasi-free dynamics on
the CAR $C^{\ast }$-algebra $\mathcal{U}$, defined by the unique
two-parameter group $\{\xi _{t,t_{0}}^{(\omega )}\}_{t_{0},t\in {\mathbb{R}}%
} $ of (Bogoliubov) $\ast $-automorphisms satisfying
\begin{equation}
\xi _{t,t_{0}}^{(\omega )}(a(\psi ))=a((\mathrm{U}_{t,t_{0}}^{(\omega
)})^{^{\ast }}\psi ),\text{\qquad }t_{0},t\in {\mathbb{R}},\ \psi \in
\mathfrak{h}.  \label{rescaledbis3}
\end{equation}%
The above procedure for coupling charged lattice fermions to a vector
potential is sometimes called \textquotedblleft Peierls
coupling\textquotedblright .

Additionally to the paramagnetic current observable $I_{(x,y)}^{(\omega )}$ (%
\ref{current observable}), the perturbing vector potential $\mathbf{A}\in
\mathbf{C}_{0}^{\infty }$ yields a second type of current observable, defined%
\footnote{%
Observe that the sign in the exponent in \cite[Eq. (50)]{OhmV} and \cite[%
(4.2)]{OhmVI} for negatively charged fermions is wrong, with no consequence
on the corresponding results.} by
\begin{equation}
\tilde{I}_{(x,y)}^{(\omega )}\doteq -2\Im \mathrm{m}\left( \left( \mathrm{e}%
^{i\int_{0}^{1}[\mathbf{A}(t,\alpha y+(1-\alpha )x)](y-x)d\alpha }-1\right)
\langle \mathfrak{e}_{x},\Delta _{\omega ,\vartheta }\mathfrak{e}_{y}\rangle
_{\mathfrak{h}}a(\mathfrak{e}_{x})^{\ast }a(\mathfrak{e}_{y})\right)
\label{diamagnetic current}
\end{equation}%
for any $\omega \in \Omega $, $\vartheta \in \mathbb{R}_{0}^{+}$, $t\in
\mathbb{R}$ and $x,y\in \mathbb{Z}^{d}$, where we recall that $\Im \mathrm{m}%
(A)\in \mathcal{U}$ is the imaginary part of $A\in \mathcal{U}$, see (\ref%
{im and real part}). We name it \emph{diamagnetic} current observable. The
derivation of the paramagnetic and diamagnetic current observables is
explained in detail in\ Appendix \ref{One-particle formulation2}. The
decomposition of the full current observable%
\begin{equation}
\tilde{I}_{(x,y)}^{(\omega )}+I_{(x,y)}^{(\omega )}=-2\Im \mathrm{m}\left(
\langle \mathfrak{e}_{x},\Delta _{\omega ,\vartheta }^{(\eta \mathbf{A}_{L})}%
\mathfrak{e}_{y}\rangle _{\mathfrak{h}}a(\mathfrak{e}_{x})^{\ast }a(%
\mathfrak{e}_{y})\right) \doteq \mathbf{I}_{(x,y)}^{(\omega ,\mathbf{A})}
\label{full current}
\end{equation}%
in so-called paramagnetic and diamagnetic current observables has a physical
relevance. First, it comes from the physics literature, see, e.g., \cite[Eq.
(A2.14)]{dia-current}. Secondly, the paramagnetic current observable is
intrinsic to the system and related to a heat production, whereas the
diamagnetic one is only non--vanishing in presence of vector potentials and
refers to the ballistic accelerations, induced by electromagnetic fields, of
charged particles. For more details, see \cite{OhmII,OhmIII}.

Observe that the time evolution of the KMS state $\varrho ^{(\omega )}\in
\mathcal{U}^{\ast }$ (see (\ref{stationary})-(\ref{2-point correlation
function})) is given by $\varrho ^{(\omega )}\circ \xi _{t,t_{0}}^{(\omega
)} $ for $t,t_{0}\in \mathbb{R}$. In \cite{OhmI,OhmII,OhmIII,OhmIV}\footnote{%
In all our papers we use smooth electric fields, but the extension to the
continuous case is straightforward.} we perform a detailed study the
behavior of current densities when $\eta \rightarrow 0$, \emph{uniformly}
w.r.t. the volume $\mathcal{O}\left( L^{d}\right) $ of the boxes where the
vector potential $\mathbf{A}_{L}$ is non-zero. In \cite%
{OhmV,OhmVI,brupedraLR}, these results are generalized to lattice-fermion
systems in disordered media with \emph{very general interactions}\footnote{%
Sufficiently strong polynomial decays of interactions are necessary. This
includes basically standard models of physics that describes interacting
fermions in crystal.} and on \emph{passive} states (not necessarily KMS).
These mathematically rigorous studies yield an alternative physical picture
of Ohm and Joule's laws (at least in the AC-regime), different from usual
explanations coming from the Drude model or the Landau theory of Fermi
liquids.

To shortly present how the linear response current naturally appears,
without requiring a thorough reading of this series of papers, consider a
space homogeneous electric fields in the box $\Lambda _{L}$ (\ref{eq:boxesl1}%
) for any $L\in \mathbb{R}^{+}$. To be more precise, let $\mathcal{A}\in
C_{0}^{\infty }(\mathbb{R};\mathbb{R}^{d})$ and set $\mathcal{E}(t)\doteq
-\partial _{t}\mathcal{A}(t)$ for all $t\in \mathbb{R}$. Therefore, $\mathbf{%
A}$ is defined to be the vector potential such that the electric field is
given by $\mathcal{E}(t)\in C_{0}^{\infty }(\mathbb{R};\mathbb{R}^{d})$ at
time $t\in \mathbb{R}$, for all $x\in \lbrack -1,1]^{d}$, and $(0,0,\ldots
,0)$ for $t\in \mathbb{R}$ and $x\notin \lbrack -1,1]^{d}$. It yields a
rescaled vector potential $\eta \mathbf{A}_{L}$ for $L\in \mathbb{R}^{+}$
and $\eta \in \mathbb{R}_{0}^{+}$.

Then, by (\ref{current observable}) and (\ref{diamagnetic current}), the
space-averaged response current observable, or response current density
observable, in the box $\Lambda _{L}$ and in the direction $\overrightarrow{w%
}=(w_{1},\ldots ,w_{d})\in \mathbb{R}^{d}$ ($|\overrightarrow{w}|=1$), for
any $\omega \in \Omega $, $\lambda ,\vartheta ,\eta \in \mathbb{R}_{0}^{+},$
$L\in \mathbb{R}^{+}$, $\mathbf{A}\in \mathbf{C}_{0}^{\infty }$ and $%
t_{0},t\in \mathbb{R}\in \mathbb{R}$, is, by definition, equal to
\begin{equation}
\mathbb{J}_{L}^{(\omega )}\left( t,\eta \right) \doteq \frac{1}{\left\vert
\Lambda _{L}\right\vert }\underset{k=1}{\sum^{d}}w_{k}\underset{x\in \Lambda
_{L}}{\sum }\left( \xi _{t,t_{0}}^{(\omega )}\left( I_{\left(
x+e_{k},x\right) }^{(\omega )}+\tilde{I}_{(x+e_{k},x)}^{(\omega ,\eta
\mathbf{A}_{L})}\right) -I_{\left( x+e_{k},x\right) }^{(\omega )}\right)
\label{linear response current0}
\end{equation}%
with $\{e_{k}\}_{k=1}^{d}$ being the canonical orthonormal basis of the
Euclidian space $\mathbb{R}^{d}$.

By using the generalization done in \cite{brupedraLR} of the celebrated
Lieb-Robin%
\-%
son bounds (for commutators) to multi-com%
\-%
mutators, the full current density observable in the direction $%
\overrightarrow{w}\in \mathbb{R}^{d}$ ($|\overrightarrow{w}|=1$) satisfies
\begin{equation}
\mathbb{J}_{L}^{(\omega )}\left( t,\eta \right) =\eta \mathbf{J}%
_{L}^{(\omega )}(t)+\mathcal{O}\left( \eta ^{2}\right)
\label{linear response current1}
\end{equation}%
in the CAR $C^{\ast }$-algebra $\mathcal{U}$. The correction terms of order $%
\mathcal{O}(\eta ^{2})$ are \emph{uniformly bounded} in $L\in \mathbb{R}^{+}$%
, $\omega \in \Omega $, $\lambda ,t\in \mathbb{R}_{0}^{+}$ and $\vartheta $
on compacta. By explicit computations, one checks that the linear part is%
\begin{equation}
\mathbf{J}_{L}^{(\omega )}(t)=\underset{k,q=1}{\sum^{d}}w_{k}\int_{-\infty
}^{t}\left\{ \mathcal{E}\left( \alpha \right) \right\} _{q}\left\{ \mathcal{C%
}_{\Lambda _{L}}^{(\omega )}\left( t-\alpha \right) \right\} _{k,q}\mathrm{d}%
\alpha ,  \label{linear response current1bis}
\end{equation}%
which is equal to $\mathbb{I}_{\Lambda _{L}}^{(\omega ,\mathcal{E}_{t})}$\ (%
\ref{current}) for the electric field defined by (\ref{new field}). See also
(\ref{defininion para coeff observable}) for the definition of $\mathcal{C}%
_{\Lambda }^{(\omega )}\in C^{1}(\mathbb{R};\mathcal{B}(\mathbb{R}^{d};%
\mathcal{U}^{d}))$. This current density observable is therefore the \emph{%
space-averaged linear response current observable} (or linear response
current density observable) in the direction $\overrightarrow{w}\in \mathbb{R%
}^{d}$ we study in all the paper. Because of (\ref{linear response
current1bis}), $\mathcal{C}_{\Lambda _{L}}^{(\omega )}$ is called the
conductivity observable matrix associated with $\Lambda _{L}$. For more
details, see also \cite[Theorem 3.7]{OhmV}.

In \cite{OhmIII,OhmIV,OhmVI,brupedraLR}, for any time $t\in \mathbb{R}$, we
prove the existence of the limit $L\rightarrow \infty $ of the random linear
response current density
\begin{equation*}
\varrho ^{(\omega )}\left( \mathbf{J}_{L}^{(\omega )}(t)\right) ,\qquad L\in
\mathbb{R}^{+},
\end{equation*}%
to a deterministic value, with probability one. At time $t=0$ this refers to
the following assertion:
\begin{equation}
x^{(\mathcal{E})}=\lim_{L\rightarrow \infty }\varrho ^{(\omega )}\left(
\mathbf{J}_{L}^{(\omega )}(0)\right) ,  \label{limit current}
\end{equation}%
which is directly related with (\ref{def currents}) and (\ref{current utiles}%
) at $s=0$.

\subsection{Discrete Continuity Equation in the CAR\ Algebra\label%
{One-particle formulation2}}

As is usual, the self-adjoint element%
\begin{equation*}
a(\mathfrak{e}_{x})^{\ast }a(\mathfrak{e}_{x})\in \mathcal{U}
\end{equation*}%
represents the particle number observable at the lattice site $x\in \mathbb{Z%
}^{d}$. Fixing once for all $\omega \in \Omega $, $\lambda ,\vartheta ,\eta
\in \mathbb{R}_{0}^{+},$ $L\in \mathbb{R}^{+}$, $\mathbf{A}\in \mathbf{C}%
_{0}^{\infty }$, its time-evolution by the two-parameter group $\{\xi
_{t,t_{0}}^{(\omega )}\}_{t_{0},t\in {\mathbb{R}}}$ of (Bogoliubov) $\ast $%
-automorphisms defined by (\ref{rescaledbis3}) equals
\begin{equation}
\xi _{t,t_{0}}^{(\omega )}\left( a\left( \mathfrak{e}_{x}\right) ^{\ast
}a\left( \mathfrak{e}_{x}\right) \right) =a((\mathrm{U}_{t,t_{0}}^{(\omega
)})^{^{\ast }}\mathfrak{e}_{x})^{\ast }a((\mathrm{U}_{t,t_{0}}^{(\omega
)})^{^{\ast }}\mathfrak{e}_{x})  \label{rescaledbis3rescaledbis3}
\end{equation}%
for any $t_{0},t\in \mathbb{R}$ and $x\in \mathbb{Z}^{d}$. Observe that $(%
\mathrm{U}_{t,t_{0}}^{(\omega )})^{^{\ast }}=\mathrm{U}_{t_{0},t}^{(\omega
)} $ for any $t_{0},t\in \mathbb{R}$ while
\begin{equation}
\forall t_{0},t\in {\mathbb{R}}:\text{\qquad }\partial _{t_{0}}\mathrm{U}%
_{t,t_{0}}^{(\omega )}=i\mathrm{U}_{t,t_{0}}^{(\omega )}(\Delta _{\omega
,\vartheta }^{(\eta \mathbf{A}_{L})}+\lambda \omega _{1}),\quad \mathrm{U}%
_{t_{0},t_{0}}^{(\omega )}\doteq \mathbf{1}_{\mathfrak{h}}\ .
\label{rererzfsdfsdfg}
\end{equation}%
From standard properties of the so-called fermionic creation/annihilation
operators, the time derivative of (\ref{rescaledbis3rescaledbis3}) equals%
\begin{equation*}
\partial _{t}\left( \xi _{t,t_{0}}^{(\omega )}\left( a\left( \mathfrak{e}%
_{x}\right) ^{\ast }a\left( \mathfrak{e}_{x}\right) \right) \right) =\xi
_{t,t_{0}}^{(\omega )}\left( \left( a(i(\Delta _{\omega ,\vartheta }^{(\eta
\mathbf{A}_{L})}+\lambda \omega _{1})\mathfrak{e}_{x})^{\ast }a(\mathfrak{e}%
_{x})+a(\mathfrak{e}_{x})^{\ast }a(i(\Delta _{\omega ,\vartheta }^{(\eta
\mathbf{A}_{L})}+\lambda \omega _{1})\mathfrak{e}_{x})\right) \right) .
\end{equation*}%
Recall now that the map $\psi \mapsto a(\psi )^{\ast }$ from $\mathfrak{h}$
to $\mathcal{U}$ is linear and, by (\ref{equation sup}) and (\ref{eq
discrete lapla A}), for any $x\in \mathbb{Z}^{d}$,%
\begin{equation*}
(\Delta _{\omega ,\vartheta }^{(\eta \mathbf{A}_{L})}+\lambda \omega _{1})%
\mathfrak{e}_{x}=\lambda \omega _{1}\left( x\right) \mathfrak{e}%
_{x}+\sum\limits_{z\in \mathbb{Z}^{d},|z|=1}\langle \mathfrak{e}%
_{x+z},\Delta _{\omega ,\vartheta }^{(\eta \mathbf{A}_{L})}\mathfrak{e}%
_{x}\rangle _{\mathfrak{h}}\mathfrak{e}_{x+z}.
\end{equation*}%
It follows that%
\begin{equation}
\partial _{t}\left( \xi _{t,t_{0}}^{(\omega )}\left( a\left( \mathfrak{e}%
_{x}\right) ^{\ast }a\left( \mathfrak{e}_{x}\right) \right) \right)
=\sum\limits_{z\in \mathbb{Z}^{d},|z|=1}\xi _{t,t_{0}}^{(\omega )}\left(
-2\Im \mathrm{m}\left( \langle \mathfrak{e}_{x},\Delta _{\omega ,\vartheta
}^{(\eta \mathbf{A}_{L})}\mathfrak{e}_{x+z}\rangle _{\mathfrak{h}}a(%
\mathfrak{e}_{x})^{\ast }a(\mathfrak{e}_{x+z})\right) \right)
\label{discrete}
\end{equation}%
for any $t_{0},t\in \mathbb{R}$ and $x\in \mathbb{Z}^{d}$. Another way to
prove this equation is to use \cite[Theorem 2.1 (ii) with $\Psi ^{\mathrm{IP}%
}=0$]{OhmV} together with straightforward computations using the CAR (\ref%
{eq:CAR}). Proceeding in this manner, observe that the quasi-free property
of the dynamics is not needed at all. In particular this derivation easily
extends to the interacting case. It is not so for the one-particle picture
discussed in the next section, which is much more restrictive than the
algebraic approach.

Equation (\ref{discrete}) is interpreted as a \emph{discrete continuity
equation}
\begin{equation*}
\partial _{t}\left( \xi _{t,t_{0}}^{(\omega )}(a(\mathfrak{e}_{x})^{\ast }a(%
\mathfrak{e}_{x})\right) =\sum\limits_{z\in \mathbb{Z}^{d},|z|=1}\xi
_{t,t_{0}}^{(\omega )}\left( \mathbf{I}_{(x,x+z)}^{(\omega ,\eta \mathbf{A}%
_{L})}\right)
\end{equation*}%
in the CAR $C^{\ast }$-algebra $\mathcal{U}$. The observable $\mathbf{I}%
_{(x,y)}^{(\omega ,\mathbf{A})}$ defined by (\ref{full current}) is the
observable related to the flow of particles from the lattice site $x$ to the
lattice site $y$ or the current from $y$ to $x$ for negatively charged
particles. [Positively charged particles can of course be treated in the
same way.] In the non-interacting case, this definition of current
observable is mathematically equivalent to the usual one in the one-particle
picture, like in \cite{[SBB98],jfa,klm}. See Equation (\ref{equation
refereee}).

\subsection{The One-particle Picture\label{One-particle formulation}}

When dealing with non-interacting fermions, most of the time, the
one-particle picture of such a physical system is employed, as for instance
in \cite{klm}. This is frequently technically convenient. Indeed, note that
various important estimates in the current study were obtained in this
picture and even all the analysis performed here could have been done in the
one-particle Hilbert space $\mathfrak{h}$. However, in many cases, this
preference is only subjective and motivated by the fact that, by some
reason, people feels more comfortable in dealing with Hilbert spaces than
with $C^{\ast }$-algebras. We stress that the algebraic formulation is, from
a conceptual point of view, the natural one, as the underlying physical
system is many-body. Moreover, it has some advantageous technical aspects,
both specific (like the possibility of using Bogoliubov-type inequalities in
important estimates) and general ones (like the very powerful theory of KMS
states). For convenience of those preferring the one-particle picture of
free fermion systems, we establish in the following the precise relation of
the \textquotedblleft second quantized\textquotedblright\ objects we used
here with this picture.

As in the previous subsection, fix once for all $\omega \in \Omega $, $%
\lambda ,\vartheta ,\eta \in \mathbb{R}_{0}^{+},$ $L\in \mathbb{R}^{+}$, $%
\mathbf{A}\in \mathbf{C}_{0}^{\infty }$. Recall that the corresponding KMS
state $\varrho ^{(\omega )}$ is the gauge-invariant quasi-free state
satisfying (\ref{2-point correlation function}), i.e.,
\begin{equation}
\varrho ^{(\omega )}(a^{\ast }\left( \varphi \right) a\left( \psi \right)
)=\left\langle \psi ,\mathbf{d}^{(\omega )}\varphi \right\rangle _{\mathfrak{%
h}},\qquad \varphi ,\psi \in \mathfrak{h},  \label{def1}
\end{equation}%
where%
\begin{equation*}
\mathbf{d}^{(\omega )}\doteq (1+\mathrm{e}^{\beta h^{(\omega )}})^{-1}\in
\mathcal{B}\left( \mathfrak{h}\right)
\end{equation*}%
and the one-particle Hamiltonian $h^{(\omega )}=(h^{(\omega )})^{\ast }\in
\mathcal{B}\left( \mathfrak{h}\right) $ is defined by (\ref{eq:Ham_lap_pot}%
). The positive bounded operator $\mathbf{d}^{(\omega )}$ satisfies $0\leq
\mathbf{d}^{(\omega )}\leq \mathbf{1}_{\mathfrak{h}}$ and is called the
symbol, or one-particle density matrix, of the quasi-free state $\varrho
^{(\omega )}$. See (\ref{ass O0-00})-(\ref{ass O0-00bis}) for the definition
of gauge-invariant quasi-free states.

The time-evolution $\varrho ^{(\omega )}$ by the two-parameter group $\{\xi
_{t,t_{0}}^{(\omega )}\}_{t_{0},t\in {\mathbb{R}}}$ of (Bogoliubov) $\ast $%
-automorphisms defined by (\ref{rescaledbis3}) is $\varrho ^{(\omega )}\circ
\xi _{t,t_{0}}^{(\omega )}$ for any $t_{0},t\in {\mathbb{R}}$. It is again a
gauge-invariant quasi-free state and satisfies
\begin{equation}
\varrho ^{(\omega )}\circ \xi _{t,t_{0}}^{(\omega )}\left( a^{\ast }\left(
\varphi \right) a\left( \psi \right) \right) =\left\langle \psi ,\mathrm{U}%
_{t,t_{0}}^{(\omega )}\mathbf{d}^{(\omega )}(\mathrm{U}_{t,t_{0}}^{(\omega
)})^{^{\ast }}\varphi \right\rangle _{\mathfrak{h}},\qquad \varphi ,\psi \in
\mathfrak{h},  \label{def1bis}
\end{equation}%
for any $t_{0},t\in {\mathbb{R}}$, by (\ref{rescaledbis3}) and (\ref{def1}).
Again,
\begin{equation}
\mathbf{d}_{t,t_{0}}^{(\omega )}\doteq \mathrm{U}_{t,t_{0}}^{(\omega )}(1+%
\mathrm{e}^{\beta h^{(\omega )}})^{-1}(\mathrm{U}_{t,t_{0}}^{(\omega
)})^{^{\ast }}\in \mathcal{B}\left( \mathfrak{h}\right)  \label{def2}
\end{equation}%
is a positive bounded operator $\mathbf{d}_{t,t_{0}}^{(\omega )}$ satisfying
$0\leq \mathbf{d}_{t,t_{0}}^{(\omega )}\leq \mathbf{1}_{\mathfrak{h}}$. It
is the symbol, or one-particle density matrix, of the quasi-free state $%
\varrho ^{(\omega )}\circ \xi _{t,t_{0}}^{(\omega )}$. Recall that the
unitary operators $\mathrm{U}_{t,t_{0}}^{(\omega )}\in \mathcal{B}\left(
\mathfrak{h}\right) $, $t_{0},t\in {\mathbb{R}}$, are uniquely defined by (%
\ref{rererzfsdfsdfg}).

By (\ref{rescaledbis300}), (\ref{rererzfsdfsdfg}) and (\ref{def2}) together
with $(\mathrm{U}_{t,t_{0}}^{(\omega )})^{^{\ast }}=\mathrm{U}%
_{t_{0},t}^{(\omega )}$, the symbol $\mathbf{d}_{t,t_{0}}^{(\omega )}$ is
the solution of the \emph{Liouville equation}:
\begin{equation}
\forall t_{0},t\in {\mathbb{R}}:\text{\qquad }i\partial _{t}\mathbf{d}%
_{t,t_{0}}^{(\omega )}=\left[ \left( \Delta _{\omega ,\vartheta }^{(\eta
\mathbf{A}_{L})}+\lambda \omega _{1}\right) ,\mathbf{d}_{t,t_{0}}^{(\omega )}%
\right] ,\text{\qquad }\mathbf{d}_{t_{0},t_{0}}^{(\omega )}\doteq \mathbf{d}%
^{(\omega )},  \label{liouville}
\end{equation}%
as for instance in \cite[Eq. (2.5)]{klm}. Then, all the study performed in
the current paper for second quantized currents of non-interacting fermions
can be translated into the one-particle picture by using the Liouville
equation and the fact that the corresponding quasi-free states are
completely determined by the one-particle density matrices $\{\mathbf{d}%
_{t,t_{0}}^{(\omega )}\}_{t_{0},t\in {\mathbb{R}}}$, solving the above
initial value problem.

In this framework, the current observable discussed in Section \ref{Linear
Response Current}, and studied along the paper, can be represented by
self-adjoint operators on the one-particle Hilbert space $\mathfrak{h}$.
See, e.g., (\ref{dekldjfskldjf}). In this perspective, note that the full
current density observable in a box $\Lambda _{L}$ in a fixed direction $%
e_{k}$, $k\in \{1,\ldots ,d\}$, in $\mathbb{R}^{d}$ is the so-called second
quantization of the operator defined by%
\begin{equation}
\mathfrak{I}_{L}^{(\omega )}\doteq -\frac{2}{\left\vert \Lambda
_{L}\right\vert }\underset{x\in \Lambda _{L}}{\sum }\Im \mathrm{m}\{\langle
\mathfrak{e}_{x+e_{k}},\Delta _{\omega ,\vartheta }^{(\eta \mathbf{A}_{L})}%
\mathfrak{e}_{x}\rangle _{\mathfrak{h}}P_{\left\{ x+e_{k}\right\}
}s_{e_{k}}P_{\left\{ x\right\} }\},\qquad L\in \mathbb{R}^{+},
\label{definition}
\end{equation}%
using the notation (\ref{shift}) for shift operators. See also (\ref{shift2}%
). In other words, by Definition \ref{def bilineqr},%
\begin{equation*}
\frac{1}{\left\vert \Lambda _{L}\right\vert }\underset{x\in \Lambda _{L}}{%
\sum }\mathbf{I}_{(x+e_{k},x)}^{(\omega ,\eta \mathbf{A}_{L})}=\langle
\mathrm{A},\mathfrak{I}_{L}^{(\omega )}\mathrm{A}\rangle \ .
\end{equation*}%
The one-particle operator $\mathfrak{I}_{L}^{(\omega )}$ is directly related
with the commonly used current observable in the one-particle Hilbert space,
like in \cite{[SBB98],jfa,klm}. To see this, for $k\in \{1,\ldots ,d\}$,
define the (unbounded) multiplication operator on $\mathfrak{h}$ with the $%
k^{th}$ component by
\begin{equation*}
X_{k}(\psi )(x_{1},\ldots ,x_{d})\doteq x_{k}\psi (x_{1},\ldots ,x_{d}),
\end{equation*}%
for $\psi $ within the domain of $X_{k}$. For any $x\in \mathbb{Z}^{d}$,
remark that%
\begin{equation*}
\Delta _{\omega ,\vartheta }^{(\eta \mathbf{A}_{L})}\mathfrak{e}%
_{x}=\sum\limits_{z\in \mathbb{Z}^{d},|z|=1}\langle \mathfrak{e}%
_{x+z},\Delta _{\omega ,\vartheta }^{(\eta \mathbf{A}_{L})}\mathfrak{e}%
_{x}\rangle _{\mathfrak{h}}\mathfrak{e}_{x+z}
\end{equation*}%
and
\begin{equation*}
-i\left[ \Delta _{\omega ,\vartheta }^{(\eta \mathbf{A}_{L})},X_{k}\right]
\mathfrak{e}_{x}=i\left( \langle \mathfrak{e}_{x+e_{k}},\Delta _{\omega
,\vartheta }^{(\eta \mathbf{A}_{L})}\mathfrak{e}_{x}\rangle _{\mathfrak{h}}%
\mathfrak{e}_{x+e_{k}}-\langle \mathfrak{e}_{x-e_{k}},\Delta _{\omega
,\vartheta }^{(\eta \mathbf{A}_{L})}\mathfrak{e}_{x}\rangle _{\mathfrak{h}}%
\mathfrak{e}_{x-e_{k}}\right) .
\end{equation*}%
Combining this with (\ref{definition}), one checks that%
\begin{equation}
\mathfrak{I}_{L}^{(\omega )}=\frac{1}{\left\vert \Lambda _{L}\right\vert }%
\mathrm{P}_{L}\left( -i\left[ \Delta _{\omega ,\vartheta }^{(\eta \mathbf{A}%
_{L})}+\lambda \omega _{1},X_{k}\right] \right) \mathrm{P}_{L}+\mathcal{O}%
(L^{-1})\ ,\qquad L\in \mathbb{R}^{+},  \label{equation refereee}
\end{equation}%
uniformly in $\mathcal{U}$ w.r.t. all parameters, where $\mathrm{P}_{L}$ is
the orthogonal projection with range $\mathrm{lin}\left\{ \mathfrak{e}%
_{x}\colon x\in \Lambda _{L}\right\} $, that is, the multiplication operator
with the characteristic function of the box $\Lambda _{L}$. The term of
order $\mathcal{O}(L^{-1})$ results from the existence of $\mathcal{O}%
(L^{d-1})$ points $x\in \Lambda _{L}$ such that $x+e_{k}\notin \Lambda _{L}$.

We recover from (\ref{equation refereee}) the usual description for the
current observable as a self-adjoint operator on the one-particle Hilbert
space $\mathfrak{h}$, in our case the velocity operator $-i[\Delta _{\omega
,\vartheta }^{(\eta \mathbf{A}_{L})}+\lambda \omega _{1},X_{k}]$. See, e.g.,
\cite{[SBB98],jfa,klm}. Observe additionally that the quantity obtained by
applying the state $\varrho ^{(\omega )}\circ \xi _{t,t_{0}}^{(\omega )}$ on
the full current density observable gives, in the large volume limit (i.e., $%
L\rightarrow \infty $), the density of trace of the product of symbol $%
\mathbf{d}_{t,t_{0}}^{(\omega )}$ with the velocity operator on the
one-particle Hilbert space $\mathfrak{h}$, similar to \cite[Equation (2.6)]%
{klm}.\bigskip

\noindent \textbf{Acknowledgements: }This research is supported by CNPq
(308337/2017-4), FAPESP (2016/02503-8, 2017/22340-9), as well as by the
Basque Government through the grant IT641-13 and the BERC 2018-2022 program,
and by the Spanish Ministry of Economy and Competitiveness MINECO: BCAM
Severo Ochoa accreditation SEV-2017-0718, MTM2017-82160-C2-2-P. We are very
grateful to the BCAM and its management, which supported this project via
the visiting researcher program. Finally, we thank very much the reviewers
for the very careful work, contributing to improve the readability of the
paper.

\end{document}